\documentclass[a4paper, 11pt]{article}

\usepackage{amsmath,amsthm,latexsym,amssymb,wasysym}
\usepackage[margin=2cm]{geometry}
\usepackage{color}
\usepackage{mathtools}
\usepackage{authblk}
\usepackage{cancel}
\usepackage{bbm}
\usepackage{xargs}
\usepackage{enumerate}
\usepackage[textsize=footnotesize]{todonotes} % pour voir les commentaires
\usepackage[font=footnotesize]{caption}
\usepackage{multirow}

\usepackage{hyperref}

\let\counterwithin\relax
\usepackage{chngcntr}

\allowdisplaybreaks

\newtheorem{theorem}{Theorem}
\newtheorem{proposition}{Proposition}
\newtheorem{lemma}{Lemma}
\newtheorem{corollary}{Corollary}
\newtheorem{remark}{Remark}
\newtheorem{definition}{Definition}

\newtheorem{assumption}{Assumption}

\counterwithin{figure}{section}

\def\E{\mathbb{E}}

\def\calF{\mathcal{F}}
\def\calL{\mathcal{L}}

\def\Rd{\mathbb{R}^{d}}

\def\1{\mathds{1}}

\def\nset{\mathbb{N}}
\def\rset{\mathbb{R}}

\newcommand{\argmin}{\operatornamewithlimits{argmin}}

\newcommandx{\as}[1][1=P]{\ensuremath{#1\, -\mathrm{a.s.}}}

\newcommandx{\set}[2]{\{ {#1} \, ; \, {#2}\}}

\newcommandx{\CPE}[3][1=]{{\mathbb E}^{#1}\left[\left. #2 \, \right| #3 \right]} %%%% esperance conditionnelle

\newcommand{\floor}[1]{{\lfloor #1 \rfloor}}
\newcommand{\ceil}[1]{{\lceil #1 \rceil}}

\begin{document}

\title{A fully data-driven approach to minimizing CVaR for portfolio of assets via SGLD with discontinuous updating
\thanks{This work was supported by The Alan Turing Institute for Data Science and AI under EPSRC grant EP/N510129/1.  Y. Z. was supported by The Maxwell Institute Graduate School in Analysis and its Applications, a Centre for Doctoral Training funded by the UK Engineering and Physical Sciences Research Council (grant EP/L016508/01), the Scottish Funding Council, Heriot-Watt University and the University of Edinburgh. }}%We thank the Alan Turing Institute, London, UK; the R\'enyi Institute, Budapest, Hungary and the \'Ecole Polytechnique, Palaiseau, France for hosting research meetings of the authors.}}

\author[1,2]{Sotirios Sabanis}
\author[1]{Ying Zhang}

\affil[1]{\footnotesize School of Mathematics, The University of Edinburgh, UK.}
\affil[2]{\footnotesize The Alan Turing Institute, UK.}

\date{\today}

\maketitle
\begin{abstract}
A new approach in stochastic optimization via the use of stochastic gradient Langevin dynamics (SGLD) algorithms, which is a variant of stochastic gradient decent (SGD) methods, allows us to efficiently approximate global minimizers of possibly complicated, high-dimensional landscapes. With this in mind, we extend here the non-asymptotic analysis of SGLD to the case of discontinuous stochastic gradients. We are thus able to provide theoretical guarantees for the algorithm's convergence in (standard) Wasserstein distances for both convex and non-convex objective functions. We also provide explicit upper estimates of the expected excess risk associated with the approximation of global minimizers of these objective functions.

All these findings allow us to devise and present a fully data-driven approach for the optimal allocation of weights for the minimization of CVaR of portfolio of assets with complete theoretical guarantees for its performance. Numerical results illustrate our main findings.
\end{abstract}

\section{Introduction}
We are concerned in this article with the study of stochastic optimization problems of the form
\begin{equation}\label{main_opt_prob}
 \text{minimize} \quad U(\theta) : = \mathbb{E}[f(\theta, X)],
\end{equation}
where the gradient of $f$ is discontinuous in $\theta \in \mathbb{R}^d$ and $X$ is a random element with a smooth density. Within this framework, we highlight and solve the problem of minimizing CVaR (expected shortfall) of a portfolio of assets in terms of  optimal selection of weights for individual assets as explained in Section \ref{cvar_port}.  We offer theoretical guarantees for the approximate solution of the optimization problem \eqref{main_opt_prob} by generating a $\hat{\theta}$ such that the expected excess risk
\[
 \mathbb{E}[U(\hat{\theta})] - \inf_{\theta \in \mathbb{R}^d} U(\theta)
\]
is minimized. To achieve this, we analyse the convergence properties of the stochastic gradient Langevin dynamics (SGLD) algorithm with discontinuous updating $H$, which is given by
\begin{equation}\label{SGLD}
\theta^{\lambda}_0 =\theta_0,\quad \theta^{\lambda}_{n+1}=\theta^{\lambda}_n-\lambda H(\theta^{\lambda}_n,X_{n+1})+\sqrt{2\beta^{-1}\lambda}\xi_{n+1},\ n\in\mathbb{N},
\end{equation}
where $\theta_0$ is an $\mathbb{R}^d$-valued random variable, $\lambda >0$ is the stepsize, $\beta>0$ is the so-called inverse temperature parameter, $H:\mathbb{R}^d\times\mathbb{R}^m\to\mathbb{R}^d$ is a measurable function satisfying $\nabla U(\theta)=\mathbb{E}[H(\theta,X_0)]$ with $(X_n)_{n \in \mathbb{N}}$ being an i.i.d. sequence, and $(\xi_n)_{n\in\mathbb{N}}$ is an independent sequence of standard $d$-dimensional Gaussian random variables. One recalls heere that the SGLD algorithm \eqref{SGLD} can be viewed as a discretization of the Langevin SDE:
\begin{equation}\label{sde}
Z_0 = \theta_0, \quad d Z_t = -  h(Z_t) d t + \sqrt{2\beta^{-1}} d B_t,
\end{equation}
where $h \coloneqq \nabla U$ and $(B_t)_{t \geq 0}$ represents the standard Brownian motion. Moreover, it is well-known that, under appropriate conditions, the Langevin SDE \eqref{sde} admits a unique invariant measure $\pi_{\beta} \wasypropto \exp(-\beta U(\theta))$ which concentrates around the minimizers of $U$ when $\beta$ is sufficiently large, , see \cite{hwang} for more details.

Theoretical guarantees of the SGLD algorithm \eqref{SGLD} to the target distribution $\pi_{\beta}$ have been established in Wasserstein-2 distance under the assumptions that $H$ is convex and (locally) Lipschitz continuous, see \cite{convex}, \cite{ppbdm}, \cite{dk} and references therein. Recently, these results are considered under more generalised conditions aiming to include a wider range of practical applications. To relax the convexity condition, a dissipativity condition is proposed in \cite{raginsky}, and the convergence result is obtained in Wasserstein-2 distance with the rate $\lambda^{5/4}n$. This is the first such result in non-convex optimization, which is then improved in the work \cite{xu} and \cite{nonconvex}. Compared to \cite{raginsky}, a higher rate of convergence with dependence on $n$ is achieved in \cite{xu} following a direct analysis of the ergodicity of the overdamped Langevin Monte Carlo (LMC) algorithms, while a rate 1/2 in Wasserstein-1 distance is obtained in \cite{nonconvex} by using the contraction results developed in \cite{eberle}.

As for the generalisation of the smoothness of $H$, to the best of the author's knowledge, there are no theoretical guarantees established in the literature for the SGLD algorithm \eqref{SGLD} with discontinuous gradient. We present here the first such results. We are inspired by similar studies for stochastic gradient descent (SGD) algorithms, see \cite{sgddiscnt} and \cite{4} and references therein. In particular, \cite{sgddiscnt} provides an almost sure convergence result, while  \cite{4} provides a strong $L_1$ convergence result with rate 1/2.

In this paper, we establish non-asymptotic error bounds for the SGLD algorithm \eqref{SGLD} with discontinuous gradient $H$. More precisely, non-asymptotic results in Wasserstein-1 and Wasserstein-2 distances between the law of the $n$-th iterate of the SGLD algorithm \eqref{SGLD} and the target distribution $\pi_{\beta}$ are obtained under convexity and dissipativity conditions for $H$. This allows us to then provide full analytic results concerning the  expected excess risk of the associated optimization problem \eqref{main_opt_prob}. All this is achieved by assuming that $H$ is decomposed in to two parts  $F$ and $G$, where $F: \mathbb{R}^d \times \mathbb{R}^m \rightarrow \mathbb{R}^d$ is locally Lipschitz continuous and $G: \mathbb{R}^d \times \mathbb{R}^m \rightarrow \mathbb{R}^d$ is bounded. Furthermore, $H$ is assumed to satisfy a conditional Lipschitz-continuity (CLC) property proposed in \cite{4}, which is given explicitly in Assumption \ref{clc} below.

We illustrate the applicability of our findings by presenting examples from quantile and VaR, CVaR estimations in Section \ref{application}. In particular, we solve the problem of optimal allocation of weights for the minimization of CVaR of a portfolio of assets. This is also the first such result in the literature to the best of the author's knowledge. Numerical experiments are implemented and their results support our theoretical findings.

The paper is organised as follows. Section \ref{result} presents the assumptions and main results. In Section \ref{po}, the proofs for the main theorems in the non-convex case are provided, which are followed by the proofs for the results in the convex case in Section \ref{poconv}. Practical examples along with the minimization algorithm of CVaR for a portfolio of assets are presented in Section \ref{application} while auxiliary results are provided in Section \ref{app}.

We conclude this section by introducing some notation. Let $(\Omega,\mathcal{F},P)$ be a probability space. We denote by $\E[X]$  the expectation of a random variable $X$. For any $x \in \mathbb{R}^d$, denote by $x^{(i)}$ the $i$-th entry of the vector. %For $1\leq p<\infty$, $L^p$ is used to denote the usual space of $p$-integrable real-valued random variables.
Fix an integer $d\geq 1$. For an $\mathbb{R}^d$-valued random variable $X$, its law on $\mathcal{B}(\mathbb{R}^d)$ (the Borel sigma-algebra of $\mathbb{R}^d$) is denoted by $\mathcal{L}(X)$. Scalar product is denoted
by $\langle \cdot,\cdot\rangle$, with $|\cdot|$ standing for the
corresponding norm (where the dimension of the space may vary depending on the context). For $\mu\in\mathcal{P}(\mathbb{R}^d)$ and for a non-negative measurable $f:\mathbb{R}^d\to\mathbb{R}$, the notation $\mu(f):=\int_{\mathbb{R}^d} f(\theta)\mu(d \theta)$ is used. Given a Markov kernel $R$ on $\mathbb{R}^d$ and a function $f$ integrable under $R(x, \cdot)$, for any $x \in \mathbb{R}^d$, denote by $Rf(x) = \int_{\mathbb{R}^d}f(y)R(x,dy)$. For any integer $q \geq 1$, let $\mathcal{P}(\mathbb{R}^q)$ denote the set of probability measures on $\mathcal{B}(\mathbb{R}^q)$.
For $\mu,\nu\in\mathcal{P}(\mathbb{R}^d)$, let $\mathcal{C}(\mu,\nu)$ denote the set of probability measures $\zeta$
on $\mathcal{B}(\mathbb{R}^{2d})$ such that its respective marginals are $\mu,\nu$. For two probability measures $\mu$ and $\nu$, the Wasserstein distance of order $p \geq 1$ is defined as
\begin{equation}\label{eq:definition-W-p}
W_p(\mu,\nu):=\inf_{\zeta\in\mathcal{C}(\mu,\nu)}
\left(\int_{\rset^d}\int_{\rset^d}|\theta-\theta'|^p\zeta(d \theta d \theta')\right)^{1/p},\ \mu,\nu\in\mathcal{P}(\rset^d).
\end{equation}

\section{Main results}\label{result}
Denote by $\mathcal{G}_n:=\sigma(X_k,\ k\leq n,\ k\in\mathbb{N})$, for any $n \in\mathbb{N}$. $(X_n)_{n\in\mathbb{N}}$ is an $\mathbb{R}^m$-valued, $(\mathcal{G}_n)_{n\in\mathbb{N}}$-adapted process. It is assumed throughout the paper that $\theta_0$, $\mathcal{G}_{\infty}$ and $(\xi_{n})_{n\in\mathbb{N}}$ are independent. Moreover, the following assumptions are considered:

\begin{assumption}\label{expressionH}
Let $H: \mathbb{R}^d \times \mathbb{R}^m \rightarrow \mathbb{R}^d$ take the form
\[
H(\theta, x) = F(\theta, x) +   G(\theta, x), \quad \theta \in \mathbb{R}^d, \quad x \in \mathbb{R}^m,
\]
where $F: \mathbb{R}^d \times \mathbb{R}^m \rightarrow \mathbb{R}^d$ and $G: \mathbb{R}^d \times \mathbb{R}^m \rightarrow \mathbb{R}^d$ satisfy the following:
\begin{enumerate}[(i)]
\item $ F: \mathbb{R}^d \times \mathbb{R}^m \rightarrow \mathbb{R}^d$ is jointly Lipschitz continuous in both variables, i.e. there exist $L_1, L_2>0$, $\rho \geq 0$ such that for any $\theta, \theta' \in \mathbb{R}^d$, $x, x' \in \mathbb{R}^m$,
\[
| F(\theta,x) -  F(\theta',x')| \leq (1+|x|+|x'|)^{\rho}(L_1|\theta - \theta'| +L_2|x - x'|).
\]
\item $ G(\theta,x) : \mathbb{R}^d \times \mathbb{R}^m \rightarrow \mathbb{R}^d$ is bounded in $\theta$, i.e. there exist $K_1:\mathbb{R}^m \to \mathbb{R}_+$ such that for any $\theta \in \mathbb{R}^d$, $x  \in \mathbb{R}^m$,
\[
| G(\theta,x)| \leq K_1(x).
\]
\end{enumerate}
\end{assumption}

\begin{assumption}\label{iid}
We assume the inital value $\theta_0$ satisfies $\E[|\theta_0|^4]<\infty$. The process $(X_n)_{n \in \nset}$ is i.i.d. with $\E[|X_0|^{4\rho+4}]<\infty$ and $\E[K_1^4(X_0)]<\infty$. Moreover, it satisfies
\[
\E[H(\theta,X_0)]=h(\theta).
\]
\end{assumption}

\begin{remark}\label{growth}
By Assumption \ref{expressionH}, for all $\theta \in\mathbb{R}^d$ and $x \in \mathbb{R}^m$,
\[
|H(\theta, x)| \leq  (1+|x|)^{\rho+1}(L_1 |\theta| +L_2) +F_*(x) ,
\]
where $F_*(x) = | F(0,0)| +K_1(x)$. For any $x \in \mathbb{R}^m$, $\rho\geq 0$, denote by
\begin{equation}\label{krho}
K_{\rho}(x) = (1+2|x|)^{4\rho+4}.
\end{equation}
One notices that by Assumption \ref{iid}, $\mathbb{E}[K_{\rho}(X_0)]$ is well defined.
\end{remark}

\begin{assumption}\label{clc}
There exists a positive constant $L>0$ such that, for all $\theta, \theta'\in\mathbb{R}^d$,
\[
  \mathbb{E}[|H(\theta,X_0)- H(\theta',X_0)|]  \le L|\theta-\theta'|.
\]
\end{assumption}

\begin{remark}\label{hlip} Assumptions \ref{iid} and \ref{clc} imply, for all $\theta,\theta' \in \rset^{d}$,
\begin{equation}\label{mulyan}
| h(\theta)-h(\theta')| \leq  L |\theta-\theta'|.
\end{equation}
\end{remark}

\begin{remark}\label{clcex} Assumption \ref{clc} is satisfied for a wide class of $(X_n)_{n\in\mathbb{N}}$, see  Section \ref{application} for the examples. Here, for the illustrative purpose, one considers the following simple example. Suppose $ G(\theta,x) = \sum_{j=1}^N \dot{g}_j(\theta, x)\mathbbm{1}_{\bigcap_{i = 1}^m\{x^{(i)} \in I_{i,j}(\theta)\} }$ is a lower semi-continuous function, where $N \in \mathbb{N}^*$, $\dot{g}_j: \mathbb{R}^d \times \mathbb{R}^m \rightarrow \mathbb{R}^d$ are bounded and jointly Lipschitz continuous functions, i.e. there exist $L_3, L_4, K_2>0$ such that for any $\theta, \theta' \in \mathbb{R}^d$, $x, x' \in \mathbb{R}^m$, $j = 1, \dots, N$
\[
|\dot{g}_j(\theta,x) - \dot{g}_j(\theta',x')| \leq (1+|x|+|x'|)^{\rho}(L_3|\theta - \theta'| +L_4|x - x'|), \quad |\dot{g}_j(\theta, x)| \leq K_2,
\]
the intervals $ I_{i,j}(\theta)$ take the form $(-\infty, \bar{g}_j^{(i)}(\theta))$, $(\bar{g}_j^{(i)}(\theta), \infty)$ or $(\tilde{g}_j^{(i)}(\theta), \hat{g}_j^{(i)}(\theta))$, and $\bar{g}^{(i)}_j, \tilde{g}^{(i)}_j, \hat{g}^{(i)}_j: \mathbb{R}^d \rightarrow \mathbb{R}$ are Lipschitz continuous functions. In this case, it is enough to require the marginal density function of $X_0^{(i)}$ is continuous and bounded for any $i = 1, \dots, m$. Then, the property stated in Assumption \ref{clc} holds.
\end{remark}
\begin{proof} See Appendix \ref{proofclcex}.
\end{proof}

\subsection{Nonconvex case}
Further to the assumptions above, we consider the following conditions on $U$, which can be viewed as a generalization of the convexity assumption.
\begin{assumption}\label{assum:dissipativity}
There exist $A:\rset^m\to\rset^{d\times d}$, $b: \mathbb{R}^m \to \mathbb{R}$ such that for any $ x,y \in \rset^d$,
\begin{align*}
\langle y, A(x) y\rangle \geq 0
\end{align*}
and for all $\theta \in \mathbb{R}^d$ and $x\in\mathbb{R}^m$,
\[
\langle  F(\theta,x),\theta\rangle\geq \langle \theta, A(x) \theta\rangle -b(x).
\]
The smallest eigenvalue of $\E[A(X_0)]$ is a positive real number $a>0$ and $E[b(X_0)] = b>0$.
\end{assumption}
Define first
\begin{equation}
\label{eq:definition-lambda-max}
\lambda_{\max}= \min\left\{\frac{\min\{a,a^{1/3}\}}{24 (1+L_1)^2\mathbb{E}\left[K_{\rho}(X_0)\right]},\frac{1}{4a}\right\},
\end{equation}
where $L_1, a$ are given in Assumption \ref{expressionH} and \ref{assum:dissipativity} respectively, and $K_{\rho}(x)$ for any $x \in \mathbb{R}^m$ is defined in \eqref{krho}.
\begin{theorem}\label{main} Let Assumptions \ref{expressionH}, \ref{iid}, \ref{clc} and \ref{assum:dissipativity} hold. Then, for any $n \in \nset$, $0<\lambda\leq \lambda_{\max}$, there exist constants $C_0,C_1,C_2>0$ such that,
\begin{equation}\label{menyi}
W_1(\mathcal{L}(\theta^{\lambda}_n),\pi_{\beta})\leq C_1  e^{-C_0\lambda n}(\mathbb{E}[|\theta_0|^4]+1) +C_2\sqrt{\lambda},\ n\in\mathbb{N},
\end{equation}
where $C_0$, $C_1$ and $C_2$ are given explicitly in \eqref{mainthmconst}.
\end{theorem}

Theorem \ref{main} provides the rate of convergence between the law of the SGLD algorithm \eqref{SGLD} and the target distribution $\pi_{\beta}$ in $W_1$ distance. An analogous result in Wasserstein-2 distance can be obtained.

\begin{corollary}\label{cw2}  Let Assumptions \ref{expressionH}, \ref{iid}, \ref{clc} and \ref{assum:dissipativity} hold. Then, for any $n \in \mathbb{N}$, $0<\lambda\leq \lambda_{\max}$ given in \eqref{eq:definition-lambda-max}, there exist constants $C_3,C_4,C_5>0$ such that,
\begin{equation*}
W_2(\mathcal{L}(\theta^{\lambda}_n),\pi_{\beta})\leq C_4 e^{-C_3\lambda n}(\mathbb{E}[|\theta_0|^4]+1) +C_5\lambda^{1/4},\ n\in\mathbb{N},
\end{equation*}
where $C_3$, $C_4$ and $C_5$ are given explicitly in \eqref{cw2const}.
\end{corollary}

By using the convergence result in Wasserstein-2 distance as presented in Corollary \ref{cw2}, one can obtain an upper bound for the expected excess risk $ \mathbb{E}[U(\hat{\theta})] - \inf_{\theta \in \mathbb{R}^d} U(\theta)$.

\begin{corollary}\label{eer}  Let Assumptions \ref{expressionH}, \ref{iid}, \ref{clc} and \ref{assum:dissipativity} hold. %Moreover, let the inital value $\theta_0$ satisfy $\E\left[e^{|\theta_0|^2}\right]<\infty$.
Then, for every $0<\lambda\leq \lambda_{\max}$ given in \eqref{eq:definition-lambda-max}, there exist constants $\hat{C}_0, \hat{C}_1, \hat{C}_2, \hat{C}_3>0$ such that the expected excess risk can be estimated as
\begin{equation*}
 \mathbb{E}[U(\hat{\theta})] - \inf_{\theta \in \mathbb{R}^d} U(\theta) \leq \hat{C}_1e^{-\hat{C}_0\lambda n}+\hat{C_2}\lambda^{1/4} + \hat{C}_3/\beta,
\end{equation*}
where $\hat{\theta} =  \theta^{\lambda}_n$, and $\hat{C}_0, \hat{C}_1, \hat{C}_2$,  $\hat{C}_3>0$ are given explicitly in \eqref{eerconst} and \eqref{eerconst2}.
\end{corollary}

\subsection{Convex case}
Recall Assumption \ref{expressionH}, where it is assumed $H =  F+ G$. In this section, we present (improved) convergence results of the SGLD algorithm \eqref{SGLD} under the convexity condition of $F$ and $G$.

In the case that $F$ satisfies a convexity condition but not $G$, the result in Theorem \ref{main} can be recovered.% as the convexity condition implies the dissipativity condition.
\begin{assumption}\label{convex}
There exist $\hat{A}_1:\rset^m\to\rset^{d\times d}$ such that for any $ x,y \in \rset^d$,
\begin{align*}
\langle y, \hat{A}_1(x) y\rangle \geq 0
\end{align*}
and for each $\theta, \theta' \in \rset^d$, $x\in\mathbb{R}^m$,
\[
\langle F(\theta,x)-  F(\theta',x), \theta-\theta'\rangle \geq \langle \theta-\theta', \hat{A}_1(x) (\theta-\theta')\rangle.
\]
The smallest eigenvalue of $\E[\hat{A}_1(X_0)]$ is a positive real number $\hat{a}_1 > \epsilon$ with $\epsilon >0$.
\end{assumption}

\begin{remark}\label{disF} By Assumptions \ref{expressionH} and \ref{convex}, one obtains, for $\theta\in\mathbb{R}^d$ and $x\in\mathbb{R}^m$,
\[
\langle   F(\theta,x),\theta\rangle\geq \langle \theta, \hat{A}_1^*(x) \theta\rangle -\hat{b}(x),
\]
where $\hat{A}_1^*(x) = \hat{A}_1(x) -\epsilon\mathbf{I}_d$ and $\hat{b}(x)= (L_2(1+|x|)^{\rho+1}+|   F(0,0)|)^2/(4\epsilon)$.
\end{remark}
\begin{proof}See Appendix \ref{proofre4}.
\end{proof}

\begin{corollary}\label{mainconvexdis}
Let Assumptions \ref{expressionH}, \ref{iid}, \ref{clc} and \ref{convex} hold. Then, for any $n \in \nset$, $0<\lambda\leq \lambda_{\max}^*$, where
\[
\lambda_{\max}^*= \min\left\{\frac{\min\{ a^*,( a^*)^{1/3}\}}{24 (1+L_1)^2\mathbb{E}\left[K_{\rho}(X_0)\right]},\frac{1}{ 4a^*}\right\}
\]
with $ a^* = \hat{a}_1 - \epsilon$, there exist constants $C_0^*,C_1^*,C_2^*>0$ such that,
\begin{equation}
W_1(\mathcal{L}(\theta^{\lambda}_n),\pi_{\beta})\leq C_1^* e^{-C_0^*\lambda n}(\mathbb{E}[|\theta_0|^4]+1) +C_2^*\sqrt{\lambda},\ n\in\mathbb{N}.
\end{equation}
\end{corollary}

If $G$ is assumed to be convex in addition to Assumption \ref{convex}, then it can be shown that the rate of convergence is 1/2 in Wasserstein-2 distance between the law of the SGLD algorithm \eqref{SGLD} and the target distribution $\pi_{\beta}$, which appeared to be optimal, see \cite[Example~3.4]{convex}.
\begin{assumption}\label{convexG}
There exist $\hat{A}_2:\rset^m\to\rset^{d\times d}$ such that for any $ x,y \in \rset^d$,
\begin{align*}
\langle y, \hat{A}_2(x) y\rangle \geq 0
\end{align*}
and for each $\theta, \theta' \in \rset^d$, $x\in\mathbb{R}^m$,
\[
\langle  G(\theta,x)- G(\theta',x), \theta-\theta'\rangle \geq \langle \theta-\theta', \hat{A}_2(x) (\theta-\theta')\rangle.
\]
The smallest eigenvalue of $\E[\hat{A}_2(X_0)]$ is a positive real number $\hat{a}_2 > 0$.
\end{assumption}
\begin{remark}\label{hconvex} Assumptions \ref{convex} and \ref{convexG} imply, for each $\theta, \theta' \in \rset^d$, $x\in\mathbb{R}^m$,
\[
\langle H(\theta,x)-H(\theta',x), \theta-\theta'\rangle \geq \langle \theta-\theta', \hat{A}(x) (\theta-\theta')\rangle,
\]
where $\hat{A}(x) =\hat{A}_1(x)  + \hat{A}_2(x) $. Moreover, one obtains
\[
\langle h(\theta)-h(\theta'), \theta-\theta'\rangle \geq \hat{a}| \theta-\theta'|^2,
\]
where $\hat{a} = \hat{a}_1 +\hat{a}_2 $.
\end{remark}

%\noindent Denote by
%\begin{equation}\label{barL}
%\bar{L} = L\mathbb{E}[(1+2|X_0|)^{\rho}].
%\end{equation}

\begin{remark}\label{hconvex2} By Remark \ref{hlip} and Remark \ref{hconvex}, \cite[Theorem~2.1.12]{nesterov} shows that
\[
\langle h(\theta)-h(\theta'), \theta-\theta'\rangle \geq \hat{a}^*| \theta-\theta'|^2 +\frac{1}{\hat{a}+L}|h(\theta)-h(\theta')|^2,
\]
where $\hat{a}^* = \hat{a}L/(\hat{a}+L)$.
\end{remark}
%By assuming the convexity of both $F$ and $G$, the optimal rate of convergence in Wasserstein-2 distance between the SGLD algorithm and the target distribution $\pi_{\beta} $ can be recovered.

Define
\begin{equation}\label{lambdamaxconv}
\bar{\lambda}_{\max} = \min\{1/2(\hat{a}+L), \hat{a}/(4L_1^2\mathbb{E}[K_{\rho}(X_0)])\}
\end{equation}
with $\hat{a} = \hat{a}_1 +\hat{a}_2 $ given in Remark \ref{hconvex}. Under the convexity condition of $H$, the non-asymptotic bound for $W_2(\mathcal{L}(\theta^{\gamma}_n),\pi_{\beta})$ is obtained with the optimal convergence rate 1/2. The explicit statement is given below.
\begin{theorem}\label{mainconvex} Let Assumptions \ref{expressionH}, \ref{iid}, \ref{clc}, \ref{convex} and \ref{convexG} hold. Then, for any $n \in \nset$, $0<\lambda <\bar{\lambda}_{\max}$ given in \eqref{lambdamaxconv}, there exist constants $C_6,C_7,C_8>0$ such that,
\[
W_2(\mathcal{L}(\theta^{\lambda}_n),\pi_{\beta})\leq C_7 e^{-C_6\lambda n} +C_8\sqrt{\lambda},
\]
where $C_6, C_7$ and $C_8$ are given explicitly in \eqref{mainconvexconst}. If $\rho = 0$ in Assumption \ref{expressionH}, then the result holds  for $\lambda \in \min\{1/2(\hat{a}+L) , 1/(6L_1)\}$.
\end{theorem}

By using Theorem \ref{mainconvex}, one can obtain an upper bound for the expected excess risk $ \mathbb{E}[U(\hat{\theta})] - \inf_{\theta \in \mathbb{R}^d} U(\theta)$ in the convex case.
\begin{corollary}\label{eer2}  Let Assumptions \ref{expressionH}, \ref{iid}, \ref{clc}, \ref{convex} and \ref{convexG} hold. Then,  for every $0<\lambda\leq \bar{\lambda}_{\max}$ given in \eqref{lambdamaxconv}, there exist constants $\hat{C}_4, \hat{C}_5, \hat{C}_6, \hat{C}_7>0$ such that the expected excess risk can be estimated as
\begin{equation*}
 \mathbb{E}[U(\hat{\theta})] - \inf_{\theta \in \mathbb{R}^d} U(\theta) \leq \hat{C}_5e^{-\hat{C}_4\lambda n}+\hat{C_6}\sqrt{\lambda} + \hat{C}_7/\beta,
\end{equation*}
where $\hat{\theta} =  \theta^{\lambda}_n$, and $\hat{C}_4, \hat{C}_5, \hat{C}_6, \hat{C}_7>0$ are given explicitly in \eqref{eerconst3} and \eqref{eerconst4}.
\end{corollary}

\section{Proofs of the main results: nonconvex case}\label{po}
Denote by  $\mathcal{F}_t$ the natural filtration of $B_t$, $t\in\mathbb{R}_+$. It is a classic result that SDE \eqref{sde} has a unique solution adapted to $(\mathcal{F}_t)_{t\in\mathbb{R}_+}$, since $h$ is Lipschitz-continuous by \eqref{mulyan}. In order to obtain the convergence results in Theorem \ref{main} and Corollary \ref{cw2}, we first introduce some auxiliary processes.
\subsection{Further notation and introduction of auxiliary processes}\label{ap}
Define the Lyapunov function for each $p\geq 1$ by
\[
V_p(\theta):=(1+|\theta|^2)^{p/2},\ \theta\in\mathbb{R}^d,
\]
and similarly $\operatorname{v}_p (x):= (1+x^2)^{p/2}$, for any real $x\ge0$. Notice that these functions are twice continuously differentiable and
\[
\lim_{|\theta|\to\infty}\frac{\nabla V_p(\theta)}{V_p(\theta)}=0.
\]
Let $\mathcal{P}_{\, V_p}$ denote the set of $\mu\in\mathcal{P}(\mathbb{R}^d)$ satisfying $\int_{\mathbb{R}^d}V_p(\theta)\,\mu(d \theta)<\infty$.

Consider the following auxiliary processes. For each $\lambda>0$,
\[
Z^{\lambda}_t:=Z_{\lambda t},\ t\in\mathbb{R}_+.
\]
Notice that $\tilde{B}^{\lambda}_t:=B_{\lambda t}/\sqrt{\lambda}$, $t\in\mathbb{R}_+$
is also a Brownian motion and
\[
d Z^{\lambda}_t=-\lambda h(Z^{\lambda}_t)\, d t+\sqrt{2\beta^{-1}\lambda} d \tilde{B}^{\lambda}_t,\
Z^{\lambda}_0=\theta_0.
\]
Then, $\mathcal{F}_t^{\lambda}:=\mathcal{F}_{\lambda t}$, $t\in\mathbb{R}_+$ is the natural filtration of $\tilde{B}^{\lambda}_t$, $t\in\mathbb{R}_+$. One notice that $\mathcal{F}_t^{\lambda}$ is independent of $\mathcal{G}_{\infty}\vee \sigma(\theta_0)$. Then, define the continuous-time interpolation of the SGLD algorithm \eqref{SGLD} as
\begin{equation}\label{SGLDprocess}
d \bar{\theta}^{\lambda}_t=-\lambda H(\bar{\theta}^{\lambda}_{\lfloor t\rfloor},{X}_{\lceil t\rceil})\, d t
+ \sqrt{2\beta^{-1}\lambda} d \tilde{B}^{\lambda}_{t},
\end{equation}
with initial condition $\bar{\theta}^{\lambda}_0=\theta_0$. In addition, %for each integer $n\in\mathbb{N}$,
%\begin{equation}\label{antoniojobim}
%\mathcal{L}(\bar{\theta}^{\lambda}_n)=\mathcal{L}(\theta_n^{\lambda}).
%\end{equation}
due to the homogeneous nature of the coefficients of equation \eqref{SGLDprocess}, the law of the interpolated process coincides with the law of the SGLD algorithm \eqref{SGLD} at grid-points, i.e. $\mathcal{L}(\bar{\theta}^{\lambda}_n)=\mathcal{L}(\theta_n^{\lambda})$, for each $n\in\mathbb{N}$. Hence, crucial estimates for the SGLD can be derived by studying  equation \eqref{SGLDprocess}.

Furthermore, consider a continuous-time process $ \zeta^{s,v, \lambda}_t$, $t\geq s$, which denotes the solution of the SDE
\begin{equation*}%\label{auxiliarypro}
d \zeta^{s,v, \lambda}_t= -\lambda h(\zeta^{s,v, \lambda}_t) d t + \sqrt{2\beta^{-1}\lambda} d \tilde{B}_t^{\lambda}.
\end{equation*}
with initial condition $\zeta^{s,v, \lambda}_s := v$, $v \in \mathbb{R}^d$.
\begin{definition}\label{zetaprocess} Fix $n \in \nset$ and define
\begin{align*}
\bar{\zeta}_t^{\lambda,n} = \zeta^{nT, \bar{\theta}^{\lambda}_{nT}, \lambda}_t
\end{align*}
where $T := \floor{{1}/{\lambda}}$.
\end{definition}
Intuitively, $\bar{\zeta}_t^{\lambda,n}$ is a process started from the value of the SGLD process \eqref{SGLDprocess} at time $nT$ and made run until time $t \geq nT$ with the continuous-time Langevin dynamics.

\subsection{Preliminary estimates}
We proceed by establishing the moment bounds of the processes $(\bar{\theta}^{\lambda}_t)_{t\geq 0}$ and $(\bar{\zeta}^{\lambda,n}_t)_{t\geq 0}$. %and of the auxiliary processses introduced above.

\begin{lemma}
\label{lem:moment_SGLD_2p}
Let Assumptions \ref{expressionH}, \ref{iid} and \ref{assum:dissipativity} hold. For any $0<  \lambda < \lambda_{\max}$ given in \eqref{eq:definition-lambda-max}, $n\in\mathbb{N}$, $t \in (n,n+1]$,
\[
\mathbb{E} \left[|\bar{\theta}^{\lambda}_t|^2\right] \leq (1-  a\lambda (t-n) )(1-a\lambda)^n \mathbb{E} \left[|\theta_0|^2\right] +c_1 (\lambda_{\max}+a^{-1}) \, ,
\]
where
\begin{equation}\label{2ndmomentconst}
c_1= (c_0+ 2d/\beta), \quad c_0 = 8\mathbb{E}\left[ K_1^2(X_0)\right]a^{-1}+2 b +4\lambda_{\max}L_2^2\mathbb{E}\left[K_{\rho}(X_0)\right]  +4\lambda_{\max}\mathbb{E}\left[ F_*^2(X_0)\right].
%\quad c_0 = 4\lambda_{\max}L_2^2\mathbb{E}\left[(1+|X_0|)^{2\rho+2}\right] +4\lambda_{\max}H_\star^2 +2 b.
\end{equation}
In addition, $\sup_t\mathbb{E} |\bar{\theta}^{\lambda}_t|^2 \leq \mathbb{E}\left[ |\theta_0|^2\right] +c_1 (\lambda_{\max}+a^{-1})< \infty$. Similarly, one obtains
\[
\mathbb{E}\left[|\bar{\theta}^{\lambda}_t|^4\right] \leq (1-  a\lambda (t-n) )(1-a\lambda)^n\mathbb{E}| \theta_0|^4 +c_3 (\lambda_{\max}+a^{-1}),
\]
where
\begin{equation}\label{4thmomentconst}
c_3 = (1+a\lambda_{\max})c_2+12d^2\beta^{-2}(\lambda_{\max}+9a^{-1})
\end{equation}
with $c_2$ given in \eqref{estc2}. Moreover, this implies $\sup_t\mathbb{E} |\bar{\theta}^{\lambda}_t|^4 < \infty$.
\end{lemma}
\begin{proof}
For any $n \in \mathbb{N}$ and $t\in (n, n+1]$, define $\Delta_{n,t}= \bar{\theta}^{\lambda}_n - \lambda H(\bar{\theta}^{\lambda}_n, X_{n+1})(t-n)$. By using \eqref{SGLDprocess}, it is easily seen that for $t\in (n, n+1]$
\begin{equation*}
\mathbb{E}\left[|\bar{\theta}^{\lambda}_t|^2\left| \bar{\theta}^{\lambda}_n \right.\right] = \mathbb{E}\left[|\Delta_{n,t}|^2\left| \bar{\theta}^{\lambda}_n \right.\right]  + (2 \lambda/\beta)d (t-n).
\end{equation*}
Then, by using Assumptions \ref{expressionH}, \ref{iid}, \ref{assum:dissipativity} and Remark \ref{growth}, one obtains
\begin{align*}
\mathbb{E}\left[|\Delta_{n,t}|^2\left| \bar{\theta}^{\lambda}_n\right.\right]
& = | \bar{\theta}^{\lambda}_n|^2 - 2\lambda (t-n)\mathbb{E}\left[ \left\langle  \bar{\theta}^{\lambda}_n, H(\bar{\theta}^{\lambda}_n, X_{n+1}) \right\rangle\left| \bar{\theta}^{\lambda}_n\right.\right]\\
&\quad +\lambda^2(t-n)^2\mathbb{E}\left[|H(\bar{\theta}^{\lambda}_n, X_{n+1})|^2\left| \bar{\theta}^{\lambda}_n\right.\right] \\
& \leq | \bar{\theta}^{\lambda}_n|^2 - 2\lambda (t-n)\left\langle  \bar{\theta}^{\lambda}_n, \mathbb{E}\left[ A(X_0)\right] \bar{\theta}^{\lambda}_n\right\rangle +2\lambda (t-n) b\\
&\quad - 2\lambda (t-n)\mathbb{E}\left[ \left\langle  \bar{\theta}^{\lambda}_n, G(\bar{\theta}^{\lambda}_n, X_{n+1}) \right\rangle\left| \bar{\theta}^{\lambda}_n\right.\right]\\
&\quad +\lambda^2(t-n)^2\mathbb{E}\left[ ((1+|X_{n+1}|)^{\rho +1}(L_1|\bar{\theta}^{\lambda}_n|+L_2)+F_*(X_{n+1}))^2\left| \bar{\theta}^{\lambda}_n\right.\right]\\
& \leq (1-2a\lambda (t-n))| \bar{\theta}^{\lambda}_n|^2  +  2\lambda (t-n) b+2\lambda (t-n) \mathbb{E}\left[K_1(X_0)\right]| \bar{\theta}^{\lambda}_n|\\
&\quad +2\lambda^2(t-n)^2L_1^2\mathbb{E}\left[ K_{\rho}(X_0)\right]|\bar{\theta}^{\lambda}_n|^2 +4\lambda^2(t-n)^2L_2^2\mathbb{E}\left[K_{\rho}(X_0)\right]+4\lambda^2(t-n)^2\mathbb{E}\left[F_*^2(X_0)\right],
\end{align*}
where the last inequality is obtained by using $(a+b)^2 \leq 2a^2 +2b^2$, for $a,b \geq 0$ twice. For $\lambda <\lambda_{\max}$ with $\lambda_{\max}$ given in \eqref{eq:definition-lambda-max},
\begin{align*}
\mathbb{E}\left[|\Delta_{n,t}|^2\left|  \bar{\theta}^{\lambda}_n \right.\right]
&\leq  \left(1-\frac{3}{2}a\lambda (t-n)\right)| \bar{\theta}^{\lambda}_n|^2+2\lambda (t-n)\mathbb{E}\left[ K_1(X_0)\right]| \bar{\theta}^{\lambda}_n|\\
&\quad  +  2\lambda (t-n) b+4\lambda^2(t-n)^2L_2^2\mathbb{E}\left[K_{\rho}(X_0)\right]+4\lambda^2(t-n)^2\mathbb{E}[F_*^2(X_0)].
\end{align*}
For $| \bar{\theta}^{\lambda}_n|>4\mathbb{E}\left[ K_1(X_0)\right]a^{-1}$, one obtains
\[
-\frac{1}{2}a\lambda (t-n)| \bar{\theta}^{\lambda}_n|^2+2\lambda (t-n)\mathbb{E}\left[ K_1(X_0)\right]| \bar{\theta}^{\lambda}_n|<0,
\]
which implies
\begin{align*}
\mathbb{E}\left[|\Delta_{n,t}|^2\left|  \bar{\theta}^{\lambda}_n \right.\right]
&\leq  \left(1- a\lambda (t-n)\right)| \bar{\theta}^{\lambda}_n|^2 +  2\lambda (t-n) b\\
&\quad +4\lambda^2(t-n)^2L_2^2\mathbb{E}\left[K_{\rho}(X_0)\right]+4\lambda^2(t-n)^2\mathbb{E}\left[F_*^2(X_0)\right].
\end{align*}
For $| \bar{\theta}^{\lambda}_n|\leq 4\mathbb{E}\left[ K_1(X_0)\right]a^{-1}$, we have
\begin{align*}
\mathbb{E}\left[|\Delta_{n,t}|^2\left|  \bar{\theta}^{\lambda}_n \right.\right]
&\leq  \left(1-\frac{3}{2}a\lambda (t-n)\right)| \bar{\theta}^{\lambda}_n|^2+8\lambda (t-n)\mathbb{E}\left[ K_1^2(X_0)\right]a^{-1}\\
&\quad  +  2\lambda (t-n) b+4\lambda^2(t-n)^2L_2^2\mathbb{E}\left[K_{\rho}(X_0)\right]+4\lambda^2(t-n)^2\mathbb{E}\left[F_*^2(X_0)\right].
\end{align*}
Combining the two cases yields
\begin{align*}
\mathbb{E}\left[|\Delta_{n,t}|^2\left|  \bar{\theta}^{\lambda}_n \right.\right]
&\leq  \left(1- a\lambda (t-n)\right)| \bar{\theta}^{\lambda}_n|^2 +  \lambda (t-n) c_0,
\end{align*}
where $c_0 = 8\mathbb{E}\left[ K_1^2(X_0)\right]a^{-1}+2 b +4\lambda_{\max}L_2^2\mathbb{E}\left[K_{\rho}(X_0)\right] +4\lambda_{\max}\mathbb{E}\left[F_*^2(X_0)\right]$. Therefore, one obtains
\[
\mathbb{E}\left[| \bar{\theta}^{\lambda}_t|^2\left|  \bar{\theta}^{\lambda}_n \right.\right] \leq (1-a\lambda (t-n))| \bar{\theta}^{\lambda}_n|^2 +\lambda (t-n)c_1,
\]
where $c_1= (c_0+ 2d/\beta)$ and the result follows by induction. To calculate a higher moment, denote by $\Xi_{n,t}^{\lambda}= \{ 2 \lambda \beta^{-1} \}^{1/2}(\tilde{B}_t^{\lambda}-\tilde{B}_n^{\lambda})$, for $t \in (n, n+1]$, one calculates
\begin{align}\label{4thmoment}
\mathbb{E}\left[|  \bar{\theta}^{\lambda}_t|^4\left| \bar{\theta}^{\lambda}_n\right.\right]
& = \mathbb{E}\left[\left(|\Delta_{n,t}|^2+|\Xi_{n,t}^{\lambda}|^2+2\left\langle \Delta_{n,t}, \Xi_{n,t}^{\lambda}\right\rangle\right)^2\left| \bar{\theta}^{\lambda}_n\right.\right]\nonumber\\
& = \mathbb{E}\left[ |\Delta_{n,t}|^4+|\Xi_{n,t}^{\lambda}|^4+2|\Delta_{n,t}|^2|\Xi_{n,t}^{\lambda}|^2+4|\Delta_{n,t}|^2\left\langle\Delta_{n,t}, \Xi_{n,t}^{\lambda}\right\rangle \right.\nonumber\\
&\qquad \left. +4|\Xi_{n,t}^{\lambda}|^2\left\langle \Delta_{n,t}, \Xi_{n,t}^{\lambda}\right\rangle+4\left(\left\langle \Delta_{n,t}, \Xi_{n,t}^{\lambda}\right\rangle\right)^2 \left| \bar{\theta}^{\lambda}_n\right.\right]\nonumber\\
& \leq \mathbb{E}\left[ |\Delta_{n,t}|^4+|\Xi_{n,t}^{\lambda}|^4+6|\Delta_{n,t}|^2|\Xi_{n,t}^{\lambda}|^2 \left| \bar{\theta}^{\lambda}_n\right.\right]\nonumber\\
&\leq (1+a\lambda(t-n))\mathbb{E}\left[ |\Delta_{n,t}|^4\left| \bar{\theta}^{\lambda}_n\right.\right] + (1+9/(a\lambda(t-n))) \mathbb{E}\left[ |\Xi_{n,t}^{\lambda}|^4\right].
\end{align}
where the last inequality holds due to $2ab \leq \varepsilon a^2 +\varepsilon^{-1}b^2$, for $a,b \geq 0$ and $\varepsilon>0$ with $\varepsilon = a\lambda(t-n)$. Then, one continues with calculating
\begin{align*}
\mathbb{E}\left[|\Delta_{n,t}|^4\left| \bar{\theta}^{\lambda}_n \right.\right]
&= \mathbb{E}\left[\left(| \bar{\theta}^{\lambda}_n|^2 -2\lambda (t-n) \left\langle  \bar{\theta}^{\lambda}_n, H(\bar{\theta}^{\lambda}_n, X_{n+1}) \right\rangle+\lambda^2(t-n)^2|H(\bar{\theta}^{\lambda}_n, X_{n+1})|^2\right)^2\left| \bar{\theta}^{\lambda}_n \right.\right] \\
&\leq | \bar{\theta}^{\lambda}_n|^4+\mathbb{E}\left[6\lambda^2(t-n)^2| \bar{\theta}^{\lambda}_n|^2|H(\bar{\theta}^{\lambda}_n, X_{n+1})|^2 - 4\lambda (t-n) \left\langle  \bar{\theta}^{\lambda}_n, H(\bar{\theta}^{\lambda}_n, X_{n+1}) \right\rangle | \bar{\theta}^{\lambda}_n|^2 \right. \\
&\qquad \left. -4\lambda^3 (t-n)^3|H(\bar{\theta}^{\lambda}_n, X_{n+1})|^2 \left\langle  \bar{\theta}^{\lambda}_n, H(\bar{\theta}^{\lambda}_n, X_{n+1}) \right\rangle +\lambda^4(t-n)^4|H(\bar{\theta}^{\lambda}_n, X_{n+1})|^4   \left| \bar{\theta}^{\lambda}_n \right.\right].
\end{align*}
By Remark \ref{growth}, for $q\geq 1$, one observes
\begin{equation}\label{Hest}
\mathbb{E}\left[|H(\bar{\theta}^{\lambda}_n, X_{n+1})|^q	\left| \bar{\theta}^{\lambda}_n \right.\right]
\leq \mathbb{E}\left[(1+|X_0|)^{q\rho+q}\right] (2^{q-1}L_1^q|\bar{\theta}^{\lambda}_n|^q +2^{2q-2}L_2^q)+2^{2q-2}\mathbb{E}\left[ F_*^q(X_0)\right].
\end{equation}
Then, by using Assumption \ref{assum:dissipativity} and by taking $q = 2, 3, 4$ in \eqref{Hest}, one obtains
\begin{align*}
&\mathbb{E}\left[|\Delta_{n,t}|^4\left| \bar{\theta}^{\lambda}_n \right.\right]\\
&\leq  (1- 4a\lambda (t-n) )| \bar{\theta}^{\lambda}_n|^4 +4b\lambda (t-n) | \bar{\theta}^{\lambda}_n|^2  + 4\lambda (t-n) \mathbb{E}\left[K_1(X_0)\right]| \bar{\theta}^{\lambda}_n|^3\\
&\quad +12\lambda^2(t-n)^2L_1^2\mathbb{E}\left[K_{\rho}(X_0)\right]| \bar{\theta}^{\lambda}_n|^4+24\lambda^2(t-n)^2\left(L_2^2\mathbb{E}\left[K_{\rho}(X_0)\right]+\mathbb{E}\left[ F_*^2(X_0)\right]\right) |\bar{\theta}^{\lambda}_n|^2\\
&\quad + 16\lambda^3(t-n)^3L_1^3\mathbb{E}\left[K_{\rho}(X_0)\right]| \bar{\theta}^{\lambda}_n|^4+64\lambda^3(t-n)^3\left(L_2^3\mathbb{E}\left[K_{\rho}(X_0)\right]+\mathbb{E}\left[ F_*^3(X_0)\right]\right) |\bar{\theta}^{\lambda}_n|\\
&\quad +8\lambda^4(t-n)^4L_1^4\mathbb{E}\left[K_{\rho}(X_0)\right]| \bar{\theta}^{\lambda}_n|^4+64\lambda^4(t-n)^4\left(L_2^4\mathbb{E}\left[K_{\rho}(X_0)\right]+\mathbb{E}\left[ F_*^4(X_0)\right]\right),
\end{align*}
which implies, by using $\lambda < \lambda_{\max}$
\begin{align*}
\mathbb{E}\left[|\Delta_{n,t}|^4\left| \bar{\theta}^{\lambda}_n \right.\right]
&\leq  (1- 3a\lambda (t-n) )| \bar{\theta}^{\lambda}_n|^4   + 4\lambda (t-n) \mathbb{E}\left[K_1(X_0)\right]| \bar{\theta}^{\lambda}_n|^3\\
&\quad +4b\lambda (t-n) | \bar{\theta}^{\lambda}_n|^2+24\lambda^2(t-n)^2\left(L_2^2\mathbb{E}\left[K_{\rho}(X_0)\right]+\mathbb{E}\left[ F_*^2(X_0)\right]\right) |\bar{\theta}^{\lambda}_n|^2\\
&\quad +64\lambda^3(t-n)^3\left(L_2^3\mathbb{E}\left[K_{\rho}(X_0)\right]+\mathbb{E}\left[ F_*^3(X_0)\right]\right) |\bar{\theta}^{\lambda}_n|\\
&\quad +64\lambda^4(t-n)^4\left(L_2^4\mathbb{E}\left[K_{\rho}(X_0)\right]+\mathbb{E}\left[ F_*^4(X_0)\right]\right).
\end{align*}
For $|\bar{\theta}^{\lambda}_n|>12\mathbb{E}\left[K_1(X_0)\right]a^{-1}$, one obtains
\[
-\frac{a\lambda (t-n) }{3} | \bar{\theta}^{\lambda}_n|^4 + 4\lambda (t-n) \mathbb{E}\left[K_1(X_0)\right]| \bar{\theta}^{\lambda}_n|^3<0,
\]
similarly, for $|\bar{\theta}^{\lambda}_n|>(12ba^{-1}+72a^{-1}\lambda_{\max}\left(L_2^2\mathbb{E}\left[K_{\rho}(X_0)\right]+\mathbb{E}\left[ F_*^2(X_0)\right]\right))^{1/2}$, we have
\[
-\frac{a\lambda (t-n) }{3} | \bar{\theta}^{\lambda}_n|^4 +4b\lambda (t-n)| \bar{\theta}^{\lambda}_n|^2+24\lambda^2(t-n)^2\left(L_2^2\mathbb{E}\left[K_{\rho}(X_0)\right]+\mathbb{E}\left[ F_*^2(X_0)\right]\right) |\bar{\theta}^{\lambda}_n|^2<0,
\]
moreover, for $| \bar{\theta}^{\lambda}_n|>(192a^{-1}\lambda_{\max}^2\left(L_2^3\mathbb{E}\left[K_{\rho}(X_0)\right]+\mathbb{E}\left[ F_*^3(X_0)\right]\right) )^{1/3}$
\[
- \frac{a\lambda (t-n) }{3}| \bar{\theta}^{\lambda}_n|^4+64\lambda^3(t-n)^3\left(L_2^3\mathbb{E}\left[K_{\rho}(X_0)\right]+\mathbb{E}\left[ F_*^3(X_0)\right]\right) |\bar{\theta}^{\lambda}_n|<0.
\]
Denote by
\begin{align}\label{estM}
\begin{split}
M &= \max\left\{12\mathbb{E}\left[K_1(X_0)\right]a^{-1}, (12ba^{-1}+72a^{-1}\lambda_{\max}\left(L_2^2\mathbb{E}\left[K_{\rho}(X_0)\right]+\mathbb{E}\left[ F_*^2(X_0)\right]\right))^{1/2},\right.\\
&\hspace{4em} \left. (192a^{-1}\lambda_{\max}^2\left(L_2^3\mathbb{E}\left[K_{\rho}(X_0)\right]+\mathbb{E}\left[ F_*^3(X_0)\right]\right) )^{1/3}\right\}.
\end{split}
\end{align}
For $|\bar{\theta}^{\lambda}_n|>M$, one obtains
\begin{align*}
\mathbb{E}\left[|\Delta_{n,t}|^4\left| \bar{\theta}^{\lambda}_n \right.\right]
& \leq  (1- 2a\lambda (t-n) )| \bar{\theta}^{\lambda}_n|^4  +64\lambda^4(t-n)^4\left(L_2^4\mathbb{E}\left[K_{\rho}(X_0)\right]+\mathbb{E}\left[ F_*^4(X_0)\right]\right). % +64\lambda^4(t-n)^4\left(L_2^4\mathbb{E}\left[(1+|X_0|)^{4\rho+4}\right]+H_\star^4\right).
\end{align*}
As for $|\bar{\theta}^{\lambda}_n|\leq M$, we have
\begin{align*}
\mathbb{E}\left[|\Delta_{n,t}|^4\left| \bar{\theta}^{\lambda}_n \right.\right]
&\leq  (1- 3a\lambda (t-n) )| \bar{\theta}^{\lambda}_n|^4   + 4\lambda (t-n) \mathbb{E}\left[K_1(X_0)\right]M^3+4b\lambda (t-n)M^2\\
&\quad +24\lambda^2(t-n)^2\left(L_2^2\mathbb{E}\left[K_{\rho}(X_0)\right]+\mathbb{E}\left[ F_*^2(X_0)\right]\right) M^2\\
&\quad +64\lambda^3(t-n)^3\left(L_2^3\mathbb{E}\left[K_{\rho}(X_0)\right]+\mathbb{E}\left[ F_*^3(X_0)\right]\right) M\\
&\quad +64\lambda^4(t-n)^4\left(L_2^4\mathbb{E}\left[K_{\rho}(X_0)\right]+\mathbb{E}\left[ F_*^4(X_0)\right]\right).
\end{align*}
Combining the two cases yields
\begin{equation}\label{deltaest4}
\mathbb{E}\left[|\Delta_{n,t}|^4\left| \bar{\theta}^{\lambda}_n \right.\right]\leq  (1- 2a\lambda (t-n) )| \bar{\theta}^{\lambda}_n|^4 + \lambda (t-n) c_2,
\end{equation}
where
\begin{equation}\label{estc2}
c_2 = 4\mathbb{E}\left[K_1(X_0)\right]M^3+ 4bM^2 + 152(1+\lambda_{\max})^3 \left((1+L_2)^4\mathbb{E}\left[K_{\rho}(X_0)\right]+(1+\mathbb{E}\left[ F_*^4(X_0)\right]\right)(1+M)^2
%8\mathbb{E}\left[(1+|X_0|)^{4\rho+4}\right](3L_2^2M^2+8L_2^3M+8L_2^4)+8(3H_\star^2M^2+8H_\star^3M+8H_\star^4),
\end{equation}
with $M$ given in \eqref{estM}. %and $C_{\rho}$ defined in \eqref{crho}.
Substituting \eqref{deltaest4} into \eqref{4thmoment}, one obtains
\begin{align*}
\mathbb{E}\left[|\bar{\theta}^{\lambda}_t|^4\left| \bar{\theta}^{\lambda}_n \right.\right]
& \leq  (1+a\lambda(t-n))(1- 2a\lambda (t-n) )| \bar{\theta}^{\lambda}_n|^4 \\
&\quad +(1+a\lambda(t-n)) \lambda (t-n) c_2 +12d^2\lambda^2 \beta^{-2}(t-n)^2 (1+9/(a\lambda(t-n)))\\
& \leq (1-  a\lambda (t-n) )| \bar{\theta}^{\lambda}_n|^4 +\lambda (t-n) c_3,
\end{align*}
where $c_3 = (1+a\lambda_{\max})c_2+12d^2\beta^{-2}(\lambda_{\max}+9a^{-1})$. The proof completes by induction.
\end{proof}

\begin{remark}\label{stepr} One notices that in Lemma \ref{lem:moment_SGLD_2p}, the step-size restriction is the following:
\[
\hat{\lambda}_{\max} = \min\left\{\frac{a }{24 L_1^2\mathbb{E}\left[ K_{\rho}(X_0)\right] }, \frac{a^{1/2}}{8 (L_1^3\mathbb{E}\left[K_{\rho}(X_0)\right])^{1/2}}, \frac{a^{1/3}}{  (32L_1^4\mathbb{E}\left[K_{\rho}(X_0)\right])^{1/3}},\frac{1}{4a}\right\} \,.
\]
Theorem \ref{main} and Corollary \ref{cw2} still hold by using $\hat{\lambda}_{\max} $. However, in order to make notation compact, the restriction is chosen to be $\lambda_{\max}$ given in \eqref{eq:definition-lambda-max}, which can be deduced from the above expression.
\end{remark}

\begin{corollary}\label{corr:Moments} Let Assumptions \ref{expressionH}, \ref{iid} and \ref{assum:dissipativity} hold. For any $0< \lambda < \lambda_{\max}$ given in \eqref{eq:definition-lambda-max}, $n\in\mathbb{N}$, $t \in (n,n+1]$,
\begin{align*}
\E[V_4(\bar{\theta}^{\lambda}_t)] \leq 2(1 - a\lambda)^\floor{t} \E[V_4(\theta_0)] + 2  c_3  (\lambda_{\textnormal{max}} + a^{-1})+2,
\end{align*}
where $c_3$ is given in \eqref{4thmomentconst}.
\end{corollary}

Next, we present a drift condition associated with the SDE \eqref{sde}, which will be used to obtain the moment bounds of the process $(\bar{\zeta}^{\lambda,n}_t)_{t\geq 0}$.
\begin{lemma}\label{lem:PreLimforDriftY} Let Assumptions \ref{expressionH}, \ref{iid} and \ref{assum:dissipativity} hold. Then, for each $p\geq 2$, $\theta \in \rset^d$,
\begin{align*}
\frac{\Delta V_p}{\beta} - \langle h(\theta), \nabla V_p(\theta) \rangle \leq -\bar{c}(p) V_p(\theta) + \tilde{c}(p),
\end{align*}
where $\bar{c}(p)= ap/4$ and $ \tilde{c}(p) = (3/4) a p \mathrm{v}_{p+1}(\overline{M}_p)$ with $\overline{M}_p$ given in \eqref{Mp}.
\end{lemma}
\begin{proof}
%See \cite[Lemma~3.6]{nonconvex}.
One notices that, by Assumptions \ref{expressionH} and \ref{iid}, for any $\theta \in \mathbb{R}^d$, $ h(\theta) = \mathbb{E}[H(\theta, X_0)] =  \mathbb{E}[ F(\theta, X_0)+ G(\theta, X_0)] $. Then, one calculates,
\begin{align*}
&\frac{\Delta V_p}{\beta} - \langle h(\theta), \nabla V_p(\theta) \rangle\\
&= \beta^{-1}p(p-2)|\theta|^2V_{p-4}(\theta) + \beta^{-1}pdV_{p-2}(\theta) \\
&\quad - pV_{p-2}(\theta) \langle \mathbb{E}[ F(\theta, X_0)+ G(\theta, X_0)], \theta \rangle\\
&\leq -apV_p(\theta) + (ap+bp+ \beta^{-1}p(p-2)  + \beta^{-1}pd)V_{p-2}(\theta) +p\mathbb{E}\left[K_1(X_0)\right]|\theta|V_{p-2}(\theta),
\end{align*}
where the last inequality is obtained due to Assumption \ref{assum:dissipativity}. By observing $|\theta| \leq \sqrt{1+|\theta|^2}$, denote by
\begin{equation}\label{Mp}
\overline{M}_p =\sqrt{(4/3 + 4b/(3a) + 4d/(3a\beta) + 4(p-2)/(3a\beta)+4\mathbb{E}\left[K_1(X_0)\right]/(3a))^2-1}.
\end{equation}
For $|\theta|>\overline{M}_p $, one obtains $\frac{\Delta V_p}{\beta} - \langle h(\theta), \nabla V_p(\theta) \rangle \leq -(ap/4)V_p(\theta)$, while for $|\theta|\leq \overline{M}_p $, we have $\frac{\Delta V_p}{\beta} - \langle h(\theta), \nabla V_p(\theta) \rangle \leq (3/4) a p \mathrm{v}_{p+1}(\overline{M}_p)$. Combining the two cases yields the desired result.
\end{proof}

The following Lemma provides the second and the fourth moment of the process $(\bar{\zeta}^{\lambda,n}_t)_{t\geq 0}$.
\begin{lemma}\label{zetaprocessmoment} Let Assumptions \ref{expressionH}, \ref{iid} and \ref{assum:dissipativity} hold. For any $0< \lambda < \lambda_{\max}$ given in \eqref{eq:definition-lambda-max}, $t \geq nT$, $n\in \nset$, one obtains the following inequality
\begin{align*}
\E[V_2(\bar{\zeta}_t^{\lambda,n})] &\leq e^{-a\lambda  t/2}   \E[V_2(\theta_0)] +3 \mathrm{v}_3(\overline{M}_2)+c_1(\lambda_{\max}+a^{-1})+1,
\end{align*}
where the process $\bar{\zeta}_t^{\lambda,n}$ is defined in Definition \ref{zetaprocess} and $c_1$ is given in \eqref{2ndmomentconst}. Furthermore,
\begin{align*}
\E[V_4(\bar{\zeta}_t^{\lambda,n})] \leq 2e^{-a\lambda  t}   \E[V_4(\theta_0)] +3 \mathrm{v}_5(\overline{M}_4)+2c_3 (\lambda_{\max}+a^{-1})+2,
\end{align*}
where $c_3$ is given in \eqref{4thmomentconst}.
\end{lemma}
\begin{proof} For any $p\geq 1$, application of Ito's lemma and taking expectation yields
\begin{align*}
\E[V_p(\bar{\zeta}_t^{\lambda,n})] = \E[V_p( \bar{\theta}^{\lambda}_{nT})] + \int_{nT}^t \E\left[\lambda \frac{\Delta V_p(\bar{\zeta}_s^{\lambda,n})}{\beta} - \lambda \langle h(\bar{\zeta}_s^{\lambda,n}), \nabla V_p(\bar{\zeta}_s^{\lambda,n}) \rangle\right] d s.
\end{align*}
Differentiating both sides and using Lemma~\ref{lem:PreLimforDriftY}, we arrive at
\begin{align*}
\frac{d}{d t} \E[V_p(\bar{\zeta}_t^{\lambda,n})] = \E\left[ \lambda \frac{\Delta V_p(\bar{\zeta}_t^{\lambda,n})}{\beta} - \lambda \langle h(\bar{\zeta}_t^{\lambda,n}), \nabla V_p(\bar{\zeta}_t^{\lambda,n}) \rangle\right] \leq -\lambda \bar{c}(p) \E[V_p(\bar{\zeta}_t^{\lambda,n})] + \lambda \tilde{c}(p),
\end{align*}
which yields
\begin{align*}
\E[V_p(\bar{\zeta}_t^{\lambda,n})] &\leq e^{-\lambda (t - nT) \bar{c}(p)} \E[V_p( \bar{\theta}^{\lambda}_{nT})] + \frac{\tilde{c}(p)}{\bar{c}(p)} \left( 1 - e^{-\lambda \bar{c}(p) (t - nT)} \right) \\
&\leq e^{-\lambda (t - nT) \bar{c}(p)} \E[V_p(\bar{\theta}^{\lambda}_{nT})] + \frac{\tilde{c}(p)}{\bar{c}(p)}.
\end{align*}
Now for $p = 2$, by using Corollary~\ref{corr:Moments},% and Lemma \ref{lem:PreLimforDriftY},
one obtains
\begin{align*}
\E[V_2(\bar{\zeta}_t^{\lambda,n})] &\leq e^{-\lambda(t-nT) \bar{c}(2)} \E[V_2(\bar{\theta}^{\lambda}_{nT})] + \frac{\tilde{c}(2)}{\bar{c}(2)} \\
&\leq (1-a\lambda)^{nT} e^{-\lambda(t-nT) \bar{c}(2)} \mathbb{E}[ V_2(\theta_0)]+ \frac{\tilde{c}(2)}{\bar{c}(2)} +c_1 (\lambda_{\max}+a^{-1})+1\\
&\leq e^{-a\lambda  t/2}   \E[V_2(\theta_0)] +3 \mathrm{v}_3(\overline{M}_2)+c_1(\lambda_{\max}+a^{-1})+1,
\end{align*}
where the last inequality holds due to $1-z \leq e^{-z}$ for $z \geq 0$ and $\bar{c}(2) =a/2$. Similarly, for $p = 4$, one obtains
\begin{align*}
\E[V_4(\bar{\zeta}_t^{\lambda,n})] &\leq e^{-\lambda(t-nT) \bar{c}(4)} \E[V_4(\bar{\theta}^{\lambda}_{nT})] + \frac{\tilde{c}(4)}{\bar{c}(4)} \\
&\leq 2(1 - a\lambda)^{nT} e^{-\lambda(t-nT) \bar{c}(4)} \E[V_4(\theta_0)] + \frac{\tilde{c}(4)}{\bar{c}(4)}+2c_3 (\lambda_{\max}+a^{-1})+2\\
&\leq 2e^{-a\lambda  t}   \E[V_4(\theta_0)] +3 \mathrm{v}_5(\overline{M}_4)+2c_3 (\lambda_{\max}+a^{-1})+2,
\end{align*}
where the last inequality holds due to $1-z \leq e^{-z}$ for $z \geq 0$ and $\bar{c}(4) =a$.
%Noting $\tilde{c}(p)/\bar{c}(p) = 3 \mathrm{v}_p(\overline{M}_p)$ concludes the proof.
\end{proof}

\subsection{Proof of the main theorems}%{Bounding $W_1(\mathcal{L}({Y}^{\lambda}_t),\mathcal{L}(L^{\lambda}_t))$}
We introduce a functional which is crucial to obtain the convergence rate in $W_1$. For any $p\geq 1$,
$ \mu,\nu \in \mathcal{P}_{\, V_p}$,
\begin{equation}
\label{eq:definition-w-1}
w_{1,p}(\mu,\nu):=\inf_{\zeta\in\mathcal{C}(\mu,\nu)}\int_{\mathbb{R}^d}\int_{\mathbb{R}^d} [1\wedge |\theta-\theta'|](1+V_p(\theta)+V_p(\theta'))\zeta(d\theta d\theta'),
\end{equation}
and it satisfies trivially
\begin{equation}\label{lucia}
W_1(\mu,\nu)\leq w_{1,p}(\mu,\nu).
\end{equation}
The case $p=2$, i.e. $w_{1,2}$, is used throughout the section.  The result below states a contraction property of $w_{1,2}$.

\begin{proposition} \label{contr} Let $Z_t'$, $t\in\rset_+$ be the solution of \eqref{sde} with initial condition $Z'_0 = \theta_0$ which is independent of $\calF_\infty$ and satisfies $|\theta_0|_2$ is finite. Then,
\begin{align*}
w_{1,2}(\calL(Z_t),\calL(Z'_t)) \leq \hat{c} e^{-\dot{c} t} w_{1,2}(\calL(\theta_0),\calL(\theta_0')),
\end{align*}
where the constants $\dot{c}$ and $\hat{c}$ are given in Lemma \ref{contractionconst}.
\end{proposition}
\begin{proof}
See Proposition~3.14 of \cite{nonconvex}.
\end{proof}
By using the contraction property provided in Proposition \ref{contr}, one can construct the non-asymptotic bound between $\mathcal{L}(\bar{\theta}^{\lambda}_t)$ and $\mathcal{L}(Z^{\lambda}_t)$, $t \in [nT, (n+1)T]$, in $W_1$ distance by decomposing the error using the auxiliary process $\bar{\zeta}_t^{\lambda,n}$:
\begin{align}\label{decomp}
W_1(\mathcal{L}(\bar{\theta}^{\lambda}_t),\mathcal{L}(Z^{\lambda}_t)) \leq W_1(\mathcal{L}(\bar{\theta}^{\lambda}_t),\mathcal{L}(\bar{\zeta}_t^{\lambda,n})) + W_1(\mathcal{L}(\bar{\zeta}_t^{\lambda,n}),\mathcal{L}(Z^{\lambda}_t)).
\end{align}
One notices that when $1< \lambda \leq \lambda_{\max}$, the result holds trivially. Thus, we consider the case $0<\lambda \leq 1$, which implies $1/2 < \lambda T \leq 1$.

An upper bound for the first term in \eqref{decomp} is obtained in the Lemma below.

\begin{lemma}\label{convergencepart1} Let Assumption \ref{expressionH}, \ref{iid}, \ref{clc}, and \ref{assum:dissipativity} hold. For any $0< \lambda < \lambda_{\max}$ given in \eqref{eq:definition-lambda-max}, $t \in [nT, (n+1)T]$,
\begin{align*}
W_1(\mathcal{L}(\bar{\theta}^{\lambda}_t),\mathcal{L}(\bar{\zeta}_t^{\lambda,n}))   \leq    \sqrt{\lambda}(e^{-an/2}\bar{C}_{2,1}\mathbb{E}[V_2(\theta_0)]  +\bar{C}_{2,2})^{1/2},
\end{align*}
where $\bar{C}_{2,1}$ and $\bar{C}_{2,2}$ are given in \eqref{barc2}.
\end{lemma}
\begin{proof}
To handle the first term in \eqref{decomp}, we start by establishing an upper bound in Wasserstein-2 distance and the statement follows by noticing $W_1\leq W_2$. By employing synchronous coupling, using \eqref{SGLDprocess} and the definition of $\bar{\zeta}_t^{\lambda,n} $ in Definition \ref{zetaprocess}, one obtains
\begin{align*}
\left| \bar{\zeta}_t^{\lambda,n} - \bar{\theta}^{\lambda}_t\right| \leq \lambda \left| \int_{nT}^t \left[ H(\bar{\theta}^{\lambda}_\floor{s}, X_{\ceil{s}}) - h(\bar{\zeta}_s^{\lambda,n}) \right] d s \right|.
\end{align*}
Then, the triangle inequality leads
\[
\left| \bar{\zeta}_t^{\lambda,n} - \bar{\theta}^{\lambda}_t \right| \leq \lambda \left| \int_{nT}^t \left[ H(\bar{\theta}^{\lambda}_\floor{s},X_{\ceil{s}}) -h(\bar{\theta}^{\lambda}_\floor{s})\right] d s \right| + \lambda \left| \int_{nT}^t \left[h(\bar{\theta}^{\lambda}_\floor{s}) - h(\bar{\zeta}_s^{\lambda,n}) \right] d s \right|.
\]
Taking squares on both sides and the application of Remark \ref{hlip} yield
\[
\left| \bar{\zeta}_t^{\lambda,n} - \bar{\theta}^{\lambda}_t \right|^2 \leq 2\lambda^2 \left| \int_{nT}^t \left[ H(\bar{\theta}^{\lambda}_\floor{s},X_{\ceil{s}}) -h(\bar{\theta}^{\lambda}_\floor{s})\right] d s \right|^2 + 2\lambda L^2 \int_{nT}^t  \left|\bar{\theta}^{\lambda}_\floor{s} - \bar{\zeta}_s^{\lambda,n} \right|^2 d s.
\]
By taking expectations on both sides and by using $(a+b)^2 \leq 2 a^2 + 2 b^2$, for $a, b >0$, % and by using Jensen's inequality,
one obtains
\begin{align*}
\E\left[\left| \bar{\zeta}_t^{\lambda,n} -\bar{\theta}^{\lambda}_t \right|^2\right]
&\leq 2\lambda^2\mathbb{E}\left[ \left| \int_{nT}^t \left[ H(\bar{\theta}^{\lambda}_\floor{s},X_{\ceil{s}}) -h(\bar{\theta}^{\lambda}_\floor{s})\right] d s \right|^2\right] \\
&\quad + 4\lambda L^2  \int_{nT}^t  \mathbb{E}\left[\left|\bar{\theta}^{\lambda}_\floor{s} - \bar{\theta}^{\lambda}_s  \right|^2\right] d s + 4\lambda L^2  \int_{nT}^t  \mathbb{E}\left[\left|\bar{\theta}^{\lambda}_s  - \bar{\zeta}_s^{\lambda,n} \right|^2\right] d s,
\end{align*}
which implies due to $\lambda T \leq 1$ and Lemma \ref{onestepest}
\begin{align}\label{proof:ErrorSplit}
\E\left[\left| \bar{\zeta}_t^{\lambda,n} -\bar{\theta}^{\lambda}_t \right|^2\right]
&\leq  4\lambda L^2   (e^{-a\lambda nT}\bar{\sigma}_Y \mathbb{E}[V_2(\theta_0)] + \tilde{\sigma}_Y) + 4\lambda L^2  \int_{nT}^t  \mathbb{E}\left[\left|\bar{\theta}^{\lambda}_s  - \bar{\zeta}_s^{\lambda,n} \right|^2\right] d s \nonumber\\
&\quad +2\lambda^2\mathbb{E}\left[ \left| \int_{nT}^t \left[ H(\bar{\theta}^{\lambda}_\floor{s},X_{\ceil{s}}) -h(\bar{\theta}^{\lambda}_\floor{s})\right] d s \right|^2\right],
\end{align}

where $\bar{\sigma}_Y$ and $\tilde{\sigma}_Y$ are provided in \eqref{sigmaY}. Next, we bound the last term in \eqref{proof:ErrorSplit} by partitioning the integral. Assume that $nT + K \leq t \leq nT + K + 1$ where $K + 1 \leq T$. Thus we can write
\begin{align*}
 \left| \int_{nT}^t \left[ H(\bar{\theta}^{\lambda}_\floor{s},X_{\ceil{s}}) -h(\bar{\theta}^{\lambda}_\floor{s})\right] d s \right| = \left| \sum_{k=1}^{K} I_k + R_K \right|
\end{align*}
where
\[
I_k = H(\bar{\theta}^{\lambda}_{nT + k-1},X_{nT + k}) -h(\bar{\theta}^{\lambda}_{nT + k-1}), \quad R_K =(t-(nT+K))(H(\bar{\theta}^{\lambda}_{nT+K},X_{nT+K+1}) -h(\bar{\theta}^{\lambda}_{nT+K})).
\]
Taking squares of both sides
\begin{align*}
\left| \sum_{k=1}^{K} I_k + R_K \right|^2 = \sum_{k=1}^K | I_k |^2 + 2 \sum_{k=2}^K \sum_{j = 1}^{k-1} \langle I_k, I_j \rangle +2 \sum_{k=1}^K \langle I_k, R_K \rangle + |R_K|^2,
\end{align*}
Finally, we take expectations of both sides. Define the filtration $\mathcal{H}_t = \mathcal{F}^{\lambda}_{\infty} \vee \mathcal{G}_{\lfloor t \rfloor}$. We first note that for any $k =2, \dots, K$, $j = 1, \dots, k-1$,
\begin{align*}
&\E \langle I_k, I_j \rangle\\
 &= \E\left[ \E [\langle I_k, I_j \rangle | \mathcal{H}_{nT+k-1} ] \right], \\
&= \E\left[ \E \left[\left\langle  H(\bar{\theta}^{\lambda}_{nT + k-1},X_{nT + k}) -h(\bar{\theta}^{\lambda}_{nT + k-1}), \left. H(\bar{\theta}^{\lambda}_{nT + j-1},X_{nT + j}) -h(\bar{\theta}^{\lambda}_{nT + j-1}) \right\rangle \right|  \mathcal{H}_{nT+k-1} \right] \right], \\
& = \E\left[ \left\langle  \E \left[\left.  H(\bar{\theta}^{\lambda}_{nT + k-1},X_{nT + k}) -h(\bar{\theta}^{\lambda}_{nT + k-1})\right|  \mathcal{H}_{nT+k-1} \right],  H(\bar{\theta}^{\lambda}_{nT + j-1},X_{nT + j}) -h(\bar{\theta}^{\lambda}_{nT + j-1}) \right\rangle  \right], \\
&= 0.
\end{align*}
By the same argument $\E \langle I_k, R_K\rangle = 0$ for all $1 \leq k \leq K$. Therefore, the last term of \eqref{proof:ErrorSplit} is bounded as
\begin{align*}
2\lambda^2\mathbb{E}\left[ \left| \int_{nT}^t \left[ H(\bar{\theta}^{\lambda}_\floor{s},X_{\ceil{s}}) -h(\bar{\theta}^{\lambda}_\floor{s})\right] d s \right|^2\right] &=  2\lambda^2 \sum_{k=1}^K \E\left[ |I_k|^2\right] + 2\lambda^2 \E\left[ |R_K|^2\right]\\
&\leq 2\lambda  (e^{-a\lambda nT}\bar{\sigma}_Z\mathbb{E}[V_2(\theta_0)]+\tilde{\sigma}_Z),
\end{align*}
where the last inequality holds due to Lemma~\ref{lem:boundedVariance} and $\bar{\sigma}_Z$ and $\tilde{\sigma}_Z$ are provided in \eqref{sigmaZ}. Therefore, the bound \eqref{proof:ErrorSplit} becomes
\begin{align*}
\E\left[\left| \bar{\zeta}_t^{\lambda,n} -\bar{\theta}^{\lambda}_t \right|^2\right]
&\leq   4\lambda L^2  \int_{nT}^t  \mathbb{E}\left[\left|\bar{\theta}^{\lambda}_s  - \bar{\zeta}_s^{\lambda,n} \right|^2\right] d s \nonumber\\
&\quad +4\lambda e^{-a\lambda nT}( L^2   \bar{\sigma}_Y+\bar{\sigma}_Z) \mathbb{E}[V_2(\theta_0)] + 4\lambda (L^2 \tilde{\sigma}_Y+\tilde{\sigma}_Z) ,
\end{align*}
Using Gr\"onwall's inequality yields
\[
\E\left[\left| \bar{\zeta}_t^{\lambda,n} -\bar{\theta}^{\lambda}_t \right|^2\right] \leq 4\lambda e^{4L^2  }\left[e^{-a\lambda nT}( L^2   \bar{\sigma}_Y+\bar{\sigma}_Z) \mathbb{E}[V_2(\theta_0)] +  (L^2 \tilde{\sigma}_Y+\tilde{\sigma}_Z)\right],
\]
which implies by $\lambda T \geq 1/2$,
\begin{equation}\label{vibd1}
W_2^2(\calL(\bar{\theta}^{\lambda}_t),\calL(\bar{\zeta}_t^{\lambda,n})) \leq \E\left[\left| \bar{\zeta}_t^{\lambda,n} - \bar{\theta}^{\lambda}_t \right|^2\right]  \leq \lambda(e^{-an/2}\bar{C}_{2,1}\mathbb{E}[V_2(\theta_0)]  +\bar{C}_{2,2}) ,
\end{equation}
where
\begin{equation}\label{barc2}
\bar{C}_{2,1} = 4e^{4L^2  } ( L^2   \bar{\sigma}_Y+\bar{\sigma}_Z)  , \quad \bar{C}_{2,2} =  4e^{4L^2  }(L^2 \tilde{\sigma}_Y+\tilde{\sigma}_Z)
\end{equation}
with $ \bar{\sigma}_Y$, $ \tilde{\sigma}_Y$ provided in \eqref{sigmaY} and $\bar{\sigma}_Z$, $\tilde{\sigma}_Z$ given in \eqref{sigmaZ}.
\end{proof}
Then, the following Lemma provides the bound for the second term in \eqref{decomp}.
\begin{lemma}\label{convergencepart2} Let Assumption \ref{expressionH}, \ref{iid}, \ref{clc} and \ref{assum:dissipativity} hold. For any $0< \lambda < \lambda_{\max}$ given in \eqref{eq:definition-lambda-max}, $t \in [nT, (n+1)T]$,
\begin{align*}
W_1(\mathcal{L}(\bar{\zeta}_t^{\lambda,n}),\calL(Z_t^\lambda)) \leq \sqrt{\lambda}(e^{-\min\{\dot{c},a/2\}n/2}\bar{C}_{2,3}\mathbb{E}[V_4(\theta_0)] +\bar{C}_{2,4}),
\end{align*}
where $\bar{C}_{2,3} $, $\bar{C}_{2,4}$ is given in \eqref{mainthmc2}.
\end{lemma}
\begin{proof}
To upper bound the second term $W_1(\mathcal{L}(\bar{\zeta}_t^{\lambda,n}),\calL(Z_t^\lambda))$ in \eqref{decomp}, we adapt the proof from Lemma~3.28 in \cite{nonconvex}. By Proposition \ref{contr}, Corollary \ref{corr:Moments}, Lemma \ref{zetaprocessmoment} and \ref{convergencepart1}, one obtains
\begin{align*}
&W_1(\mathcal{L}(\bar{\zeta}_t^{\lambda,n}),\calL(Z_t^\lambda)) \\
&\leq \sum_{k=1}^n W_1(\calL(\bar{\zeta}_t^{\lambda,k}),\calL(\bar{\zeta}_t^{\lambda,k-1})), \\
&\leq \sum_{k=1}^n w_{1,2}(\calL(\zeta^{kT,\bar{\theta}^{\lambda}_{kT}, \lambda}_t ),\calL(\zeta^{kT,\bar{\zeta}_{kT}^{\lambda,k-1}, \lambda}_t)) \\
&\leq \hat{c} \sum_{k=1}^n \exp(-\dot{c} (n-k)) w_{1,2}(\calL(\bar{\theta}^{\lambda}_{kT}),\calL(\bar{\zeta}_{kT}^{\lambda,k-1})) \\
%&\leq  \frac{\hat{c} }{1 - \exp(-\dot{c}))} \max_{1 \leq k \leq n} w_{1,2}(\calL(\bar{\theta}^{\lambda}_{kT}),\calL(\bar{\zeta}_{kT}^{\lambda,k-1})) \\
&\leq \hat{c} \sum_{k=1}^n \exp(-\dot{c} (n-k))  W_{2}(\calL(\bar{\theta}^{\lambda}_{kT}),\calL(\bar{\zeta}_{kT}^{\lambda,k-1}))  \left[1 + \left\lbrace \E[V_4(\bar{\theta}^{\lambda}_{kT})]\right\rbrace^{1/2} + \left\lbrace\E[V_4(\bar{\zeta}_{kT}^{\lambda,k-1})] \right\rbrace^{1/2}\right] \\
&\leq (\sqrt{\lambda})^{-1}\hat{c} \sum_{k=1}^n \exp(-\dot{c} (n-k))  W^2_{2}(\calL(\bar{\theta}^{\lambda}_{kT}),\calL(\bar{\zeta}_{kT}^{\lambda,k-1}))\\
&\quad +3\sqrt{\lambda}\hat{c} \sum_{k=1}^n \exp(-\dot{c} (n-k))  \left[1 + \E[V_4(\bar{\theta}^{\lambda}_{kT})] +\E[V_4(\bar{\zeta}_{kT}^{\lambda,k-1})] \right] \\
& \leq  \sqrt{\lambda}e^{-\min\{\dot{c},a/2\}n}n\hat{c} (e^{\min\{\dot{c},a/2\}}\bar{C}_{2,1}\mathbb{E}[V_2(\theta_0)]  +12\mathbb{E}[V_4(\theta_0)] )\\
&\quad + \sqrt{\lambda}\frac{\hat{c} }{1 - \exp(-\dot{c})}(\bar{C}_{2,2}+12c_3(\lambda_{\max}+a^{-1})+9\mathrm{v}_5(\overline{M}_4)+15)\\
&\leq \sqrt{\lambda}(e^{-\min\{\dot{c},a/2\}n/2}\bar{C}_{2,3}\mathbb{E}[V_4(\theta_0)] +\bar{C}_{2,4})
\end{align*}
where the last inequality holds due to $e^{-\alpha n}(n+1) \leq 1+\alpha^{-1}$, for $\alpha>0$, and we take $\alpha = \min\{\dot{c},a/2\}/2$, moreover,
\begin{align}\label{mainthmc2}
\begin{split}
\bar{C}_{2,3}& =\hat{c}\left(1+\frac{2}{\min\{\dot{c},a/2\}}\right) (e^{\min\{\dot{c},a/2\}}\bar{C}_{2,1}  +12)\\
\bar{C}_{2,4}& = \frac{\hat{c} }{1 - \exp(-\dot{c})}(\bar{C}_{2,2}+12c_3(\lambda_{\max}+a^{-1})+9\mathrm{v}_5(\overline{M}_4)+15)
\end{split}
\end{align}
with $\bar{C}_{2,1} $, $\bar{C}_{2,2}$ given in \ref{barc2}, $\hat{c} $, $\dot{c}$ given in Lemma \ref{contractionconst}, $c_3$ is given in \eqref{4thmomentconst} and $\overline{M}_4$ given in \eqref{Mp}.
\end{proof}

By using similar arguments as in Lemma \ref{convergencepart2}, an analogous result can be obtained in $W_2$ distance, which is given in the following corollary.
\begin{corollary}\label{convergencepart2w2} Let Assumption \ref{expressionH}, \ref{iid}, \ref{clc} and \ref{assum:dissipativity} hold.  For any $0< \lambda < \lambda_{\max}$ given in \eqref{eq:definition-lambda-max}, $t \in [nT, (n+1)T]$,
\begin{align*}
W_2(\mathcal{L}(\bar{\zeta}_t^{\lambda,n}),\calL(Z_t^\lambda)) \leq \lambda^{1/4}(e^{-\min\{\dot{c},a/2\}n/4}\bar{C}^*_{2,3}\mathbb{E}^{1/2}[V_4(\theta_0)] +\bar{C}^*_{2,4}),
\end{align*}
where $\bar{C}^*_{2,3} $, $\bar{C}^*_{2,4}$ is given in \eqref{mainthmcstar2}.
\end{corollary}
\begin{proof}
One notices that $W_2 \leq \sqrt{2w_{1,2}}$, then one writes
\begin{align*}%\label{vibd2}
&W_2(\mathcal{L}(\bar{\zeta}_t^{\lambda,n}),\calL(Z_t^\lambda))  \nonumber\\
&\leq \sum_{k=1}^n W_2(\calL(\bar{\zeta}_t^{\lambda,k}),\calL(\bar{\zeta}_t^{\lambda,k-1}))  \nonumber \\
%&\leq \sum_{k=1}^n W_1(\calL(\bar{\zeta}_t^{\lambda,k}),\calL(\bar{\zeta}_t^{\lambda,k-1})), \\
&\leq \sum_{k=1}^n \sqrt{2}w^{1/2}_{1,2}(\calL(\zeta^{kT,\bar{\theta}^{\lambda}_{kT}, \lambda}_t ),\calL(\zeta^{kT,\bar{\zeta}_{kT}^{\lambda,k-1}, \lambda}_t)) \nonumber\\
&\leq \sqrt{2\hat{c}}\sum_{k=1}^n \exp(-\dot{c} (n-k)/2)  W^{1/2}_2(\calL(\bar{\theta}^{\lambda}_{kT}),\calL(\bar{\zeta}_{kT}^{\lambda,k-1}))  \left[1 + \left\lbrace \E[V_4(\bar{\theta}^{\lambda}_{kT})]\right\rbrace^{1/2} + \left\lbrace\E[V_4(\bar{\zeta}_{kT}^{\lambda,k-1})] \right\rbrace^{1/2}\right]^{1/2} \nonumber \\
&\leq \lambda^{-1/4}\sqrt{2\hat{c}} \sum_{k=1}^n \exp(-\dot{c} (n-k)/2)  W_2(\calL(\bar{\theta}^{\lambda}_{kT}),\calL(\bar{\zeta}_{kT}^{\lambda,k-1})) \nonumber\\
&\quad + \lambda^{ 1/4}\sqrt{2\hat{c}} \sum_{k=1}^n \exp(-\dot{c} (n-k)/2) \left[1 + \left\lbrace \E[V_4(\bar{\theta}^{\lambda}_{kT})]\right\rbrace^{1/2} + \left\lbrace\E[V_4(\bar{\zeta}_{kT}^{\lambda,k-1})] \right\rbrace^{1/2}\right]\nonumber\\
& \leq  \sqrt{2\hat{c}}\lambda^{1/4}e^{-\min\{\dot{c},a/2\}n/2}n (e^{\min\{\dot{c},a/2\}/2}\bar{C}^{1/2}_{2,1}\mathbb{E}^{1/2}[V_2(\theta_0)]  +2\sqrt{2}\mathbb{E}^{1/2}[V_4(\theta_0)] ) \nonumber\\
&\quad + \sqrt{2\hat{c}}\lambda^{1/4}\frac{1 }{1 - \exp(-\dot{c}/2)}(\bar{C}^{1/2}_{2,2}+2\sqrt{2c_3}(\lambda_{\max}+a^{-1})^{1/2}+\sqrt{3}\mathrm{v}^{1/2}_5(\overline{M}_4)+\sqrt{15}) \nonumber\\
&\leq \lambda^{1/4}(e^{-\min\{\dot{c},a/2\}n/4}\bar{C}^*_{2,3}\mathbb{E}^{1/2}[V_4(\theta_0)] +\bar{C}^*_{2,4}),
\end{align*}
where
\begin{align}\label{mainthmcstar2}
\begin{split}
\bar{C}^*_{2,3}& =\sqrt{2\hat{c}}\left(1+\frac{4}{\min\{\dot{c},a/2\}}\right) (e^{\min\{\dot{c},a/2\}/2}\bar{C}^{1/2}_{2,1}  +2\sqrt{2})\\
\bar{C}^*_{2,4}& = \frac{\sqrt{2\hat{c}}}{1 - \exp(-\dot{c}/2)}(\bar{C}^{1/2}_{2,2}+2\sqrt{2c_3}(\lambda_{\max}+a^{-1})^{1/2}+\sqrt{3}\mathrm{v}^{1/2}_5(\overline{M}_4)+\sqrt{15}),
\end{split}
\end{align}
with $\bar{C}_{2,1} $, $\bar{C}_{2,2}$ given in \ref{barc2}, $\hat{c} $, $\dot{c}$ given in Lemma \ref{contractionconst}, $c_3$ is given in \eqref{4thmomentconst} and $\overline{M}_4$ given in Lemma \ref{lem:PreLimforDriftY}. This completes the proof.
\end{proof}

Finally, by using the inequality \eqref{decomp} and the results from previous lemmas, one can obtain the non-asymptotic bound between $\bar{\theta}^{\lambda}_t$ and $Z^{\lambda}_t$, $t \in [nT, (n+1)T]$, in $W_1$ distance.

\begin{lemma}\label{convergencepart3} Let Assumption \ref{expressionH}, \ref{iid}, \ref{clc} and \ref{assum:dissipativity} hold. For any $0< \lambda < \lambda_{\max}$ given in \eqref{eq:definition-lambda-max}, $t \in [nT, (n+1)T]$,
\begin{align*}
W_1(\mathcal{L}(\bar{\theta}^{\lambda}_t),\mathcal{L}(Z^{\lambda}_t)) \leq \bar{C}_2\sqrt{\lambda}(e^{-\min\{\dot{c},a/2\}n/2}\mathbb{E}[V_4(\theta_0)] +1),
\end{align*}
where $\bar{C}_2$ is given in \eqref{lemmac2}.
\end{lemma}
\begin{proof}
By using Lemma \ref{convergencepart1} and \ref{convergencepart2}, one obtains
\begin{align*}
&W_1(\mathcal{L}(\bar{\theta}^{\lambda}_t),\mathcal{L}(Z^{\lambda}_t))\\
&\leq W_1(\calL(\bar{\theta}^{\lambda}_t),\calL(\bar{\zeta}_t^{\lambda,n})) +W_1(\mathcal{L}(\bar{\zeta}_t^{\lambda,n}),\calL(Z_t^\lambda)) \\
& \leq  \sqrt{\lambda} (e^{-an /2}\bar{C}_{2,1}^{1/2}\mathbb{E}^{1/2}[V_2(\theta_0)]  +\bar{C}_{2,2}^{1/2}) +\sqrt{\lambda}(e^{-\min\{\dot{c},a/2\}n/2}\bar{C}_{2,3}\mathbb{E}[V_4(\theta_0)] +\bar{C}_{2,4})\\
&\leq  \bar{C}_2\sqrt{\lambda}(e^{-\min\{\dot{c},a/2\}n/2}\mathbb{E}[V_4(\theta_0)] +1),
\end{align*}
where
\begin{equation}\label{lemmac2}
\bar{C}_2 =\bar{C}_{2,1}^{1/2} +\bar{C}_{2,2}^{1/2}+ \bar{C}_{2,3}+\bar{C}_{2,4}.
\end{equation}
\end{proof}
Before proceeding to the proofs of the main results, we provide explicitly the constants $\dot{c}$ and $\hat{c}$ in Proposition \ref{contr}.
\begin{lemma} \label{contractionconst} The contraction constant in Proposition \ref{contr} is given by
$$
\dot{c}=\min\{\bar{\phi}, \bar{c}(p), 4\tilde{c}(p) \dot{\epsilon}\bar{c}(p)\}/2,
$$
where the explicit expressions for $\bar{c}(p)$ and $\tilde{c}(p)$ can be found in Lemma~\ref{lem:PreLimforDriftY} and  $\bar{\phi}$ is given by
\[
\bar{\phi}= \left(\sqrt{4\pi/L} \bar{b}  \exp\left(\left(\bar{b} \sqrt{L}/2 +2/\sqrt{L}\right)^2\right) \right)^{-1} \,.
\]
Furthermore, any $\dot{\epsilon}$ can be chosen which satisfies the following inequality
\[
\dot{\epsilon}  \leq 1 \wedge \left(8\tilde{c}(p) \sqrt{\pi/L}\int_0^{\tilde{b}}\exp\left(\left(s\sqrt{L}/2+2/\sqrt{L}\right)^2\right) \,d s\right)^{-1},
\]
where
$\tilde{b}=\sqrt{2\tilde{c}(p)/\bar{c}(p)-1}$ and $\bar{b} = \sqrt{4\tilde{c}(p)(1+\bar{c}(p))/\bar{c}(p)-1}$. The constant $\hat{c} $ is given as the ratio $C_{11}/C_{10}$, where $C_{11},\,C_{10}$ are given explicitly in \cite[Lemma~3.26]{nonconvex}.
\end{lemma}
\begin{proof}
See \cite[Lemma~3.26]{nonconvex}.
\end{proof}

\noindent{\bf Proof of Theorem \ref{main}} One notes that, by Lemma \ref{convergencepart3} and Proposition \ref{contr}, for $t \in [nT, (n+1)T]$
\begin{align*}
W_1(\mathcal{L}(\bar{\theta}^{\lambda}_t),\pi_{\beta}) &\leq W_1(\mathcal{L}(\bar{\theta}^{\lambda}_t),\mathcal{L}(Z^{\lambda}_t))+W_1(\mathcal{L}(Z^{\lambda}_t),\pi_{\beta})\\
& \leq\bar{C}_2\sqrt{\lambda}(e^{-\min\{\dot{c},a/2\}n/2}\mathbb{E}[V_4(\theta_0)] +1) +\hat{c} e^{-\dot{c}\lambda t} w_{1,2}(\theta_0, \pi_{\beta})\\
& \leq \bar{C}_2\sqrt{\lambda}(e^{-\min\{\dot{c},a/2\}n/2}\mathbb{E}[V_4(\theta_0)] +1) +\hat{c} e^{-\dot{c}\lambda t}\left[1+ \mathbb{E}[V_2(\theta_0)]+\int_{\mathbb{R}^d}V_2(\theta)\pi_{\beta}(d\theta)\right]\\
& \leq 2e^{-\min\{\dot{c},a/2\}n/2} (\lambda_{\max}^{1/2}\bar{C}_2+\hat{c})(1+\mathbb{E}[|\theta_0|^4])\\
&\quad +\hat{c}e^{-\min\{\dot{c},a/2\}n/2}\left[1 +\int_{\mathbb{R}^d}V_2(\theta)\pi_{\beta}(d\theta)\right]+\sqrt{\lambda}\bar{C}_2,
\end{align*}
which implies, for any $n \in \mathbb{N}$
\[
W_1(\mathcal{L}(\theta^{\lambda}_n),\pi_{\beta}) \leq C_1e^{-C_0\lambda n}(1+\mathbb{E}[|\theta_0|^4])+C_2\sqrt{\lambda},
\]
where
\begin{equation}\label{mainthmconst}
C_0 =  \min\{\dot{c},a/2\}/2, \quad C_1 = 2 \left[ (\lambda_{\max}^{1/2}\bar{C}_2+\hat{c})+\hat{c}\left(1+ \int_{\mathbb{R}^d}V_2(\theta)\pi_{\beta}(d\theta)\right)\right], \quad C_2 = \bar{C}_2,
\end{equation}
with $\bar{C}_{2} $ given in \ref{lemmac2}.

\noindent {\bf Proof of Corollary \ref{cw2}} By using \eqref{vibd1} in Lemma \ref{convergencepart1}, Corollary \ref{convergencepart2w2} and Proposition \ref{contr}, one obtains
\begin{align*}
W_2(\mathcal{L}(\bar{\theta}^{\lambda}_t),\pi_{\beta})
&\leq W_2(\mathcal{L}(\bar{\theta}^{\lambda}_t),\mathcal{L}(Z^{\lambda}_t))+W_2(\mathcal{L}(Z^{\lambda}_t),\pi_{\beta})\\
& \leq W_2(\calL(\bar{\theta}^{\lambda}_t),\calL(\bar{\zeta}_t^{\lambda,n})) +W_2(\mathcal{L}(\bar{\zeta}_t^{\lambda,n}),\calL(Z_t^\lambda))+W_2(\mathcal{L}(Z^{\lambda}_t),\pi_{\beta})\\
& \leq \sqrt{\lambda} (e^{-an/2}\bar{C}_{2,1}\mathbb{E}[V_2(\theta_0)]  +\bar{C}_{2,2})^{1/2} \\
&\quad+\lambda^{1/4}(e^{-\min\{\dot{c},a/2\}n/4}\bar{C}^*_{2,3}\mathbb{E}^{1/2}[V_4(\theta_0)] +\bar{C}^*_{2,4}), +\sqrt{2w_{1,2}(\mathcal{L}(Z^{\lambda}_t),\pi_{\beta})}\\
& \leq \lambda^{1/4}\tilde{C}_2 (e^{-\min\{\dot{c},a/2\}n/4}\mathbb{E}[V_4(\theta_0)] +1)+\hat{c}^{1/2} e^{-\dot{c}\lambda t/2} \sqrt{2w_{1,2}(\theta_0, \pi_{\beta})},
\end{align*}
where $\tilde{C}_2 = \lambda_{\max}^{1/4}\bar{C}_{2,1}^{1/2} +\lambda_{\max}^{1/4}\bar{C}_{2,2}^{1/2}+ \bar{C}^*_{2,3} +\bar{C}^*_{2,4} $ and it can be further calculated as
\begin{align*}
W_2(\mathcal{L}(\bar{\theta}^{\lambda}_t),\pi_{\beta})
& \leq \lambda^{1/4}\tilde{C}_2 (e^{-\min\{\dot{c},a/2\}n/4}\mathbb{E}[V_4(\theta_0)] +1) \\
&\quad +\sqrt{2}\hat{c}^{1/2} e^{-\dot{c}\lambda t/2}\left(1+ \mathbb{E}[V_2(\theta_0)]+\int_{\mathbb{R}^d}V_2(\theta)\pi_{\beta}(d\theta)\right)^{1/2}\\
& \leq 2e^{-\min\{\dot{c},a/2\}n/4} (\lambda_{\max}^{1/4}\tilde{C}_2+\sqrt{2}\hat{c}^{1/2})(1+\mathbb{E}[|\theta_0|^4])\\
&\quad +\sqrt{2}\hat{c}^{1/2}e^{-\min\{\dot{c},a/2\}n/4}\left[1 +\int_{\mathbb{R}^d}V_2(\theta)\pi_{\beta}(d\theta)\right]+ \lambda^{1/4}\tilde{C}_2,
\end{align*}
Finally, one obtains
\[
W_2(\mathcal{L}(\theta^{\lambda}_n),\pi_{\beta}) \leq C_4 e^{-C_3\lambda n}\E[|\theta_{0}|^{4}+1] +C_5\lambda^{1/4}
\]
where
\begin{equation}\label{cw2const}
C_3 =  \min\{\dot{c},a/2\}/4, \quad C_4 = 2 \left[ (\lambda_{\max}^{1/4}\tilde{C}_2+\sqrt{2}\hat{c}^{1/2})+\hat{c}^{1/2}\left(1+ \int_{\mathbb{R}^d}V_2(\theta)\pi_{\beta}(d\theta)\right)\right], \quad C_5 = \tilde{C}_2.
\end{equation}

\noindent {\bf Proof of Corollary \ref{eer}} To obtain an upper bound for the expected excess risk $\mathbb{E}[U(\hat{\theta})] - \inf_{\theta \in \mathbb{R}^d} U(\theta) $, one considers the following splitting
\begin{equation}\label{eersplitting}
 \mathbb{E}[U(\hat{\theta})] - \inf_{\theta \in \mathbb{R}^d} U(\theta) = \left( \mathbb{E}[U(\hat{\theta})] -  \mathbb{E}[U(Z_{\infty})]\right) + \left( \mathbb{E}[U(Z_{\infty})]- \inf_{\theta \in \mathbb{R}^d} U(\theta) \right),
\end{equation}
where $\hat{\theta} =  \theta^{\lambda}_n$ and $Z_{\infty}\sim \pi_{\beta}$ with $\pi_{\beta}(\theta) = \exp(-\beta U(\theta))$ for all $\theta \in \mathbb{R}^d$. By using \cite[Lemma~3.5]{raginsky}, Lemma \ref{lem:moment_SGLD_2p}, \ref{zt2ndmoment} and Corollary \ref{cw2}, the first term on the RHS of \eqref{eersplitting} can be bounded by
\begin{align*}
&\mathbb{E}[U(\hat{\theta})] -  \mathbb{E}[U(Z_{\infty})] \\
& \leq \left(L(\mathbb{E}\left[ |\theta_0|^2\right] +(c_1+\E[K_1^2(X_0)]/a) (\lambda_{\max}+a^{-1}))^{1/2}+|h(0)|\right)W_2(\mathcal{L}(\theta^{\lambda}_n),\pi_{\beta})\\
&\leq \left(L(\mathbb{E}\left[ |\theta_0|^2\right]+(c_1+\E[K_1^2(X_0)]/a) (\lambda_{\max}+a^{-1}))^{1/2}+|h(0)|\right)\left(C_4 e^{-C_3\lambda n}\E[|\theta_{0}|^{4}+1] +C_5\lambda^{1/4}\right)\\
&\leq \hat{C}_1e^{-\hat{C}_0\lambda n}+\hat{C_2}\lambda^{1/4},
\end{align*}
where
\begin{align}\label{eerconst}
\begin{split}
\hat{C}_0 &= C_3,\\
\hat{C}_1 & = C_4 \left(L(\mathbb{E}\left[ |\theta_0|^2\right]+(c_1+\E[K_1^2(X_0)]/a) (\lambda_{\max}+a^{-1}))^{1/2}+|h(0)|\right)\E[|\theta_{0}|^{4}+1], \\
\hat{C}_2 & = C_5 \left(L(\mathbb{E}\left[ |\theta_0|^2\right]+(c_1+\E[K_1^2(X_0)]/a) (\lambda_{\max}+a^{-1}))^{1/2}+|h(0)|\right),
\end{split}
\end{align}
with $C_3, C_4, C_5$ given in \eqref{cw2const} and $c_1$ given in \eqref{2ndmomentconst}. Moreover, the second term on the RHS of \eqref{eersplitting} can be estimated by using \cite[Proposition~3.4]{raginsky}, which gives,
\[
\mathbb{E}[U(Z_{\infty})]- \inf_{\theta \in \mathbb{R}^d} U(\theta) \leq  \frac{\hat{C}_3}{\beta},
\]
where
\begin{equation}\label{eerconst2}
\hat{C}_3 = \frac{d}{2}\log\left(\frac{e\beta L}{ad}\left(\frac{2d}{\beta}+2b+\frac{\E[K_1^2(X_0)]}{a}\right)\right).
\end{equation}
Finally, one obtains
\[
\mathbb{E}[U(\hat{\theta})] - \inf_{\theta \in \mathbb{R}^d} U(\theta)  \leq  \hat{C}_1e^{-\hat{C}_0\lambda n}+\hat{C_2}\lambda^{1/4}+\hat{C}_3/\beta.
\]

\section{Proof of the main results: convex case}\label{poconv}
The analysis of the convergence results in the convex case, i.e. Theorem \ref{mainconvex}, relies on the properties of the LMC algorithm, known also as the unadjusted Langevin algorithm (ULA). The LMC algorithm associated with SDE \eqref{sde} is given explicitly by, for any $ n\in\mathbb{N}$,
\begin{equation}\label{ULA}
 \dot{\theta}^{\lambda}_{n+1}:=\dot{\theta}^{\lambda}_n-\lambda h(\dot{\theta}^{\lambda}_n)+\sqrt{2\beta^{-1}\lambda}\xi_{n+1}, \quad \dot{\theta}^{\lambda}_0:=\theta_0.
\end{equation}

For  $0<\lambda <\bar{\lambda}_{\max}$, the Markov kernel $\dot{R}_{\lambda}$ associated with \eqref{ULA} is given by, for all $A \in \mathcal{B}(\mathbb{R}^{d})$ and $\theta \in \mathbb{R}^{d}$,

\[
\dot{R}_{\lambda}(\theta,A) = \int_{A}{ (4\beta^{-1}\pi \lambda)^{-d/2} \exp\left(-\beta(4\lambda)^{-1}\left|y-\theta + \lambda  h(\theta) \right|^2  \right)d y.}
\]
In this section, the moment estimates of the SDE \eqref{sde}, the LMC algorithm \eqref{ULA} and the SGLD algorithm \eqref{SGLD} are presented which contribute to the analysis of the convergence resuts.
\subsection{Preliminary estimates}
%By \cite[Theorem~2.1.8]{nesterov},
Under Assumptions \ref{convex} and \ref{convexG}, $U$ has a unique minimizer $\theta^* \in \mathbb{R}^d$. Denote by $(P_t)_{t \geq 0}$ the semigroup associated with SDE \eqref{sde}. The statements below provide a moment bound and a convergence result for SDE \eqref{sde}.

\begin{lemma}[Proposition 1 in \cite{aew}]\label{sdemoment} Let Assumptions \ref{expressionH}, \ref{iid}, \ref{clc}, \ref{convex} and \ref{convexG} hold.
\begin{enumerate}
\item[(i)] For all $t> 0$ and $y \in \Rd$,
\begin{align*}
\int_{\Rd}|y-\theta^*|^2P_t(\theta,d y)\leq |\theta-\theta^*|^2 e^{-2\hat{a}t}+(d/(\hat{a}\beta))(1- e^{-2\hat{a}t}).
\end{align*}
\item[(ii)] The stationary distribution $\pi_{\beta}$ satisfies
\begin{align*}
\int_{\Rd}|y-\theta^*|^2\pi_{\beta}(d y) \leq d/(\hat{a}\beta).
\end{align*}
%\item[(iii)] For all $t> 0$ and $y \in \Rd$,
%\begin{align*}
%W_2(\delta_x P_t, \pi_{\beta}) \leq e^{-\hat{a}t}(|\theta-\theta^*|+\sqrt{d/(\hat{a}\beta)}).
%\end{align*}
\end{enumerate}
\end{lemma}

The following lemma provides moment estimates for $(\dot{\theta}_n)_{n \in \mathbb{N}}$ and it states that $\dot{R}_{\lambda}$ admits an invariant measure $\pi_{\lambda}$ which may differ from $\pi_{\beta}$.
\begin{lemma}\label{ulamoment} Let Assumptions \ref{expressionH}, \ref{iid}, \ref{convex} and \ref{convexG} hold. Then, for all $0<\lambda <\bar{\lambda}_{\max}$ given in \eqref{lambdamaxconv}, one obtains:
\begin{enumerate}
\item[(i)] For all $t> 0$ and $\theta \in \Rd$,
\begin{align*}
\int_{\Rd}|y-\theta^*|^2\dot{R}_{\lambda}^n(\theta,d y)\leq (1-2\hat{a}^*\lambda)^n|\theta-\theta^*|^2 +(d/(\hat{a}^*\beta))(1- (1-2\hat{a}^*\lambda)^n).
\end{align*}
\item[(ii)] The Markov kernel $\dot{R}_{\lambda}$ has a unique stationary distribution $\pi_{\lambda}$ and it satisfies
\begin{align*}
\int_{\Rd}|\theta-\theta^*|^2\pi_{\lambda}(d \theta) \leq d/(\hat{a}^*\beta).
\end{align*}
\item[(iii)] For all $n \in\mathbb{R}^d$ and $\theta \in \Rd$,
\begin{align*}
W_2(\delta_{\theta} \dot{R}_{\lambda}^n, \pi_{\lambda}) \leq e^{-\hat{a}^* \lambda n}(|\theta-\theta^*|^2+d/(\hat{a}^*\beta))^{1/2}.
\end{align*}
\end{enumerate}
\end{lemma}
The lemma below presents a second moment bound for $\theta^{\lambda}_n$ in the convex case.
\begin{lemma}\label{2ndbdconv}
Let Assumptions \ref{expressionH}, \ref{iid}, \ref{convex} hold. For any  $0<\lambda <\bar{\lambda}_{\max}$ given in \eqref{lambdamaxconv},
\[
\mathbb{E} \left[\left|\theta^{\lambda}_n- \theta^*\right|^2\right] \leq   (1-\hat{a}\lambda)^n\E\left[|\theta_0- \theta^*|^2\right] + \bar{c}_4\hat{a}^{-1},
\]
where
\begin{equation}\label{c4}
\bar{c}_4 = 32\mathbb{E}\left[K_1^2(X_0)\right]\hat{a}_1^{-1} +9\bar{\lambda}_{\max}(L_1^2\mathbb{E}\left[K_{\rho}(X_0)\right]|\theta^*|^2 +L_2^2\mathbb{E}\left[K_{\rho}(X_0)\right]+\mathbb{E}\left[F_*^2(X_0)\right])+2d\beta^{-1}.
\end{equation}
This implies $\sup_n \mathbb{E} \left[\left|\theta^{\lambda}_{n+1}- \theta^*\right|^2\right] \leq \E\left[|\theta_0- \theta^*|^2\right]+\bar{c}_4\hat{a}^{-1}<\infty$. Furthermore, if $\rho = 0$ in Assumption \ref{expressionH}, the result holds  for $\lambda \in \min\{1/2(\hat{a}+L) , 1/(6L_1)\}$ with $\hat{a} = \hat{a}_1 +\hat{a}_2 $.
\end{lemma}
\begin{proof} By using \eqref{SGLD}, one writes, for any $n \in \mathbb{N}$,
\begin{align*}
|\theta^{\lambda}_{n+1}- \theta^*|^2
&= |\theta^{\lambda}_n- \theta^*|^2 +2\left\langle \theta^{\lambda}_n- \theta^*, -\lambda H(\theta^{\lambda}_n,X_{n+1})+\sqrt{2\beta^{-1}\lambda}\xi_{n+1}\right\rangle\\
&\quad +|-\lambda H(\theta^{\lambda}_n,X_{n+1})+\sqrt{2\beta^{-1}\lambda}\xi_{n+1}|^2\\
& = |\theta^{\lambda}_n- \theta^*|^2 -2\lambda\left\langle \theta^{\lambda}_n- \theta^*,  H(\theta^{\lambda}_n,X_{n+1})- H(\theta^*,X_{n+1})\right\rangle\\
&\quad -2\lambda\left\langle \theta^{\lambda}_n- \theta^*,   H(\theta^*,X_{n+1})\right\rangle +2\left\langle \theta^{\lambda}_n- \theta^*, \sqrt{2\beta^{-1}\lambda}\xi_{n+1}\right\rangle\\
&\quad +\lambda^2| H(\theta^{\lambda}_n,X_{n+1})|^2 -2\lambda\left\langle H(\theta^{\lambda}_n,X_{n+1}), \sqrt{2\beta^{-1}\lambda}\xi_{n+1}\right\rangle +2\beta^{-1}\lambda| \xi_{n+1}|^2.
\end{align*}
Taking conditional expectation on both sides and by using Remark \ref{growth} and Assumption \ref{expressionH}, \ref{convex} yield
\begin{align}\label{specialrho}
&\E\left[ \left. |\theta^{\lambda}_{n+1}- \theta^*|^2 \right|\theta^{\lambda}_n \right] \nonumber\\
& = |\theta^{\lambda}_n- \theta^*|^2 -2\lambda\E\left[ \left. \left\langle \theta^{\lambda}_n- \theta^*, F(\theta^{\lambda}_n,X_{n+1})-  F(\theta^*,X_{n+1})\right\rangle \right|\theta^{\lambda}_n \right] \nonumber\\
&\quad -2\lambda\E\left[ \left. \left\langle \theta^{\lambda}_n- \theta^*,  G(\theta^{\lambda}_n,X_{n+1}) - G(\theta^*,X_{n+1})\right\rangle \right|\theta^{\lambda}_n \right] \nonumber\\
&\quad -2\lambda\left\langle \theta^{\lambda}_n- \theta^*,   h(\theta^*)\right\rangle+\lambda^2\E\left[ \left.  | H(\theta^{\lambda}_n,X_{n+1})|^2 \right|\theta^{\lambda}_n \right] +2d\beta^{-1}\lambda\\
& \leq |\theta^{\lambda}_n- \theta^*|^2 - 2\lambda\hat{a}_1 |\theta^{\lambda}_n- \theta^*|^2+4\lambda \mathbb{E}\left[K_1(X_0)\right]|\theta^{\lambda}_n- \theta^*| \nonumber\\
&\quad +\lambda^2\E\left[ \left. \left( (1+|X_{n+1}|)^{\rho+1}(L_1|\theta^{\lambda}_n- \theta^*|+L_1| \theta^*|+L_2) +F_*(X_{n+1})\right)^2 \right|\theta^{\lambda}_n \right] +2d\beta^{-1}\lambda \nonumber\\
& \leq (1-2\hat{a}_1\lambda)|\theta^{\lambda}_n- \theta^*|^2 +4\lambda \mathbb{E}\left[K_1(X_0)\right]|\theta^{\lambda}_n- \theta^*|+2\lambda^2L_1^2 \mathbb{E}\left[K_{\rho}(X_0)\right]|\theta^{\lambda}_n- \theta^*|^2 \nonumber\\
&\quad  +6\lambda^2L_1^2\mathbb{E}\left[K_{\rho}(X_0)\right]|\theta^*|^2 +6\lambda^2L_2^2\mathbb{E}\left[K_{\rho}(X_0)\right]+6\lambda^2\mathbb{E}\left[F_*^2(X_0)\right]+2d\beta^{-1}\lambda \nonumber
\end{align}
%where $\hat{a} = \hat{a}_1 +\hat{a}_2$, and this
which implies, for  $0<\lambda <\bar{\lambda}_{\max}$,
\begin{align*}
\E\left[ \left. |\theta^{\lambda}_{n+1}- \theta^*|^2 \right|\theta^{\lambda}_n \right]
&\leq \left(1-\frac{3}{2}\hat{a}_1\lambda\right)|\theta^{\lambda}_n- \theta^*|^2 +4\lambda \mathbb{E}\left[K_1(X_0)\right]|\theta^{\lambda}_n- \theta^*|\\
&\quad  +6\lambda^2L_1^2\mathbb{E}\left[K_{\rho}(X_0)\right]|\theta^*|^2 +6\lambda^2L_2^2\mathbb{E}\left[K_{\rho}(X_0)\right]+6\lambda^2\mathbb{E}\left[F_*^2(X_0)\right]+2d\beta^{-1}\lambda.
\end{align*}
Then, for $|\theta^{\lambda}_n- \theta^*| >8\mathbb{E}\left[K_1(X_0)\right]\hat{a}_1^{-1}$, one notices that
\[
-\frac{1}{2}\hat{a}_1\lambda|\theta^{\lambda}_n- \theta^*|^2 +4\lambda \mathbb{E}\left[K_1(X_0)\right]|\theta^{\lambda}_n- \theta^*|<0,
\]
and this indicates
\begin{align*}
\E\left[ \left. |\theta^{\lambda}_{n+1}- \theta^*|^2 \right|\theta^{\lambda}_n \right]
&\leq \left(1-\hat{a}_1\lambda\right)|\theta^{\lambda}_n- \theta^*|^2+6\lambda^2L_1^2\mathbb{E}\left[K_{\rho}(X_0)\right]|\theta^*|^2 \\
&\quad  +6\lambda^2L_2^2\mathbb{E}\left[K_{\rho}(X_0)\right]+6\lambda^2\mathbb{E}\left[F_*^2(X_0)\right]+2d\beta^{-1}\lambda.
\end{align*}
Similarly, for $|\theta^{\lambda}_n- \theta^*| \leq  8\mathbb{E}\left[K_1(X_0)\right]\hat{a}_1^{-1}$, one obtains
\begin{align*}
\E\left[ \left. |\theta^{\lambda}_{n+1}- \theta^*|^2 \right|\theta^{\lambda}_n \right]
&\leq \left(1-\frac{3}{2}\hat{a}_1\lambda\right)|\theta^{\lambda}_n- \theta^*|^2 +32\lambda \mathbb{E}\left[K_1^2(X_0)\right]\hat{a}_1^{-1}\\
&\quad  +6\lambda^2L_1^2\mathbb{E}\left[K_{\rho}(X_0)\right]|\theta^*|^2 +6\lambda^2L_2^2\mathbb{E}\left[K_{\rho}(X_0)\right]+6\lambda^2\mathbb{E}\left[F_*^2(X_0)\right]+2d\beta^{-1}\lambda.
\end{align*}
Combining the two cases yields
\[
\E\left[ \left. |\theta^{\lambda}_{n+1}- \theta^*|^2 \right|\theta^{\lambda}_n \right]  \leq  (1-\hat{a}\lambda)|\theta^{\lambda}_n- \theta^*|^2 + \lambda c_4,
\]
where $c_4 = 32\mathbb{E}\left[K_1^2(X_0)\right]\hat{a}_1^{-1} +6\bar{\lambda}_{\max}(L_1^2\mathbb{E}\left[K_{\rho}(X_0)\right]|\theta^*|^2 +L_2^2\mathbb{E}\left[K_{\rho}(X_0)\right]+\mathbb{E}\left[F_*^2(X_0)\right])+2d\beta^{-1}$. The result follows by induction.

Moreover, one observes that when $\rho = 0$ in Assumption \ref{expressionH}, $ F$ is co-coercive, i.e. for any $\theta, \theta' \in \mathbb{R}^d$ and for every $x \in \mathbb{R}^m$
\begin{equation}\label{cocoer}
\left\langle \theta- \theta',  F(\theta, x)-  F(\theta',x )\right\rangle \geq \frac{1}{L_1}| F(\theta, x)- F(\theta',x )|^2.
\end{equation}
Then, by substituting \eqref{cocoer} into \eqref{specialrho}, one obtains
\begin{align*}
\E\left[ \left. |\theta^{\lambda}_{n+1}- \theta^*|^2 \right|\theta^{\lambda}_n \right]
& \leq |\theta^{\lambda}_n- \theta^*|^2 - \frac{3}{2}\lambda\hat{a}_1 |\theta^{\lambda}_n- \theta^*|^2-\frac{\lambda}{2L_1}\E\left[ \left.  | F(\theta^{\lambda}_n,X_{n+1})- F(\theta^*,X_{n+1})|^2 \right|\theta^{\lambda}_n \right] \\
&\quad +4\lambda \mathbb{E}\left[K_1(X_0)\right]|\theta^{\lambda}_n- \theta^*|+\lambda^2\E\left[ \left.  | H(\theta^{\lambda}_n,X_{n+1})|^2 \right|\theta^{\lambda}_n \right] +2d\beta^{-1}\lambda\\
& \leq \left(1 - \frac{3}{2}\lambda\hat{a}_1 \right)|\theta^{\lambda}_n- \theta^*|^2+4\lambda \mathbb{E}\left[K_1(X_0)\right]|\theta^{\lambda}_n- \theta^*|\\
&\quad +\left(3\lambda^2-\frac{\lambda}{2L_1}\right)\E\left[ \left.  |  F(\theta^{\lambda}_n,X_{n+1})- F(\theta^*,X_{n+1})|^2 \right|\theta^{\lambda}_n \right] \\
&\quad+3\lambda^2\E\left[ \left.  | F(\theta^*,X_{n+1})|^2 \right|\theta^{\lambda}_n \right] +3\lambda^2\mathbb{E}\left[K_1^2(X_0)\right]+2d\beta^{-1}\lambda,
\end{align*}
which implies for $\lambda \in \min\{1/2\hat{a}_1 , 1/(6L_1)\}$
\begin{align*}
\E\left[ \left. |\theta^{\lambda}_{n+1}- \theta^*|^2 \right|\theta^{\lambda}_n \right]
& \leq \left(1 - \frac{3}{2}\lambda\hat{a}_1 \right)|\theta^{\lambda}_n- \theta^*|^2+4\lambda \mathbb{E}\left[K_1(X_0)\right]|\theta^{\lambda}_n- \theta^*|\\
&\quad+9\lambda^2L_1^2\mathbb{E}\left[K_{\rho}(X_0)\right]|\theta^*|^2+9\lambda^2L_2^2\mathbb{E}\left[K_{\rho}(X_0)\right] +9\lambda^2\mathbb{E}\left[F_*^2(X_0)\right]+2d\beta^{-1}\lambda.
\end{align*}
By using the same arguments as above, consider the case $|\theta^{\lambda}_n- \theta^*| >8\mathbb{E}\left[K_1(X_0)\right]\hat{a}_1^{-1}$, one notices that
\[
-\frac{1}{2}\hat{a}_1\lambda|\theta^{\lambda}_n- \theta^*|^2 +4\lambda \mathbb{E}\left[K_1(X_0)\right]|\theta^{\lambda}_n- \theta^*|<0,
\]
and this indicates
\begin{align*}
\E\left[ \left. |\theta^{\lambda}_{n+1}- \theta^*|^2 \right|\theta^{\lambda}_n \right]
&\leq \left(1-\hat{a}_1\lambda\right)|\theta^{\lambda}_n- \theta^*|^2+9\lambda^2L_1^2\mathbb{E}\left[K_{\rho}(X_0)\right]|\theta^*|^2 \\
&\quad  +9\lambda^2L_2^2\mathbb{E}\left[K_{\rho}(X_0)\right] +9\lambda^2\mathbb{E}\left[F_*^2(X_0)\right]+2d\beta^{-1}\lambda.
\end{align*}
Similarly, for $|\theta^{\lambda}_n- \theta^*| \leq  8\mathbb{E}\left[K_1(X_0)\right]\hat{a}_1^{-1}$, one obtains
\begin{align*}
\E\left[ \left. |\theta^{\lambda}_{n+1}- \theta^*|^2 \right|\theta^{\lambda}_n \right]
&\leq \left(1-\frac{3}{2}\hat{a}_1\lambda\right)|\theta^{\lambda}_n- \theta^*|^2 +32\lambda \mathbb{E}\left[K_1^2(X_0)\right]\hat{a}_1^{-1}\\
&\quad  +9\lambda^2L_1^2\mathbb{E}\left[K_{\rho}(X_0)\right]|\theta^*|^2+9\lambda^2L_2^2\mathbb{E}\left[K_{\rho}(X_0)\right] +9\lambda^2\mathbb{E}\left[F_*^2(X_0)\right]+2d\beta^{-1}\lambda.
\end{align*}
Combining the two cases yields
\[
\E\left[ \left. |\theta^{\lambda}_{n+1}- \theta^*|^2 \right|\theta^{\lambda}_n \right]  \leq  (1-\hat{a}\lambda)|\theta^{\lambda}_n- \theta^*|^2 + \lambda \bar{c}_4,
\]
where $\bar{c}_4 = 32\mathbb{E}\left[K_1^2(X_0)\right]\hat{a}_1^{-1} +9\bar{\lambda}_{\max}(L_1^2\mathbb{E}\left[K_{\rho}(X_0)\right]|\theta^*|^2 +L_2^2\mathbb{E}\left[K_{\rho}(X_0)\right]+\mathbb{E}\left[F_*^2(X_0)\right])+2d\beta^{-1}$.
\end{proof}

\subsection{Convergence results}
We aim to establish the non-asymptotic bound in Wasserstein-2 distance between $\mathcal{L}(\theta^{\lambda}_n)$ and $\pi_{\beta}$. To achieve this, we consider the following decomposition:
\begin{align}\label{decompconv}
W_2(\mathcal{L}(\theta^{\lambda}_n),\pi_{\beta})
&\leq W_2(\mathcal{L}(\theta^{\lambda}_n),\mathcal{L}(\dot{\theta}^{\lambda}_n)) +W_2 (\mathcal{L}(\dot{\theta}^{\lambda}_n), \pi_{\lambda}) +W_2(\pi_{\lambda}, \pi_{\beta}).
\end{align}
The lemma presented below provides the non-asymptotic estimates for the last two terms in \eqref{decompconv}.
\begin{theorem} {\cite[Corollary 7]{aew}} \label{converconvex}
	Let Assumptions \ref{expressionH}, \ref{iid}, \ref{clc}, \ref{convex} and \ref{convexG} hold. Then, for any  $0<\lambda <\bar{\lambda}_{\max}$ given in \eqref{lambdamaxconv}, the Markov chain $(\dot{\theta}^{\lambda}_n)_{n \in \mathbb{N}}$ admits an invariant measure $\pi_{\lambda}$ such that, for all $n \in \mathbb{N}$,
\[
W_2 (\mathcal{L}(\dot{\theta}^{\lambda}_n),\pi_{\lambda})\leq \bar{C}_7 e^{-\hat{a}^*\lambda n},
\]
where $\bar{C}_7 = (|\theta_0-\theta|^2+d/\hat{a}^*\beta)^{1/2}$ is given in Lemma~\ref{ulamoment} (iii) with $\hat{a}^*=\hat{a}L/(\hat{a}+L) $. Furthermore,
\begin{align*}
W_2(\pi_{\beta},\pi_{\lambda})\leq \bar{C}_{8,1}\sqrt{\lambda},
\end{align*}
where
\begin{equation}\label{converconst1}
\bar{C}_{8,1}=\left(d L^2 (\hat{a}^*\beta)^{-1}(2\lambda+(\hat{a}^*)^{-1})(1+\tfrac{1}{12}\lambda^2L^2+\tfrac{1}{2}L^2\lambda /\hat{a}) \right)^{1/2}.
\end{equation}
\end{theorem}
The non-asymptotic estimate for the first term in \eqref{decompconv} is provided in the following lemma.
\begin{lemma}\label{converconvex2} Let Assumptions \ref{expressionH}, \ref{iid}, \ref{clc}, \ref{convex} and \ref{convexG} hold. For any  $0<\lambda <\bar{\lambda}_{\max}$ given in \eqref{lambdamaxconv}, one obtains
\[
W_2 (\mathcal{L}(\dot{\theta}^{\lambda}_n), \mathcal{L}(\theta_n^{\lambda})) \leq  \bar{C}_{8,2}  \sqrt{\lambda},
\]
where
\begin{align}\label{converconst2}
\begin{split}
\bar{C}_{8,2} &= \sqrt{c_5/2\hat{a}^*}\\
c_5 &=(8L^2+16L_1^2\mathbb{E}\left[K_{\rho}(X_0)\right]) ( \E\left[|\theta_0|^2\right]+\hat{a}^{-1}\bar{c}_4 ) \\
&\quad +(8L^2+40L_1^2\mathbb{E}\left[K_{\rho}(X_0)\right])| \theta^*|^2+24L_2^2\mathbb{E}\left[K_{\rho}(X_0)\right]+24\mathbb{E}\left[F_*^2(X_0)\right].
\end{split}
\end{align}
\end{lemma}
\begin{proof} By using synchronous coupling for the algorithms \eqref{ULA} and \eqref{SGLD}, one obtains
\begin{align*}
| \dot{\theta}^{\lambda}_{n+1} -\theta^{\lambda}_{n+1}|^2
&= | \dot{\theta}^{\lambda}_n -\theta^{\lambda}_n -\lambda (h(\dot{\theta}^{\lambda}_n) -H(\theta^{\lambda}_n, X_{n+1}))|^2\\
&= | \dot{\theta}^{\lambda}_n -\theta^{\lambda}_n|^2 - 2\lambda \langle \dot{\theta}^{\lambda}_n -\theta^{\lambda}_n, h(\dot{\theta}^{\lambda}_n) -H(\theta^{\lambda}_n, X_{n+1}) \rangle +\lambda^2 |h(\dot{\theta}^{\lambda}_n) -H(\theta^{\lambda}_n, X_{n+1})|^2\\
& \leq | \dot{\theta}^{\lambda}_n -\theta^{\lambda}_n|^2 - 2\lambda \langle \dot{\theta}^{\lambda}_n -\theta^{\lambda}_n, h(\dot{\theta}^{\lambda}_n) -h(\theta^{\lambda}_n) \rangle - 2\lambda \langle \dot{\theta}^{\lambda}_n -\theta^{\lambda}_n, h(\theta^{\lambda}_n)  -H(\theta^{\lambda}_n, X_{n+1}) \rangle\\
&\quad +2\lambda^2 |h(\dot{\theta}^{\lambda}_n)  - h(\theta^{\lambda}_n) |^2 +2\lambda^2|h(\theta^{\lambda}_n) -H(\theta^{\lambda}_n, X_{n+1})|^2,
\end{align*}
which implies, by taking conditional expectation on both sides and by using Remark \ref{hconvex2}
\begin{align*}
\E\left[\left.| \dot{\theta}^{\lambda}_{n+1} -\theta^{\lambda}_{n+1}|^2 \right| \dot{\theta}^{\lambda}_n, \theta^{\lambda}_n\right]
& \leq| \dot{\theta}^{\lambda}_n -\theta^{\lambda}_n|^2 - 2\hat{a}^*\lambda| \dot{\theta}^{\lambda}_n -\theta^{\lambda}_n|^2  - \frac{2\lambda}{\hat{a}+L}|h(\dot{\theta}^{\lambda}_n) -h(\theta^{\lambda}_n)|^2 \\
&\quad +2\lambda^2 |h(\dot{\theta}^{\lambda}_n)  - h(\theta^{\lambda}_n) |^2 +2\lambda^2\E\left[\left.|h(\theta^{\lambda}_n) -H(\theta^{\lambda}_n, X_{n+1})|^2 \right| \dot{\theta}^{\lambda}_n, \theta^{\lambda}_n\right],
\end{align*}
where $\hat{a}^* = \hat{a}L/(\hat{a}+L)$. For $\lambda <\bar{\lambda}_{\max}$, one obtains by using Remark \ref{growth} and \ref{hlip}
\begin{align*}
&\E\left[\left.| \dot{\theta}^{\lambda}_{n+1} -\theta^{\lambda}_{n+1}|^2 \right| \dot{\theta}^{\lambda}_n, \theta^{\lambda}_n\right]\\
& \leq(1- 2\hat{a}^*\lambda)| \dot{\theta}^{\lambda}_n -\theta^{\lambda}_n|^2 +4\lambda^2\E\left[\left.|h(\theta^{\lambda}_n) |^2 \right| \dot{\theta}^{\lambda}_n, \theta^{\lambda}_n\right] \\
&\quad +4\lambda^2\E\left[\left.|H(\theta^{\lambda}_n, X_{n+1})|^2 \right| \dot{\theta}^{\lambda}_n, \theta^{\lambda}_n\right]\\
& \leq (1- 2\hat{a}^*\lambda)| \dot{\theta}^{\lambda}_n -\theta^{\lambda}_n|^2 +4\lambda^2L^2\E\left[\left.|\theta^{\lambda}_n-\theta^*|^2 \right| \dot{\theta}^{\lambda}_n, \theta^{\lambda}_n\right] \\
&\quad +4\lambda^2\E\left[\left. \left((1+|X_{n+1}|)^{\rho+1}( L_1|\theta^{\lambda}_n- \theta^*|+L_1| \theta^*|+L_2) +F_*(X_{n+1})\right)^2 \right| \dot{\theta}^{\lambda}_n, \theta^{\lambda}_n\right]\\
& \leq (1- 2\hat{a}^*\lambda)| \dot{\theta}^{\lambda}_n -\theta^{\lambda}_n|^2 +(4\lambda^2L^2+8\lambda^2L_1^2\mathbb{E}\left[K_{\rho}(X_0)\right])|\theta^{\lambda}_n-\theta^*|^2\\
&\quad +24\lambda^2L_1^2\mathbb{E}\left[K_{\rho}(X_0)\right]| \theta^*|^2+24\lambda^2L_2^2\mathbb{E}\left[K_{\rho}(X_0)\right]+24\lambda^2\mathbb{E}\left[F_*^2(X_0)\right].
\end{align*}
Finally, one calculates by using Lemma \ref{2ndbdconv},
\begin{align*}
\E\left[| \dot{\theta}^{\lambda}_{n+1} -\theta^{\lambda}_{n+1}|^2 \right]
& \leq (1- 2\hat{a}^*\lambda)\E\left[| \dot{\theta}^{\lambda}_n -\theta^{\lambda}_n|^2\right] +(4\lambda^2L^2+8\lambda^2L_1^2\mathbb{E}\left[K_{\rho}(X_0)\right])\E\left[|\theta^{\lambda}_n-\theta^*|^2\right]\\
&\quad+24\lambda^2L_1^2\mathbb{E}\left[K_{\rho}(X_0)\right]| \theta^*|^2+24\lambda^2L_2^2\mathbb{E}\left[K_{\rho}(X_0)\right]+24\lambda^2\mathbb{E}\left[F_*^2(X_0)\right]\\
&\leq (1- 2\hat{a}^*\lambda)\E\left[| \dot{\theta}^{\lambda}_n -\theta^{\lambda}_n|^2\right]+\lambda^2 c_5,
\end{align*}
where $c_5 =(8L^2+16L_1^2\mathbb{E}\left[K_{\rho}(X_0)\right]) ( \E\left[|\theta_0|^2\right]+\hat{a}^{-1}\bar{c}_4 )+(8L^2+40L_1^2\mathbb{E}\left[K_{\rho}(X_0)\right])| \theta^*|^2 +24L_2^2\mathbb{E}\left[K_{\rho}(X_0)\right]+24\mathbb{E}\left[F_*^2(X_0)\right]$. The result follows by induction.
\end{proof}
\noindent {\bf Proof of Theorem \ref{mainconvex}} One observes that by using Theorem \ref{converconvex} and Lemma \ref{converconvex2}
\begin{align*}
W_2(\mathcal{L}(\theta^{\lambda}_n),\pi_{\beta})
&\leq W_2(\mathcal{L}(\theta^{\lambda}_n),\mathcal{L}(\dot{\theta}^{\lambda}_n)) +W_2 (\mathcal{L}(\dot{\theta}^{\lambda}_n), \pi_{\lambda}) +W_2(\pi_{\lambda}, \pi_{\beta})\\
&\leq \bar{C}_{8,2}\sqrt{\lambda} + \bar{C}_7 e^{-\hat{a}^*\lambda n}+\bar{C}_{8,1}\sqrt{\lambda}\\
&\leq C_7e^{-C_6\lambda n} + C_8\sqrt{\lambda},
\end{align*}
where
\begin{equation}\label{mainconvexconst}
 C_6 = \hat{a}^*, \quad C_7 =  \bar{C}_7, \quad C_8 = \bar{C}_{8,1}+\bar{C}_{8,2}
\end{equation}
with $\hat{a}^* =\hat{a}L/(\hat{a}+L)$, $\bar{C}_7$ given in Lemma \ref{converconvex}, $\bar{C}_{8,1}$ and $\bar{C}_{8,2}$ given in \eqref{converconst1} and \eqref{converconst2} respectively.
%\section{Applications}\label{example}

\noindent {\bf Proof of Corollary \ref{eer2}} The proof follows the same lines as the proof of Corollary \ref{eer}. To obtain an upper bound for the expected excess risk $\mathbb{E}[U(\hat{\theta})] - \inf_{\theta \in \mathbb{R}^d} U(\theta) $, one considers
\begin{equation}\label{eersplitting2}
 \mathbb{E}[U(\hat{\theta})] - \inf_{\theta \in \mathbb{R}^d} U(\theta) = \left( \mathbb{E}[U(\hat{\theta})] -  \mathbb{E}[U(Z_{\infty})]\right) + \left( \mathbb{E}[U(Z_{\infty})]- \inf_{\theta \in \mathbb{R}^d} U(\theta) \right),
\end{equation}
where $\hat{\theta} =  \theta^{\lambda}_n$ and $Z_{\infty}\sim \pi_{\beta}$ with $\pi_{\beta}(\theta) = \exp(-\beta U(\theta))$ for all $\theta \in \mathbb{R}^d$. By using \cite[Lemma~3.5]{raginsky}, Lemma \ref{sdemoment}, \ref{2ndbdconv} and Theorem \ref{mainconvex}, the first term on the RHS of \eqref{eersplitting2} can be bounded by
\begin{align*}
&\mathbb{E}[U(\hat{\theta})] -  \mathbb{E}[U(Z_{\infty})] \\
& \leq \left(L(\E\left[|\theta_0- \theta^*|^2\right] + \bar{c}_4\hat{a}^{-1}+|\theta^*|^2)^{1/2}+|h(0)|\right)W_2(\mathcal{L}(\theta^{\lambda}_n),\pi_{\beta})\\
&\leq \left(L(\E\left[|\theta_0- \theta^*|^2\right] + \bar{c}_4\hat{a}^{-1}+|\theta^*|^2)^{1/2}+|h(0)|\right)\left(C_7e^{-C_6\lambda n} + C_8\sqrt{\lambda}\right)\\
&\leq \hat{C}_5e^{-\hat{C}_4\lambda n}+\hat{C_6}\sqrt{\lambda},
\end{align*}
where $\theta^* \in \mathbb{R}^d$ is the minimizer of $U$, and
\begin{align}\label{eerconst3}
\begin{split}
\hat{C}_4 &= C_6,\\
\hat{C}_5 & = C_7 \left(L(\E\left[|\theta_0- \theta^*|^2\right] + \bar{c}_4\hat{a}^{-1}+|\theta^*|^2)^{1/2}+|h(0)|\right), \\
\hat{C}_6 & = C_8 \left(L(\E\left[|\theta_0- \theta^*|^2\right] + \bar{c}_4\hat{a}^{-1}+|\theta^*|^2)^{1/2}+|h(0)|\right),
\end{split}
\end{align}
with $C_6, C_7, C_8$ given in \eqref{mainconvexconst} and $\bar{c}_4$ given in \eqref{c4}. Moreover, the second term on the RHS of \eqref{eersplitting2} can be estimated by using \cite[Proposition~3.4]{raginsky}, which gives,
\[
\mathbb{E}[U(Z_{\infty})]- \inf_{\theta \in \mathbb{R}^d} U(\theta) \leq  \frac{\hat{C}_7}{\beta},
\]
where
\begin{equation}\label{eerconst4}
\hat{C}_7 = \frac{d}{2}\log\left(\frac{e\beta L}{d}\left(\frac{d}{\hat{a}\beta}+|\theta^*|^2\right)\right).
\end{equation}
Finally, one obtains
\[
\mathbb{E}[U(\hat{\theta})] - \inf_{\theta \in \mathbb{R}^d} U(\theta)  \leq  \hat{C}_5e^{-\hat{C}_4\lambda n}+\hat{C_6}\sqrt{\lambda}+\hat{C}_7/\beta.
\]

\section{Applications}\label{application}
\subsection{Quantile estimation with \texorpdfstring{$L_2$ }\phantom{}regularization}
We consider the problem of quantile estimation for AR(1) processes, which has been discussed in \cite{4}, \cite{qr} and \cite{ql2} amongst others, with $L_2$ regularization. It assumed therefore that the data $X_t \in \mathbb{R}$, $t \in \mathbb{Z}$, follows an AR(1) process given by
\[
X_{t+1} = \alpha X_t +\bar{\xi}_{t+1},
\]
where $\alpha$ is a constant with $|\alpha|<1$ and $(\bar{\xi}_t)_{t \in \mathbb{Z}}$ are i.i.d. standard Normal random variables. The above expression can be further rewritten as
\[
X_t = \sum_{j=0}^{\infty} \alpha^j \bar{\xi}_{t-j}.
\]
One notes that $X_t$ has a stationary distribution $\pi_X$ which is normally distributed with mean 0 and variance $1/(1-\alpha^2)$. Our task is to identify the $q$-th quantile of the stationary distribution $\pi_X$ using the SGLD algorithm \eqref{SGLD}, in other words, we aim to solve the following problem:
\[
\min_\theta  \mathbb{E}\left[l_q(X_{\infty}-\theta)\right] + \gamma|\theta|^2,
\]
where $X_{\infty} \sim \pi_X$ and
\[l_q(z) =	\begin{cases}
     			 qz, & z \geq 0, \\
      			 (q-1)z,  & z<0.
  			 \end{cases}
\]
The stochastic gradient $H: \mathbb{R}\times \mathbb{R}  \rightarrow \mathbb{R}$ is given by
\begin{equation}\label{qH}
H(\theta, x) = -q +\mathbbm{1}_{\{x < \theta\}} +2\gamma \theta,
\end{equation}
where $\gamma$ is a positive constant. To check Assumption \ref{expressionH}, denote by $ F(\theta,x) = -q +2\gamma \theta$, $ G(\theta, x) = \mathbbm{1}_{\{x < \theta\}} $. It can be easily seen that Assumption \ref{expressionH} holds with $\rho = 0$, $L_1 =  2\gamma$, $L_2 = 0$ and $K_1(x) = 1$. Then, by Remark \ref{clcex} and its proof in \ref{proofclcex}, Assumption \ref{clc} holds with $L = 2\gamma+ 1$. Moreover, Assumption \ref{assum:dissipativity} holds with $A(x) = \gamma \mathbf{I}_d$ and $b(x) =  q^2/(4\gamma)$, which implies $a =  \gamma $ and $b =  q^2/(4\gamma)$.

One notes that the value of the $q$-th quantile of $\pi_{X}$ is given by $\theta^* = N(q)/\sqrt{1-\alpha^2}$ where $N(\cdot)$ is the cumulative distribution function of the standard normal distribution. For the simulation, set $\alpha = 0.5$, $q = 0.95$, and thus, $\theta^* = 1.89$. Moreover, let $m = 1$, $\theta_0 = 3$, $\beta = 10^8$ and $\gamma = 10^{-6}$. Note that we use the step restriction given in Remark \ref{stepr} for all the examples in this section. In Figure \ref{qgraph}, the left graph is obtained by using the SGLD algorithm \eqref{SGLD} with $\lambda = 10^{-4}$ and the number of iterations $n = 10^6$. It shows the path of $\theta_n$ with the first $10000$ iterations being discarded, and the path stabilises at around the true value $\theta^* = 1.89$. The right graph of Figure \ref{qgraph} illustrates the rate of convergence of the SGLD algorithm in Wasserstein-1 distance based on 5000 samples. The slope of the results in $W_1$ obtained using numerical experiments is 0.5022, which supports our theoretical finding in Theorem \ref{main} with rate $1/2$.

\begin{figure}
    \centering
    \begin{minipage}{0.5\textwidth}
        \centering
        \includegraphics[width=0.9\textwidth]{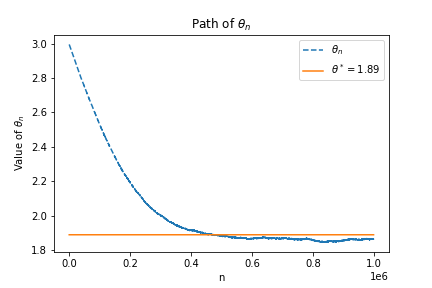} % first figure itself
  %      \caption{\footnotesize Path of $\theta_n$}
    \end{minipage}\hfill
    \begin{minipage}{0.5\textwidth}
        \centering
        \includegraphics[width=0.9\textwidth]{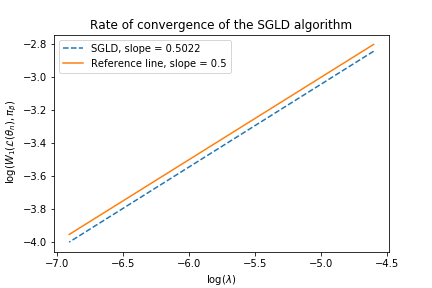} % second figure itself
    \end{minipage}
    \caption{\footnotesize \it [Left] Path of $\theta_n$ when $q = 0.95$. [Right] Rate of convergence of the SGLD algorithm.}
    \label{qgraph}
\end{figure}

\subsection{VaR-CVaR algorithm} \label{varcvarex}
In this section, we consider the problem of computing Value-at-Risk (VaR) and Conditional-Value-at-Risk (CVaR), which are two commonly used risk measures in financial risk management. In order to obtain the two quantities, one considers the following optimization problem:
\begin{equation}\label{varcvarproblem}
\min_{\theta}V(\theta)=\min_{\theta}\left(\mathbb{E}\left[\theta+\frac{1}{1-\bar{q}}(f(X)-\theta)_+\right]+\gamma|\theta|^2\right),
\end{equation}
where $0<\bar{q}<1$, $f$ is continuous and $f(X)$ is integrable with respect to the probability measure. As noted in \cite{varcvar}, $f$ can represent more complicated payoff structures than simple vanilla instruments while $X$ can accommodate a large family of asset distributions including those generated by stochastic/local volatility models, see e.g. \cite{LocalVol}, \cite{Sabanis-SV-mean} and \cite{Sabanis-SV} references therein. Then, by \cite[Proposition 2.1]{varcvar}, $\text{VaR}_{\bar{q}}(f(X)) = \argmin V(\theta) $ and $\text{CVaR}_{\bar{q}}(f(X)) = \min_{\theta}V(\theta) $. To compute VaR, the stochastic gradient $H: \mathbb{R}\times \mathbb{R}  \rightarrow \mathbb{R}$ of the SGLD algorithm \eqref{SGLD} is given by
\[
H(\theta, x) = 1- \frac{1}{1-\bar{q}}\mathbbm{1}_{\{f(x) \geq \theta\}}+2\gamma\theta =-\frac{\bar{q}}{1-\bar{q}}+ \frac{1}{1-\bar{q}}\mathbbm{1}_{\{f(x) < \theta\}}+2\gamma\theta.
\]

\subsubsection{Single asset} Let $f(x) = x$, one notices that the above expression has a similar form as \eqref{qH}. Then, one can check that Assumption \ref{expressionH} - \ref{assum:dissipativity} are satisfied. More precisely, denote by $ F(\theta,x) = -\bar{q}/(1-\bar{q}) +2\gamma \theta$, $ G(\theta, x) = \mathbbm{1}_{\{x < \theta\}}/(1-\bar{q}) $, Assumption \ref{expressionH} holds with $\rho = 0$, $L_1 =  2\gamma$, $L_2 = 0$ and $K_1(x) = 1/(1-\bar{q})$. Let $X$ be a one-dimensional random variable with finite fourth moment, then Assumption \ref{iid} is satisfied. Denote by $\bar{c}_d$ the upper bound of the density of $X$, Assumption \ref{clc} holds with $L = 2\gamma +\bar{c}_d/(1-\bar{q})$. Furthermore, Assumption \ref{assum:dissipativity} holds with $A(x) = \gamma \mathbf{I}_d$ and $b(x) = \bar{q}^2/(4\gamma(1-\bar{q})^2)$, which implies $a =  \gamma $ and $b = \bar{q}^2/(4\gamma(1-\bar{q})^2)$.

\begin{table}[t!]
\footnotesize
\centering
\begin{tabular}{c c c c c c c c c}
 \hline
 \hline
& \multicolumn{4}{c}{$\bar{q} = 0.95$} & \multicolumn{4}{c}{$\bar{q} = 0.99$}\\ %[0.5ex]
 \hline
 \hline
  &VaR* &CVaR* &$\text{VaR}_{\text{SGLD}}$ &$\text{CVaR}_{\text{SGLD}}$  &VaR* &CVaR* &$\text{VaR}_{\text{SGLD}}$ &$\text{CVaR}_{\text{SGLD}} $ \\
 \hline
 \multirow{2}{*}{ $\mu = 0, \sigma = 1$ }& \multirow{2}{*}{  1.645} & \multirow{2}{*}{  2.062 }&
 1.642&2.062& \multirow{2}{*}{ 2.326}& \multirow{2}{*}{ 2.677}&2.329&2.662\\
 \multicolumn{3}{c}{}&(0.02)&(0.0006) &\multicolumn{2}{c}{}&(0.04) &(0.0038)  \\
 \hline
 \multirow{2}{*}{  $\mu = 1, \sigma = 2$} &  \multirow{2}{*}{ 4.290 }& \multirow{2}{*}{  5.124} &
 4.294&5.126&\multirow{2}{*}{5.653}&\multirow{2}{*}{6.335}&5.640&6.336\\
 \multicolumn{3}{c}{}&(0.03) &(0.0006)&\multicolumn{2}{c}{}&(0.06) &(0.0032) \\
 \hline
   \multirow{2}{*}{$\mu = 3, \sigma = 5$} & \multirow{2}{*}{ 11.224 }& \multirow{2}{*}{13.311 }& 11.230 &13.305& \multirow{2}{*}{14.632}& \multirow{2}{*}{16.337}&14.643&16.313 \\
 \multicolumn{3}{c}{}&(0.05) &(0.0006)& \multicolumn{2}{c}{}&(0.11) &(0.006)\\
 \hline
\end{tabular}
\caption{VaR and CVaR for normal distribution $N(\mu, \sigma)$.}
\label{normalt}
\end{table}
\begin{table}[t!]
\footnotesize
\centering
\begin{tabular}{c c c c c c c c c}
 \hline
 \hline
& \multicolumn{4}{c}{$\bar{q} = 0.95$} & \multicolumn{4}{c}{$\bar{q} = 0.99$}\\ %[0.5ex]
 \hline
 \hline
  &VaR* &CVaR* &$\text{VaR}_{\text{SGLD}}$ &$\text{CVaR}_{\text{SGLD}}$  &VaR* &CVaR* &$\text{VaR}_{\text{SGLD}}$ &$\text{CVaR}_{\text{SGLD}} $ \\
 \hline
 \multirow{2}{*}{ $\text{d.f.} =10$ }& \multirow{2}{*}{  1.812} & \multirow{2}{*}{  2.416 }&
1.808&2.407& \multirow{2}{*}{ 2.764}& \multirow{2}{*}{ 3.357}&2.767&3.350\\
 \multicolumn{3}{c}{}&(0.02)&(0.0005) &\multicolumn{2}{c}{}&(0.05) &(0.003)  \\
 \hline
 \multirow{2}{*}{   $\text{d.f.} =7$} &  \multirow{2}{*}{ 1.895 }& \multirow{2}{*}{2.595} &
1.895&2.594&\multirow{2}{*}{2.998}&\multirow{2}{*}{3.757}&3.001&3.782\\
 \multicolumn{3}{c}{}&(0.03) &(0.0008)&\multicolumn{2}{c}{}&(0.05) &(0.0024) \\
 \hline
   \multirow{2}{*}{$\text{d.f.} =3$} & \multirow{2}{*}{ 2.353 }& \multirow{2}{*}{3.876 }& 2.358 &3.873& \multirow{2}{*}{4.541}& \multirow{2}{*}{6.968}&4.542&6.967 \\
 \multicolumn{3}{c}{}&(0.03) &(0.0008)& \multicolumn{2}{c}{}&(0.08) &(0.0028)\\
 \hline
\end{tabular}
\caption{VaR and CVaR for Student's t distribution.}
\label{tt}
\end{table}

\begin{figure} [t!]
    \centering
    \begin{minipage}{0.5\textwidth}
        \centering
        \includegraphics[width=0.9\textwidth]{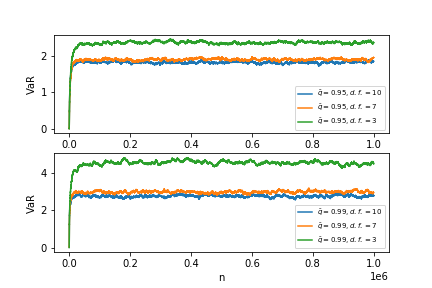} % first figure itself
  %      \caption{\footnotesize Path of $\theta_n$}
    \end{minipage}\hfill
    \begin{minipage}{0.5\textwidth}
        \centering
        \includegraphics[width=0.9\textwidth]{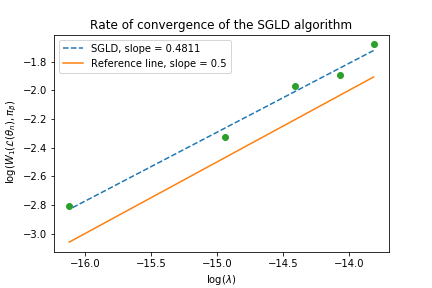} % second figure itself
    \end{minipage}
    \caption{\footnotesize \it [Left] Path of $\theta_n$ (VaR) for Student's t-distribution. [Right] Rate of convergence of the SGLD algorithm based on 5000 samples.}
    \label{vcvargraph}
\end{figure}

For the numerical experiments, we set $\theta_0 = 0$, $\beta = 10^8$, $\gamma = 10^{-8}$, $\lambda = 10^{-4}$ and the number of iterations $n = 10^6$. Table \ref{normalt} and \ref{tt} present VaR and CVaR for the normal distribution and Student's t-distribution. VaR* and CVaR* in the tables denote the theoretical values, while $\text{VaR}_{\text{SGLD}}$ and $\text{CVaR}_{\text{SGLD}}$ denote the numerical approximations from the SGLD algorithm \eqref{SGLD}. Each approximation in the table is obtained based on 10000 samples, which is followed by its sample standard deviation shown in brackets. In addition, in Figure \ref{vcvargraph}, the left graph illustrates the path of $\theta_n$ for the $t$-distribution, whereas the right graph shows that the rate of convergence of the SGLD algorithm \eqref{SGLD} is 0.4811. One notes that the samples from $\pi_{\beta}$ is generated by running the SGLD algorithm with $\lambda = 10^{-5}$.

\subsubsection{Minimizing CVaR of portfolios of assets} \label{cvar_port}
To minimize CVaR for a given portfolio, we consider the following optimization problem:
\begin{equation}\label{varcvarproblem2}
\min_{\hat{\theta}}V(\hat{\theta})=\min_{\hat{\theta}}\left(\mathbb{E}\left[ \frac{1}{1-\bar{q}}\left(\sum_{i=1}^{n}g_i(w)X_i-\theta\right)_+ + \theta \right]+\gamma|\hat{\theta}|^2\right),
\end{equation}
where the parameter $\hat{\theta}: = (\theta, w)^{\intercal} = (\theta, w_1, \dots, w_n)^{\intercal}$ and $g_i(w) := \frac{e^{w_i}}{\sum_{j=1}^{n}e^{w_j}} \in (0,\,1)$ for $i = 1, \dots, n$. By solving \eqref{varcvarproblem2}, we obtain not only VaR for a given portfolio, but also the optimal weight for each asset in the portfolio such that CVaR is minimized.

For reasons of brevity, we assume here that the $X_i$'s, for $i = 1, \dots, n$, are i.i.d. one-dimensional random variables (with finite fourth moments). Our results can be naturally extended to the case of dependent data streams via the concept of $L$-mixing as explained in \cite{nonconvex}.

Let $c_X$, $c_{\bar{X}}$ denote the first and second absolute moment respectively of $X_1$. Moreover, let $|x|f_{X_i}(x)$ be bounded for any $i$ and $x \in \mathbb{R}$. Note that this latter requirement is satisfied for a wide range of distributions, for example, the distributions shown in Table \ref{figureportfolio}. Then, the stochastic gradient $H_{\hat{\theta}}(\hat{\theta},x): \mathbb{R}^{n+1} \times \mathbb{R}^n \rightarrow \mathbb{R}^{n+1}$ is defined as
\[
H_{\hat{\theta}}(\hat{\theta},x):= (H_{\theta}(\hat{\theta},x),H_{w_1}(\hat{\theta},x), \dots, H_{w_n}(\hat{\theta},x))^{\intercal},
\]
where $H_{\theta}(\hat{\theta},x): \mathbb{R}^{n+1} \times \mathbb{R}^n \rightarrow \mathbb{R}$ and $H_{w_j}(\hat{\theta}, x):  \mathbb{R}^{n+1} \times \mathbb{R}^n \rightarrow \mathbb{R}$ for all $j$ are given by
\[
H_{\theta}(\hat{\theta}, x) = 1- \frac{1}{1-\bar{q}}\mathbbm{1}_{\{\sum_{i=1}^{n}g_i(w)x_i \geq \theta\}}+2\gamma\theta,
\]
and
\[
H_{w_j}(\hat{\theta}, x)= \frac{1}{1-\bar{q}}\hat{g}_{w_j}(w,x)\mathbbm{1}_{\{\sum_{i=1}^{n}g_i(w)x_i \geq \theta\}}+2\gamma w_j,
\]
where
\[
\hat{g}_{w_j}(w,x) = \sum_{i=1}^{n}\frac{\partial g_i(w)}{\partial w_j}x_i
\]
for any $j = 1, \dots, n$ with $\frac{\partial g_j(w)}{\partial w_j} = \frac{e^{w_j}(\sum_{l \neq j} e^{w_l})}{(\sum_{l=1}^{n}e^{w_l})^2}$, and $\frac{\partial g_i(w)}{\partial w_j} = -\frac{e^{w_i}e^{w_j}}{(\sum_{l=1}^{n}e^{w_l})^2}$ for $i \neq j$. One notes that $|\hat{g}_{w_j}(w,x) | \leq \sum_{i=1}^n |x_i|$ for any $j$. Moreover, if Assumption \ref{expressionH} - \ref{assum:dissipativity} hold for $H_{\theta}$ and $H_{w_j}$ for any $j$, then the assumptions hold for $H_{\hat{\theta}}$.

We first check assumptions for $H_{\theta}$. Denote by
\[
F_{\theta}(\hat{\theta},x) =2\gamma \theta, \quad G_{\theta}(\hat{\theta}, x) =1- \mathbbm{1}_{\{\sum_{i=1}^{n}g_i(w)x_i \geq \theta\}}/(1-\bar{q}),
\]
then $H_{\theta} = F_{\theta}+G_{\theta}$. Assumption \ref{expressionH} holds with $\rho = 0$, $L_1 =  2\gamma$, $L_2 = 0$ and $K_1(x) = (2-\bar{q})/(1-\bar{q})$. By taking into consideration the expression of $K_1(x)$ and the construction of the problem, Assumption \ref{iid} is satisfied. Assumption \ref{assum:dissipativity} holds with $A(x) = 2\gamma \mathbf{I}_d$ and $b(x) = 0$, which implies $a = 2 \gamma $ and $b =0$. To check Assumption \ref{clc}, one considers $\hat{\theta}': = (\bar{\theta}, w)^{\intercal}$, and then calculates by assuming without loss of generality $g_n(w) = \max\{g_1(w), \dots, g_n(w)\}$
\begin{align*}
&\mathbb{E}\left[\left|H_{\theta}(\hat{\theta}, X)  - H_{\theta}(\hat{\theta}', X)\right|\right]\\
& \leq 2\gamma\left|\theta -\bar{\theta}\right|+\frac{1}{1-\bar{q}}\mathbb{E}\left[\left|\mathbbm{1}_{\{\sum_{i=1}^{n}g_i(w)X_i \geq \theta\}} -  \mathbbm{1}_{\{\sum_{i=1}^{n}g_i(w)X_i \geq \bar{\theta}\}}\right|\right]\\
&\leq 2\gamma\left|\hat{\theta}-\hat{\theta}'\right|+\frac{1}{1-\bar{q}}(E_1+E_2),
\end{align*}
where
\[
E_1 =\mathbb{E}\left[ \mathbbm{1}_{\{\theta \leq \sum_{i=1}^{n}g_i(w)X_i \leq \bar{\theta} \}} \right], \quad E_2 = \mathbb{E}\left[ \mathbbm{1}_{\{\bar{\theta} \leq \sum_{i=1}^{n}g_i(w)X_i \leq  \theta\}} \right].
\]
To estimate $E_1$, one writes
\begin{align*}
&\mathbb{E}\left[ \mathbbm{1}_{\{\theta \leq \sum_{i=1}^{n}g_i(w)X_i \leq \bar{\theta} \}} \right]\\
& = \mathbb{E}\left[ \mathbb{E}\left[\left.\mathbbm{1}_{\{(\theta - \sum_{i\neq n} g_i(w)X_i)/g_n(w)) \leq X_n \leq (\bar{\theta} - \sum_{i\neq n} g_i(w)X_i)/g_n(w)\}} \right|X_1, \dots, X_{n-1} \right]\right]\\
& = \int_{-\infty}^{\infty}\cdots \int_{-\infty}^{\infty}\int_{ (\theta - \sum_{i\neq n} g_i(w)x_i)/g_n(w) }^{  (\bar{\theta} - \sum_{i\neq n} g_i(w)x_i)/g_n(w) }  f_{X_n}(z)\, dz f_{X_{n-1}}(x_{n-1})\,dx_{n-1}\cdots f_{X_1}(x_1)\, dx_1\\
&\leq  nc_{X_n}\left|\hat{\theta}-\hat{\theta}'\right|,
\end{align*}
where we use the fact $g_n(w)\geq 1/n$ in the last inequality and $c_{X_n}$ denotes the upper bound of the density of $X_n$. $E_2$ can be estimated by using similar arguments. Then, one obtains
\[
\mathbb{E}\left[\left|H_{\theta}(\hat{\theta}, X)  - H_{\theta}(\hat{\theta}', X)\right|\right] \leq (2\gamma+2nc_{X_n}/(1-\bar{q}))\left|\hat{\theta}-\hat{\theta}'\right|,
\]
which implies Assumption \ref{clc} holds with $L = 2\gamma +2nc_{X_n}/(1-\bar{q})$.

Next, we check assumptions for $H_{w_j}$. Denote by
\[
F_{w_j}(\hat{\theta},x) =2\gamma w_j, \quad G_{w_j}(\hat{\theta}, x) = \hat{g}_{w_j}(w,x)\mathbbm{1}_{\{\sum_{i=1}^{n}g_i(w)x_i \geq \theta\}}/(1-\bar{q}),
\]
then $H_{w_j} = F_{w_j}+G_{w_j}$. Assumption \ref{expressionH} holds with $\rho = 0$, $L_1 =  2\gamma$, $L_2 = 0$ and $K_1(x) =\sum_i|x_i|/(1-\bar{q})$. By taking into consideration the expression of $K_1(x)$ and the construction of the problem, Assumption \ref{iid} is satisfied. Assumption \ref{assum:dissipativity} holds with $A(x) = 2\gamma \mathbf{I}_d$ and $b(x) = 0$, which implies $a = 2 \gamma $ and $b =0$. Then, we check Assumption \ref{clc} for $H_{w_1}$, and the arguments stay the same lines for any other $H_{w_j}$, $j = 2, \dots, n$. Consider $\hat{\theta}^{\sharp}: = (\theta, \bar{w})^{\intercal} = (\theta, \bar{w}_1,  w_2, \dots, w_n)^{\intercal}$. Then, one calculates
\begin{align*}
&\mathbb{E}\left[\left|H_{w_1}(\hat{\theta}, X)  - H_{w_1}(\hat{\theta}^{\sharp}, X)\right|\right] \\
&\leq  2\gamma\left|w_1 - \bar{w}_1\right| +\frac{1}{1-\bar{q}}\mathbb{E}\left[\left|\hat{g}_{w_1}(w,X)\mathbbm{1}_{\{\sum_{i=1}^{n}g_i(w)X_i \geq \theta\}} -  \hat{g}_{w_1}(\bar{w},X)\mathbbm{1}_{\{\sum_{i = 1}g_i(\bar{w})X_i\geq \theta\}}\right|\right]\\
& \leq 2\gamma\left|\hat{\theta} - \hat{\theta}^{\sharp}\right| +\frac{1}{1-\bar{q}}\mathbb{E}\left[\left|\hat{g}_{w_1}(w,X)\mathbbm{1}_{\{\sum_{i=1}^{n}g_i(w)X_i \geq \theta\}} -  \hat{g}_{w_1}(\bar{w},X)\mathbbm{1}_{\{\sum_{i=1}^{n}g_i(w)X_i \geq \theta\}} \right|\right]\\
&\quad + \frac{1}{1-\bar{q}}\mathbb{E}\left[\left|\hat{g}_{w_1}(\bar{w},X)\mathbbm{1}_{\{\sum_{i=1}^{n}g_i(w)X_i \geq \theta\}} -  \hat{g}_{w_1}(\bar{w},X)\mathbbm{1}_{\{\sum_{i = 1}g_i(\bar{w})X_i\geq \theta\}}\right|\right]\\
& \leq 2\gamma\left|\hat{\theta} - \hat{\theta}^{\sharp}\right| +\frac{2nc_X}{1-\bar{q}}|w_1 - \bar{w}_1|\\
&\quad  + \frac{1}{1-\bar{q}}\mathbb{E}\left[\left|\hat{g}_{w_1}(\bar{w},X)\mathbbm{1}_{\{\sum_{i=1}^{n}g_i(w)X_i \geq \theta\}} -  \hat{g}_{w_1}(\bar{w},X)\mathbbm{1}_{\{\sum_{i = 1}g_i(\bar{w})X_i\geq \theta\}}\right|\right]\\
& \leq 2\gamma\left|\hat{\theta} - \hat{\theta}^{\sharp}\right| +\frac{2nc_X}{1-\bar{q}}\left|\hat{\theta} - \hat{\theta}^{\sharp}\right| \\
&\quad  + \frac{1}{1-\bar{q}}\mathbb{E}\left[\left|\hat{g}_{w_1}(\bar{w},X)\mathbbm{1}_{\{\sum_{i=1}^{n}g_i(w)X_i \geq \theta\}} -  \hat{g}_{w_1}(\bar{w},X)\mathbbm{1}_{\{\sum_{i = 1}g_i(\bar{w})X_i\geq \theta\}}\right|\right],
\end{align*}
where the third inequality holds due to the fact that $|\hat{g}_{w_1}(w,X) -  \hat{g}_{w_1}(\bar{w},X)|\leq 2|w_1 - \bar{w}_1|\sum_i|X_i|$. %$g''(w_i) = \frac{e^{w_i}}{\sum_{j=1}^{n}e^{w_j}} - \frac{3e^{2w_i}}{(\sum_{j=1}^{n}e^{w_j})^2}+\frac{2e^{3w_i}}{(\sum_{j=1}^{n}e^{w_j})^3}$ is bounded and thus $g'(w_i)$ is Lipschitz continuous.
Then, by using $|\hat{g}_{w_1}(\bar{w},x)| \leq \sum_i|x_i|$,
\begin{align}\label{inddiffappendix}
\begin{split}
&\mathbb{E}\left[\left|H_{w_1}(\hat{\theta}, X)  - H_{w_1}(\hat{\theta}^{\sharp},  X)\right|\right] \\
& \leq 2\gamma\left|\hat{\theta} - \hat{\theta}^{\sharp}\right|+\frac{2nc_X}{1-\bar{q}}\left|\hat{\theta} - \hat{\theta}^{\sharp}\right| \\
&\quad + \frac{1}{1-\bar{q}}\mathbb{E}\left[\sum_i|X_i|\left|\mathbbm{1}_{\{\sum_{i=1}^{n}g_i(w)X_i \geq \theta\}} -  \mathbbm{1}_{\{\sum_{i = 1}g_i(\bar{w})X_i\geq \theta\}}\right|\right]\\
&\leq (2\gamma+2nc_X/(1-\bar{q})) \left|\hat{\theta} - \hat{\theta}^{\sharp}\right|\\
&\quad +2(n-1)(c_{X} (\bar{c}_{X_n}+\bar{c}_{X_1})+(c_{\bar{X}}+(n-2)c_X^2)( c_{X_n}+ c_{X_1}))/(1-\bar{q})\left|\hat{\theta} - \hat{\theta}^{\sharp}\right|,
\end{split}
\end{align}
where $c_X$, $c_{\bar{X}}$ denote the first and the second absolute moment of $X_i$'s respectively, for any $i$, $\bar{c}_{X_i}$ is the upper bound of the function $|x|f_{X_i}$, and $c_{X_i}$ is the upper bound of the density of $X_i$. Detailed calculations to obtain the last inequality in \eqref{inddiffappendix} is given in Appendix \ref{ass3proofvarcvar}. Thus Assumption \ref{clc} holds with $L = 2\gamma+2nc_X/(1-\bar{q}) +2(n-1)(c_{X} (\bar{c}_{X_n}+\bar{c}_{X_1})+(c_{\bar{X}}+(n-2)c_X^2)( c_{X_n}+ c_{X_1}))/(1-\bar{q})$.
\begin{table}[t!]
\footnotesize
\centering
\begin{tabular}{c c c c c c c c c c}
 \hline
 \hline
\multicolumn{2}{c}{}& \multicolumn{4}{c}{SGLD algorithm} & \multicolumn{4}{c}{Reference}\\ %[0.5ex]
 \hline
 \hline
 $X_1$&$X_2$&$w_1$&$w_2$& \multicolumn{2}{c}{$g_1(w)X_1+g_2(w)X_2$} &$w_1^*$&$w_2^*$&\multicolumn{2}{c}{$g_1(w^*)X_1+g_2(w^*)X_2$} \\
  \hline
 \multicolumn{4}{c}{}& $\text{VaR}_{\text{SGLD}}$ &$\text{CVaR}_{\text{SGLD}} $ & \multicolumn{2}{c}{}& VaR*  &CVaR*\\
 \hline
 %\multirow{2}{*}{$X_1$}& \multirow{2}{*}{$X_2$}&\multirow{2}{*}{$w_1$}&\multirow{2}{*}{$w_2$}&\multicolumn{2}{c}{$w_1X_1+w_2Y_2$}&\multirow{2}{*}{$w_1^*$}&\multirow{2}{*}{$w_2^*$}&\multicolumn{2}{c}{$w_1^*X_1+w_2^*Y_2$}\\
 %\multicolumn{4}{c}{}& $\text{VaR}_{\text{SGLD}}$ &$\text{CVaR}_{\text{SGLD}} $ & \multicolumn{2}{c}{}& VaR*  &CVaR*\\
 %\hline
 $N(500,1)$& $N(0,10^{-4})$&0.00002&0.99998&0.025&0.03&0&1&0.016&0.021\\
 $N(0,10^6)$& $N(0,10^{-4})$&0.000006&0.999994&0.016&0.25&0&1&0.016&0.021\\
 $N(1,4)$& $N(0,1)$&0.111 &0.889  &1.615 &2.004&0.11 &0.89 & 1.617 & 1.999 \\
 \hline
 $N(0,1)$& $t$ with d.f. $= 2.01$&0.917&0.083&1.567&1.975&0.9&0.1&1.531&1.971\\
 $N(0,1)$& $t$ with d.f. $= 10$& 0.577& 0.423 &1.236 &1.554 & 0.58& 0.42 &1.224 &1.553 \\
 $N(0,1)$& $t$ with d.f. $= 1000$&0.503&0.497&1.15&1.46&0.5&0.5&1.165&1.461\\
 $N(1,4)$& $t$ with d.f. $= 2.01$& 0.596& 0.404 & 2.941&4.130 &0.61 &  0.39& 2.985& 4.115\\
 $N(1,4)$& $t$ with d.f. $= 10$&0.172 & 0.828 &1.743 & 2.290& 0.17& 0.83 &1.779 & 2.286\\
 $N(1,4)$& $t$ with d.f. $= 1000$& 0.113&  0.887& 1.594& 2.008&0.11 & 0.89 & 1.619& 2.002\\
 \hline
 $N(0,1)$& $\text{Logistic(0,1)}$& 0.775&  0.225& 1.422&1.816 & 0.78&0.22  &1.442 &1.813 \\
 $N(0,1)$& $\text{Logistic(0,29)}$& 0.999& 0.001 &1.633 &2.110 & 1& 0 &1.645 &2.063 \\
 $N(0,1)$& $\text{Logistic(2,10)}$&  0.997 & 0.003& 1.650&2.101 &  1& 0&1.648&2.065 \\
 $N(1,4)$& $\text{Logistic(0,1)}$& 0.402& 0.598 & 2.635&3.262 &0.4 & 0.6 & 2.607& 3.261\\
 $N(1,4)$& $\text{Logistic(0,29)}$& 0.998& 0.002 & 4.284& 5.145& 1&0&4.284&5.116  \\
 $N(1,4)$& $\text{Logistic(2,10)}$&0.991 &0.009  &4.255 &5.132 & 0.99& 0.01 &4.283 & 5.114\\
 \hline
 $N(0,1)$& $\text{Lognormal(0,1)}$&0.966 & 0.034 & 1.662& 2.068& 0.97& 0.03 &1.647 &2.054 \\
 $N(0,1)$& $\text{Lognormal(0,0.01)}$& 0.074& 0.926 &1.145 & 1.205&0.07 & 0.93 & 1.132&1.186 \\
 $N(0,1)$& $\text{Lognormal(1,4)}$&0.9997 &0.0003  &1.674 &2.136 & 1& 0 &1.645 &2.062 \\
 $N(1,4)$& $\text{Lognormal(0,1)}$& 0.732& 0.268 &3.750 & 4.6050.74&0.74&  0.26 &3.771 &4.599 \\
 $N(1,4)$& $\text{Lognormal(0,0.01)}$&0.010 &0.0.989  &1.173 &1.301 &0 & 1 &1.179 &1.230 \\
 $N(1,4)$& $\text{Lognormal(1,4)}$&0.997 &0.003  &4.266 & 5.194& 1&0  &4.292 &5.129 \\
 \hline
$\text{Logistic(0,1)}$& $\text{Lognormal(0,1)}$&0.817& 0.183&2.797  &3.727 & 0.81& 0.19& 2.814&3.724  \\
$\text{Logistic(0,1)}$& $\text{Lognormal(0,0.01)}$& 0.022& 0.978 &1.169 &1.256 &0.02 &0.98  &1.164 &1.217 \\
$\text{Logistic(0,1)}$ & $\text{Lognormal(1,4)}$& 0.997& 0.003 &2.961 &4.030 &1 & 0 & 2.947&3.971 \\
$\text{Logistic(2,10)}$& $\text{Lognormal(0,1)}$& 0.043& 0.956 & 5.245& 8.412&0.04 & 0.96 &5.198 &8.400 \\
$\text{Logistic(2,10)}$& $\text{Lognormal(0,0.01)}$& 0.009& 0.991 &1.184 &1.315& 0& 1 &1.179 &1.229 \\
$\text{Logistic(2,10)}$& $\text{Lognormal(1,4)}$&0.996 &0.004  &31.651&41.748 & 0.99&0.01  & 31.420&41.738 \\
 \hline
\end{tabular}
\caption{$95\%$ VaR and CVaR for portfolios of two assets $X_1$, $X_2$ with the form $w_1X_1+w_2X_2$.}
\label{figureportfolio}
\end{table}

\begin{figure} [t!]
    \centering
	\includegraphics[width=0.5\textwidth]{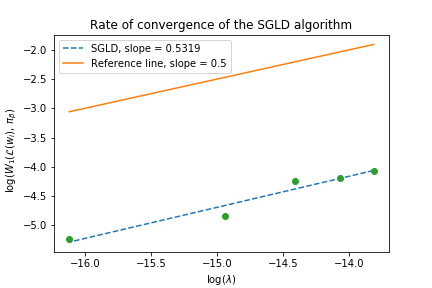}
    \caption{\footnotesize \it Rate of convergence of the SGLD algorithm for $w_1$ based on 5000 samples.}
    \label{portfoliocvargraph}
\end{figure}

For the numerical experiments, we set $\theta_0 = 0$, $\beta = 10^8$, $\gamma = 10^{-8}$, $\lambda = 10^{-4}$ and the number of iterations $n = 10^6$. Tabel \ref{figureportfolio} illustrates $95\%$ VaR and CVaR obtained using the SGLD algorithm for a portfolio of two assets $X_1$ and $X_2$ with weights $g_1(w)$ and $g_2(w)$ respectively. The reference values $w_1^*$, $w_2^*$, VaR* and CVaR* are obtained numerically in the following way:
\begin{enumerate}
\item First, we create 100 evenly spaced numbers over the interval $[0,1]$.
\item Then, for any given distributions of $X_1$ and $X_2$, assign each of the 100 numbers to $g_1(w)$, which is the weight of $X_1$, and calculate the $95\%$ CVaR for the combination $g_1(w)X_1+g_2(w)X_2$.
\item Finally, we obtain the minimum CVaR and the corresponding $g_1(w)$ among the 100 values. We denote them as CVaR* and $g_1(w^*)$. Here , one notes that the corresponding VaR* can be calculated using the optimal weights $g_1(w^*)$ and $g_2(w^*)$.
\end{enumerate}
Figure \ref{portfoliocvargraph} shows that the rate of convergence of the SGLD algorithm \eqref{SGLD} for the parameter $w_1$ is 0.5319, which supports the theoretical finding in Theorem \ref{main}. One notes that the samples from $\pi_{\beta}$ is generated by running the SGLD algorithm with $\lambda = 10^{-5}$.

\newpage

\appendix
\section{Appendix}\label{app}
\subsection{Proof of the claim in Remark \ref{clcex}}\label{proofclcex}
We adapt the proof from \cite[Lemma~4.7]{4} and extend it to an $\mathbb{R}^m$-valued random variable $X_0$. It suffices to consider $H(\theta, X_0) = \dot{g}(\theta, X_0)\mathbbm{1}_{\bigcap_{i = 1}^m\{X_0^{(i)} \in I_i(\theta)\}}$, where $\theta \in \mathbb{R}^d$, $\dot{g}$ is bounded and jointly Lipschitz continuous, i.e. there exist $L_3, L_4, K_2>0$ such that for any $\theta, \theta' \in \mathbb{R}^d$, $x, x' \in \mathbb{R}^m$,
\[
|\dot{g}(\theta,x) - \dot{g}(\theta',x')| \leq (1+|x|+|x'|)^{\rho}(L_3|\theta - \theta'| +L_4|x - x'|), \quad |\dot{g}(\theta, x)| \leq K_2,
\]
and the intervals $I_i(\theta)$ take the form $(-\infty, \bar{g}^{(i)}(\theta))$ with $\bar{g}^{(i)}$ Lipschitz. One notices that the proof follows the same lines when $I_i(\theta)$ takes the form $(\bar{g}^{(i)}(\theta), \infty)$, $(\tilde{g}^{(i)}(\theta), \hat{g}^{(i)}(\theta))$ with $\bar{g}^{(i)}, \tilde{g}^{(i)}, \hat{g}^{(i)}$ Lipschitz. One writes,
\begin{align*}
\left|H(\theta, X_0) -H(\theta', X_0)  \right|
&\leq \left|\dot{g}(\theta, X_0)\mathbbm{1}_{\bigcap_{i = 1}^m\left\{X_0^{(i)}<\bar{g}^{(i)}(\theta)\right\}} - \dot{g}(\theta', X_0)\mathbbm{1}_{\bigcap_{i = 1}^m\left\{X_0^{(i)} <\bar{g}^{(i)}(\theta')\right\}}  \right|\\
&\leq \left|\dot{g}(\theta, X_0)\mathbbm{1}_{\bigcap_{i = 1}^m\left\{X_0^{(i)} <\bar{g}^{(i)}(\theta)\right\}} - \dot{g}(\theta', X_0)\mathbbm{1}_{\bigcap_{i = 1}^m\left\{X_0^{(i)} <\bar{g}^{(i)}(\theta)\right\}} \right|\\
&\quad +\left|\dot{g}(\theta', X_0)\mathbbm{1}_{\bigcap_{i = 1}^m\left\{X_0^{(i)} <\bar{g}^{(i)}(\theta)\right\}} - \dot{g}(\theta', X_0)\mathbbm{1}_{\bigcap_{i = 1}^m\left\{X_0^{(i)} <\bar{g}^{(i)}(\theta')\right\}} \right|\\
&\leq L_3(1+2|X_0|)^{\rho}|\theta - \theta'| +K_2\mathbbm{1}_{\bigcap_{i = 1}^m\left\{X_0^{(i)} \in [\bar{g}^{(i)}(\theta), \bar{g}^{(i)}(\theta'))\right\}},
\end{align*}
where $K_{\rho}(x)$ for any $x \in \mathbb{R}^m$ is defined in \eqref{krho} and we assume without loss of generality $\bar{g}^{(i)}(\theta) \leq \bar{g}^{(i)}(\theta')$ for all $i = 1, \dots, m$. By taking expectation on both sides and by using Cauchy-Schwarz inequality, one obtains
\begin{align*}
&\mathbb{E}\left[\left|H(\theta, X_0) -H(\theta', X_0)  \right|\right]\\
&\leq L_3\mathbb{E}[(1+2|X_0|)^{\rho}]|\theta - \theta'| +K_2P\left({\bigcap_{i = 1}^m\left\{X_0^{(i)} \in [\bar{g}^{(i)}(\theta), \bar{g}^{(i)}(\theta'))\right\}}\right)\\
&\leq L_3\mathbb{E}[(1+2|X_0|)^{\rho}]|\theta - \theta'| + K_2 \int_{\bar{g}^{(m)}(\theta)}^{\bar{g}^{(m)}(\theta')}\cdots \int_{\bar{g}^{(1)}(\theta)}^{\bar{g}^{(1)}(\theta')}f_{X_0}(x^{(1)}, \dots, x^{(m)})d x^{(1)} \cdots d x^{(m)}\\
&\leq L_3\mathbb{E}[(1+2|X_0|)^{\rho}]|\theta - \theta'| + K_2 \int_{\bar{g}^{(1)}(\theta)}^{\bar{g}^{(1)}(\theta')}f_{X_0^{(1)}}(x^{(1)})d x^{(1)}\\
&\leq L_3\mathbb{E}[(1+2|X_0|)^{\rho}]|\theta - \theta'| + K_2K_3L_5|\theta - \theta'|\\
&\leq (L_3+K_2K_3L_5)\mathbb{E}[(1+2|X_0|)^{\rho}]|\theta - \theta'|,
\end{align*}
where $f_{X_0^{(i)}}$ denotes the marginal density function of $X_0^{(i)}$, $K_3$ is an upper bound of $f_{X_0^{(1)}}$ and $L_5$ is a Lipschitz constant for $\bar{g}^{(1)}$. Taking $L = L_3+K_2K_3L_5$ completes the proof.

%%%%%%%%%%%%%%%%%%%%%%%%%%%F is dissipative%%%%%%%%%%%%%%%%%%%%%%%%%%%%%%%%%%
\subsection{Proof of the claim in Remark \ref{disF}}\label{proofre4}
By Assumption \ref{convex}, one obtains, for $\theta\in\mathbb{R}^d$ and $x\in\mathbb{R}^m$,
\[
 \langle  F(\theta,x) -  F(0,x),\theta\rangle \geq \langle \theta, \hat{A}_1(x) \theta\rangle,
\]
which implies
\begin{align*}
 \langle  F(\theta,x) ,\theta\rangle
&\geq \langle \theta, \hat{A}_1(x) \theta\rangle +\langle F(0,x),\theta\rangle \\
&\geq  \langle \theta, \hat{A}_1(x) \theta\rangle -| F(0,x)||\theta| \\
&\geq \langle \theta, \hat{A}_1(x) \theta\rangle -\epsilon|\theta|^2  - (L_2(1+|x|)^{\rho+1}+|  F(0,0)|)^2/(4\epsilon)\\
&\geq \langle \theta, \hat{A}_1^*(x) \theta\rangle -\hat{b}(x),
\end{align*}
where the third inequality holds due to Assumption \ref{expressionH} and $ab <\epsilon a^2 +b^2/(4\epsilon)$, for any $a,b >0$, $\epsilon >0$, $\hat{A}_1^*(x)  = \hat{A}_1(x)  -\epsilon \mathbf{I}_d$ and $\hat{b}(x) = (L_2(1+|x|)^{\rho+1}+| F(0,0)|)^2/(4\epsilon)$.

\subsection{Validity of Assumption \ref{clc} for VaR-CVaR algorithm in Section \ref{varcvarex}}\label{ass3proofvarcvar}
We aim to show Assumption \ref{clc} is valid for $H_{w_1}$. To achieve this, it is enough to prove
\begin{enumerate}[(1)]
\item The inequality $|\hat{g}_{w_1}(w,X) -  \hat{g}_{w_1}(\bar{w},X)|\leq 2|w_1 - \bar{w}_1|\sum_i|X_i|$ holds, and
\item the last inequality in \eqref{inddiffappendix} is satisfied.
\end{enumerate}
To prove $|\hat{g}_{w_1}(w,X) -  \hat{g}_{w_1}(\bar{w},X)|\leq 2|w_1 - \bar{w}_1|\sum_i|X_i|$, recall that for every $j = 1, \dots, n$, $i \neq j$,
\[
\frac{\partial g_j(w)}{\partial w_j} = \frac{e^{w_j}(\sum_{l \neq j} e^{w_l})}{(\sum_{l=1}^{n}e^{w_l})^2}, \quad \frac{\partial g_i(w)}{\partial w_j} = -\frac{e^{w_i}e^{w_j}}{(\sum_{l=1}^{n}e^{w_l})^2}.
\]
Then, one calculates
\begin{align*}
&|\hat{g}_{w_1}(w,X) -  \hat{g}_{w_1}(\bar{w},X)|\\
& =  \left|\sum_{i=1}^{n}\frac{\partial g_i(w)}{\partial w_1}X_i - \sum_{i=1}^{n}\frac{\partial g_i(\bar{w})}{\partial w_1}X_i  \right|\\
&\leq \left|\frac{\partial g_1(w)}{\partial w_1} - \frac{\partial g_1(\bar{w})}{\partial w_1} \right||X_1| +\left|\sum_{i\neq1} \frac{\partial g_i(w)}{\partial w_1}X_i - \sum_{i\neq1}\frac{\partial g_i(\bar{w})}{\partial w_1}X_i  \right|\\
& \leq \left| \frac{e^{w_1}(\sum_{l \neq 1} e^{w_l})}{(\sum_{l=1}^{n}e^{w_l})^2} -  \frac{e^{\bar{w}_1}(\sum_{l \neq 1} e^{w_l})}{(\sum_{l\neq 1} e^{w_l}+e^{\bar{w}_1})^2} \right||X_1| +\sum_{i\neq1} \left|\frac{e^{w_i}e^{\bar{w}_1}}{(\sum_{l\neq 1} e^{w_l} +e^{\bar{w}_1})^2} - \frac{e^{w_i}e^{w_1}}{(\sum_{l=1}^{n}e^{w_l})^2}  \right||X_i|\\
& = \frac{\sum_{l \neq 1} e^{w_l}}{(\sum_{l=1}^{n}e^{w_l})^2(\sum_{l\neq 1} e^{w_l}+e^{\bar{w}_1})^2}\left|\left(\sum_{l \neq 1} e^{w_l}\right)^2\left(e^{w_1} - e^{\bar{w}_1}\right) +e^{\bar{w}_1}e^{w_1} \left(e^{\bar{w}_1}- e^{w_1} \right)\right||X_1|\\
&\quad +\sum_{i\neq1} \frac{e^{w_i}}{(\sum_{l=1}^{n}e^{w_l})^2(\sum_{l\neq 1} e^{w_l} +e^{\bar{w}_1})^2}\left|\left(\sum_{l \neq 1} e^{w_l}\right)^2\left( e^{\bar{w}_1} -e^{w_1} \right)+ e^{\bar{w}_1}e^{w_1} \left( e^{w_1}-e^{\bar{w}_1} \right)\right||X_i|\\
&\leq 2 |w_1 - \bar{w}_1|\sum_{i =1}^n|X_i|,
\end{align*}
where the last inequality holds due to $1-e^{-x} \leq x$ for all $x \geq 0$.

To prove the last inequality in \eqref{inddiffappendix} is satisfied, we assume without loss of generality $g_n(w) = \max\{g_2(w), \dots, g_n(w)\}$. Then,
\begin{enumerate}
\item[(i)] For $\bar{w}_1\geq w_1$, one calculates
\begin{align}\label{inddiffn}
&\mathbb{E}\left[\sum_i|X_i|\left|\mathbbm{1}_{\{\sum_{i=1}^{n}g_i(w)X_i \geq \theta\}} -  \mathbbm{1}_{\{\sum_{i = 1}^n g_i(\bar{w})X_i\geq \theta\}}\right|\right] \leq  I_1 +I_2,
\end{align}
where
\begin{align*}
I_1 & =  \mathbb{E}\left[\sum_i|X_i|\left|\mathbbm{1}_{\{\sum_{i=1}^{n}g_i(w)X_i \geq \theta\}} -  \mathbbm{1}_{\{\sum_{l\neq 1}g_l(w)X_l+ g_1(\bar{w})X_1\geq \theta\}}\right|\right],\\
I_2 & = \mathbb{E}\left[\sum_i|X_i|\left| \mathbbm{1}_{\{\sum_{l\neq 1}g_l(w)X_l+ g_1(\bar{w})X_1\geq \theta\}} -  \mathbbm{1}_{\{\sum_{l\neq 1,2}g_l(w)X_l+ g_1(\bar{w})X_1+g_2(\bar{w})X_2\geq \theta\}}\right|\right]\\
&\quad +\cdots\\
&\quad+\mathbb{E}\left[\sum_i|X_i|\left| \mathbbm{1}_{\{g_n(w)X_n+\sum_{l\neq n}g_l(\bar{w})X_l\geq \theta\}} -  \mathbbm{1}_{\{\sum_{i = 1}^ng_i(\bar{w})X_i\geq \theta\}}\right|\right].
\end{align*}
To estimate $I_1$, one writes
\begin{align*}
I_1 &\leq \mathbb{E}\left[\sum_i|X_i|\mathbbm{1}_{\{ (\theta
 -\sum_{l\neq n} g_l(w)X_l )/g_n(w) \leq X_n \leq (\theta-g_1(\bar{w})X_1  -\sum_{l\neq 1,n} g_l(w)X_l)/g_n(w)\}}\right] \\
&\quad+\mathbb{E}\left[\sum_i|X_i|\mathbbm{1}_{\{ (\theta-g_1(\bar{w})X_1  -\sum_{l\neq 1,n} g_l(w)X_l)/g_n(w)\}\leq X_n \leq (\theta -\sum_{l\neq n} g_l(w)X_l )/g_n(w)  \}}\right].
\end{align*}
The first term on the RHS of the inequality above can be further estimated as
\begin{align*}
& \mathbb{E}\left[\sum_i|X_i|\mathbbm{1}_{\{ (\theta
 -\sum_{l\neq n} g_l(w)X_l )/g_n(w) \leq X_n \leq (\theta-g_1(\bar{w})X_1  -\sum_{l\neq 1,n} g_l(w)X_l)/g_n(w)\}}\right]\\
& = \mathbb{E}\left[\sum_{i\neq n}|X_i|\mathbb{E}\left[\left.\mathbbm{1}_{\{ (\theta
 -\sum_{l\neq n} g_l(w)X_l )/g_n(w) \leq X_n \leq (\theta-g_1(\bar{w})X_1  -\sum_{l\neq 1,n} g_l(w)X_l)/g_n(w)\}}\right| X_1, \dots, X_{n-1}\right]\right] \\
 &\quad +\mathbb{E}\left[\mathbb{E}\left[\left.|X_n|\mathbbm{1}_{\{ (\theta
 -\sum_{l\neq n} g_l(w)X_l )/g_n(w) \leq X_n \leq (\theta-g_1(\bar{w})X_1  -\sum_{l\neq 1,n} g_l(w)X_l)/g_n(w)\}}\right| X_1, \dots, X_{n-1}\right]\right] \\
 & = \int_{-\infty}^{\infty}\sum_{i \neq n}|x_i|\cdots \int_{-\infty}^{\infty}\int_{(\theta
 -\sum_{l\neq n} g_l(w)x_l )/g_n(w) }^{ (\theta-g_1(\bar{w})x_1  -\sum_{l\neq 1,n} g_l(w)x_l)/g_n(w)}f_{X_n}(z)\, dz\\
 &\hspace{15em} \times f_{X_{n-1}}(x_{n-1})\,dx_{n-1}\cdots f_{X_1}(x_1)\, dx_{1}\\
 &\quad +\int_{-\infty}^{\infty} \cdots \int_{-\infty}^{\infty}\int_{(\theta
 -\sum_{l\neq n} g_l(w)x_l )/g_n(w) }^{ (\theta-g_1(\bar{w})x_1  -\sum_{l\neq 1,n} g_l(w)x_l)/g_n(w)}|x_n|f_{X_n}(z)\, dz\\
 &\hspace{15em} \times f_{X_{n-1}}(x_{n-1})\,dx_{n-1}\cdots f_{X_1}(x_1)\, dx_{1}\\
 &\leq \frac{c_{X_n}(c_{\bar{X}}+(n-2)c_X^2) }{g_n(w)}|g_1(w) - g_1(\bar{w})|+\frac{\bar{c}_{X_n}c_X }{g_n(w)}|g_1(w) - g_1(\bar{w})|\\
 & = (c_{X_n}(c_{\bar{X}}+(n-2)c_X^2)+\bar{c}_{X_n}c_X)  \frac{\sum_i e^{w_i}}{e^{w_n}}\frac{\left(\sum_{i \neq 1}e^{w_i}\right)|e^{\bar{w}_1}-e^{w_1}|}{\left(\sum_i e^{w_i}\right)\left(e^{\bar{w}_1}+\sum_{i \neq 1}e^{w_i}\right)}\\
 &\leq (c_{X_n}(c_{\bar{X}}+(n-2)c_X^2)+\bar{c}_{X_n}c_X) \frac{\sum_{i \neq 1} g_i(w)}{g_n(w)}\frac{e^{\bar{w}_1}}{\left(e^{\bar{w}_1}+\sum_{i \neq 1}e^{w_i}\right)}|\bar{w}_1 -w_1|\\
& \leq  (c_{X_n}(c_{\bar{X}}+(n-2)c_X^2)+\bar{c}_{X_n}c_X) (n-1)\frac{e^{\bar{w}_1}}{\left(e^{\bar{w}_1}+\sum_{i \neq 1}e^{w_i}\right)}|\bar{w}_1 -w_1|\\
&\leq    (c_{X_n}(c_{\bar{X}}+(n-2)c_X^2)+\bar{c}_{X_n}c_X) (n-1)|\bar{w}_1 -w_1|,
\end{align*}
where $c_{\bar{X}}$ denotes the second absolute moment of $X_i$'s , $c_{X_n}$ is the upper bound of the density of $X_n$, and we use $1- e^{-x} \leq x$ for $x\geq 0$ in the third inequality. Moreover, $I_2$ can be upper bounded by
\begin{align*}
I_2 &\leq \mathbb{E}\left[\sum_i|X_i|\mathbbm{1}_{\{ (\theta-\sum_{l\neq 1} g_l(w)X_l)/g_1(\bar{w}) \leq X_1 \leq (\theta-\sum_{l\neq 1,2} g_l(w)X_l  - g_2(\bar{w})X_2)/g_1(\bar{w}) \}}\right] \\
&\quad+\mathbb{E}\left[\sum_i|X_i|\mathbbm{1}_{\{ (\theta-\sum_{l\neq 1,2} g_l(w)X_l  - g_2(\bar{w})X_2)/g_1(\bar{w})\leq X_1 \leq  (\theta-\sum_{l\neq 1} g_l(w)X_l)/g_1(\bar{w})\}}\right]\\
&\quad+ \cdots\\
&\quad+\mathbb{E}\left[\sum_i|X_i|\mathbbm{1}_{\{ (\theta-g_n(w)X_n -\sum_{l\neq 1,n} g_l(\bar{w})X_l)/ g_1(\bar{w})\leq X_1\leq (\theta-\sum_{l\neq 1} g_l(\bar{w})X_l)/ g_1(\bar{w}) \}}\right] \\
&\quad+\mathbb{E}\left[\sum_i|X_i|\mathbbm{1}_{\{  (\theta-\sum_{l\neq 1} g_l(\bar{w})X_l)/ g_1(\bar{w})\leq  X_1 \leq  (\theta-g_n(w)X_n -\sum_{l\neq 1,n} g_l(\bar{w})X_l)/ g_1(\bar{w}) \}}\right].
\end{align*}
The first term on the RHS of the inequality above can be calculated as
\begin{align*}
& \mathbb{E}\left[\sum_i|X_i|\mathbbm{1}_{\{ (\theta-\sum_{l\neq 1} g_l(w)X_l)/g_1(\bar{w}) \leq X_1 \leq (\theta-\sum_{l\neq 1,2} g_l(w)X_l  - g_2(\bar{w})X_2)/g_1(\bar{w}) \}}\right]\\
& = \mathbb{E}\left[\mathbb{E}\left[\left.|X_1|\mathbbm{1}_{\{ (\theta-\sum_{l\neq 1} g_l(w)X_l)/g_1(\bar{w}) \leq X_1 \leq (\theta-\sum_{l\neq 1,2} g_l(w)X_l  - g_2(\bar{w})X_2)/g_1(\bar{w}) \}}\right| X_2, \dots, X_{n}\right]\right] \\
&\quad +\mathbb{E}\left[\sum_{i \neq 1}|X_i|\mathbb{E}\left[\left.\mathbbm{1}_{\{ (\theta-\sum_{l\neq 1} g_l(w)X_l)/g_1(\bar{w}) \leq X_1 \leq (\theta-\sum_{l\neq 1,2} g_l(w)X_l  - g_2(\bar{w})X_2)/g_1(\bar{w}) \}}\right| X_2, \dots, X_{n}\right]\right] \\
 & = \int_{-\infty}^{\infty}\cdots \int_{-\infty}^{\infty}\int_{(\theta-\sum_{l\neq 1} g_l(w)x_l)/g_1(\bar{w})  }^{(\theta-\sum_{l\neq 1,2} g_l(w)x_l  - g_2(\bar{w})x_2)/g_1(\bar{w}) }  |z|f_{X_1}(z)\, dz f_{X_{n}}(x_{n})\,dx_{n}\cdots f_{X_2}(x_2)\, dx_2\\
 &\quad +\int_{-\infty}^{\infty}\sum_{i \neq 1}|x_i|\cdots \int_{-\infty}^{\infty}\int_{(\theta-\sum_{l\neq 1} g_l(w)x_l)/g_1(\bar{w})  }^{(\theta-\sum_{l\neq 1,2} g_l(w)x_l  - g_2(\bar{w})x_2)/g_1(\bar{w}) }  f_{X_1}(z)\, dz\\
 &\hspace{15em} \times f_{X_{n}}(x_{n})\,dx_{n}\cdots f_{X_2}(x_2)\, dx_2\\
 &\leq \frac{\bar{c}_{X_1}c_X}{g_1(\bar{w})}|g_2(w) - g_2(\bar{w})|+\frac{c_{X_1}(c_{\bar{X}}+(n-2)c_X^2)}{g_1(\bar{w})}|g_2(w) - g_2(\bar{w})|\\
 & = (\bar{c}_{X_1}c_X+c_{X_1}(c_{\bar{X}}+(n-2)c_X^2)) \frac{\left(e^{\bar{w}_1}+\sum_{i \neq 1}e^{w_i}\right)}{e^{\bar{w}_1}}\frac{ e^{w_2}|e^{\bar{w}_1}-e^{w_1}|}{\left(\sum_i e^{w_i}\right)\left(e^{\bar{w}_1}+\sum_{i \neq 1}e^{w_i}\right)}\\
& \leq   (\bar{c}_{X_1}c_X+c_{X_1}(c_{\bar{X}}+(n-2)c_X^2))\frac{e^{w_2}e^{\bar{w}_1}}{e^{\bar{w}_1}\left(\sum_i e^{w_i}\right)}|\bar{w}_1 -w_1|\\
&\leq   (\bar{c}_{X_1}c_X+c_{X_1}(c_{\bar{X}}+(n-2)c_X^2))|\bar{w}_1 -w_1|,
\end{align*}
where $c_X$ denotes the first absolute moment of $X_i$'s and $\bar{c}_{X_1}$ is the upper bound of the function $|x|f_{X_1}$. Thus, in the case $\bar{w}_1\geq w_1$, \eqref{inddiffn} becomes
\begin{align*}
&\mathbb{E}\left[\sum_i|X_i|\left|\mathbbm{1}_{\{\sum_{i=1}^{n}g_i(w)X_i \geq \theta\}} -  \mathbbm{1}_{\{\sum_{i = 1}^n g_i(\bar{w})X_i\geq \theta\}}\right|\right] \\
&\leq  2(n-1)((c_{X_n}+c_{X_1})(c_{\bar{X}}+(n-2)c_X^2)+c_X(\bar{c}_{X_n}+\bar{c}_{X_1}))|\bar{w}_1 -w_1|.
\end{align*}

\item[(ii)] As for the case $w_1>\bar{w}_1$, the calculations are close to the above, however, one considers a different splitting as follows
\begin{align}\label{inddiffn2}
\mathbb{E}\left[\sum_i|X_i|\left|\mathbbm{1}_{\{\sum_{i=1}^{n}g_i(w)X_i \geq \theta\}} -  \mathbbm{1}_{\{\sum_{i = 1}^ng_i(\bar{w})X_i\geq \theta\}}\right|\right] &\leq T_1+T_2,
\end{align}
where
\begin{align*}
T_1 & =  \mathbb{E}\left[\sum_i|X_i|\left|\mathbbm{1}_{\{\sum_{i=1}^{n}g_i(w)X_i \geq \theta\}} -  \mathbbm{1}_{\{\sum_{l\neq n}g_l(w)X_l+ g_n(\bar{w})X_n\geq \theta\}}\right|\right]\\
%&\quad +\mathbb{E}\left[\sum_i|X_i|\left| \mathbbm{1}_{\{\sum_{l\neq n}g_l(w)X_l+ g_n(\bar{w})X_n\geq \theta\}} -  \mathbbm{1}_{\{\sum_{l\neq n, n-1}g_l(w)X_l+ g_{n-1}(\bar{w})X_{n-1}+g_n(\bar{w})X_n\geq \theta\}}\right|\right]\\
&\quad +\cdots\\
&\quad+\mathbb{E}\left[\sum_i|X_i|\left| \mathbbm{1}_{\{g_1(w)X_1+g_2(w)X_2+ \sum_{l\neq 1,2}g_l(\bar{w})X_l\geq \theta\}} -   \mathbbm{1}_{\{g_1(w)X_1+ \sum_{l\neq 1}g_l(\bar{w})X_l\geq \theta\}} \right|\right],\\
T_2& = \mathbb{E}\left[\sum_i|X_i|\left| \mathbbm{1}_{\{g_1(w)X_1+ \sum_{l\neq 1}g_l(\bar{w})X_l\geq \theta\}} -  \mathbbm{1}_{\{\sum_{i = 1}^ng_i(\bar{w})X_i\geq \theta\}}\right|\right].
\end{align*}
To estimate $T_1$, one calculates
\begin{align*}
T_1  &\leq \mathbb{E}\left[\sum_i|X_i|\mathbbm{1}_{\{   (\theta-\sum_{l\neq 1} g_l(w)X_l)/ g_1(w)\leq  X_1 \leq  (\theta-g_n(\bar{w})X_n -\sum_{l\neq 1,n} g_l(w)X_l)/ g_1(w)  \}}\right]\\
&\quad+\mathbb{E}\left[\sum_i|X_i|\mathbbm{1}_{\{ (\theta-g_n(\bar{w})X_n -\sum_{l\neq 1,n} g_l(w)X_l)/ g_1(w)\leq X_1\leq (\theta-\sum_{l\neq 1} g_l(w)X_l)/ g_1(w) \}}\right]\\
%&\quad +  \mathbb{E}\left[\sum_i|X_i|\mathbbm{1}_{\{   (\theta-\sum_{l\neq 1,n} g_l(w)X_l -g_n(\bar{w})X_n)/ g_1(w)\leq  X_1 \leq  (\theta-g_n(\bar{w})X_n -g_{n-1}(\bar{w})X_{n-1} -\sum_{l\neq 1,n,n-1} g_l(w)X_l)/ g_1(w)  \}}\right]\\
%&\quad+\mathbb{E}\left[\sum_i|X_i|\mathbbm{1}_{\{  (\theta-g_n(\bar{w})X_n -g_{n-1}(\bar{w})X_{n-1} -\sum_{l\neq 1,n,n-1} g_l(w)X_l)/ g_1(w)  \leq X_1\leq (\theta-\sum_{l\neq 1,n} g_l(w)X_l -g_n(\bar{w})X_n)/ g_1(w)\}}\right] \\
&\quad+ \cdots\\
&\quad +  \mathbb{E}\left[\sum_i|X_i|\mathbbm{1}_{\{   (\theta-\sum_{l\neq 1,2}g_l(\bar{w})X_l-g_2(w)X_2)/ g_1(w)\leq  X_1 \leq  (\theta-\sum_{l\neq 1}g_l(\bar{w})X_l)/ g_1(w)  \}}\right]\\
&\quad+\mathbb{E}\left[\sum_i|X_i|\mathbbm{1}_{\{ (\theta-\sum_{l\neq 1}g_l(\bar{w})X_l)/ g_1(w)  \leq X_1\leq  (\theta-\sum_{l\neq 1,2}g_l(\bar{w})X_l-g_2(w)X_2)/ g_1(w)\}}\right].
\end{align*}
The first term on the RHS of the inequality above can be further calculated as
\begin{align*}
&\mathbb{E}\left[\sum_i|X_i|\mathbbm{1}_{\{   (\theta-\sum_{l\neq 1} g_l(w)X_l)/ g_1(w)\leq  X_1 \leq  (\theta-g_n(\bar{w})X_n -\sum_{l\neq 1,n} g_l(w)X_l)/ g_1(w)  \}}\right]\\
& = \mathbb{E}\left[\mathbb{E}\left[\left.|X_1|\mathbbm{1}_{\{   (\theta-\sum_{l\neq 1} g_l(w)X_l)/ g_1(w)\leq  X_1 \leq  (\theta-g_n(\bar{w})X_n -\sum_{l\neq 1,n} g_l(w)X_l)/ g_1(w)  \}}\right| X_2, \dots, X_{n}\right]\right] \\
&\quad +\mathbb{E}\left[\sum_{i \neq 1} |X_i|\mathbb{E}\left[\left.\mathbbm{1}_{\{   (\theta-\sum_{l\neq 1} g_l(w)X_l)/ g_1(w)\leq  X_1 \leq  (\theta-g_n(\bar{w})X_n -\sum_{l\neq 1,n} g_l(w)X_l)/ g_1(w)  \}}\right| X_2, \dots, X_{n}\right]\right] \\
 & = \int_{-\infty}^{\infty}\cdots \int_{-\infty}^{\infty}\int_{  (\theta-\sum_{l\neq 1} g_l(w)x_l)/ g_1(w)  }^{(\theta-g_n(\bar{w})x_n -\sum_{l\neq 1,n} g_l(w)x_l)/ g_1(w) }  |z|f_{X_1}(z)\, dz f_{X_{n}}(x_{n})\,dx_{n}\cdots f_{X_2}(x_2)\, dx_2\\
 &\quad + \int_{-\infty}^{\infty}\sum_{i \neq 1}|x_i|\cdots \int_{-\infty}^{\infty}\int_{  (\theta-\sum_{l\neq 1} g_l(w)x_l)/ g_1(w)  }^{(\theta-g_n(\bar{w})x_n -\sum_{l\neq 1,n} g_l(w)x_l)/ g_1(w) }  f_{X_1}(z)\, dz \\
 &\hspace{15em} \times f_{X_{n}}(x_{n})\,dx_{n}\cdots f_{X_2}(x_2)\, dx_2\\
 &\leq \frac{\bar{c}_{X_1}c_X}{g_1(w)}|g_n(w) - g_n(\bar{w})| +\frac{c_{X_1}(c_{\bar{X}}+(n-2)c_X^2)}{g_1(w)}|g_n(w) - g_n(\bar{w})|\\
 & = (\bar{c}_{X_1}c_X +c_{X_1}(c_{\bar{X}}+(n-2)c_X^2))\frac{\sum_i e^{w_i}}{e^{w_1}}\frac{ e^{w_n}|e^{w_1}-e^{\bar{w}_1}|}{\left(\sum_i e^{w_i}\right)\left(e^{\bar{w}_1}+\sum_{i \neq 1}e^{w_i}\right)}\\
& \leq  (\bar{c}_{X_1}c_X +c_{X_1}(c_{\bar{X}}+(n-2)c_X^2)\frac{e^{w_n}e^{w_1}}{e^{w_1}\left(e^{\bar{w}_1}+\sum_{i \neq 1}e^{w_i}\right)}| w_1 - \bar{w}_1|\\
&\leq    (\bar{c}_{X_1}c_X +c_{X_1}(c_{\bar{X}}+(n-2)c_X^2)| w_1 - \bar{w}_1|.
\end{align*}
In addition, $T_2$ can be estimated as
\begin{align*}
T_2 & \leq \mathbb{E}\left[\sum_i|X_i|\mathbbm{1}_{\{ (\theta-g_1(w)X_1 -\sum_{l\neq 1,n} g_l(\bar{w})X_l)/g_n(\bar{w}) \leq X_n \leq (\theta -\sum_{l\neq n} g_l(\bar{w})X_l)/g_n(\bar{w})   \}}\right] \\
&\quad+\mathbb{E}\left[\sum_i|X_i|\mathbbm{1}_{\{ (\theta  -\sum_{l\neq n} g_l(\bar{w})X_l)/g_n(\bar{w}) \leq X_n \leq   (\theta-g_1(w)X_1 -\sum_{l\neq 1,n} g_l(\bar{w})X_l)/g_n(\bar{w})   \}}\right].
\end{align*}
The first term on the RHS of the above inequality can be upper bounded by
\begin{align*}
&\mathbb{E}\left[\sum_i|X_i|\mathbbm{1}_{\{ (\theta-g_1(w)X_1 -\sum_{l\neq 1,n} g_l(\bar{w})X_l)/g_n(\bar{w}) \leq X_n \leq (\theta -\sum_{l\neq n} g_l(\bar{w})X_l)/g_n(\bar{w})   \}}\right]  \\
& = \mathbb{E}\left[\mathbb{E}\left[\left.|X_n|\mathbbm{1}_{\{ (\theta-g_1(w)X_1 -\sum_{l\neq 1,n} g_l(\bar{w})X_l)/g_n(\bar{w}) \leq X_n \leq (\theta -\sum_{l\neq n} g_l(\bar{w})X_l)/g_n(\bar{w})   \}}\right| X_1, \dots, X_{n-1}\right]\right] \\
&\quad +\mathbb{E}\left[\sum_{i \neq n}|X_i|\mathbb{E}\left[\left.\mathbbm{1}_{\{ (\theta-g_1(w)X_1 -\sum_{l\neq 1,n} g_l(\bar{w})X_l)/g_n(\bar{w}) \leq X_n \leq (\theta -\sum_{l\neq n} g_l(\bar{w})X_l)/g_n(\bar{w})   \}}\right| X_1, \dots, X_{n-1}\right]\right]\\
 & = \int_{-\infty}^{\infty}\cdots \int_{-\infty}^{\infty}\int_{ (\theta-g_1(w)x_1 -\sum_{l\neq 1,n} g_l(\bar{w})x_l)/g_n(\bar{w}) }^{ (\theta-\sum_{l\neq n} g_l(\bar{w})x_l)/g_n(\bar{w})  }|x_n|f_{X_n}(z)\, dz \\
 &\hspace{15em} \times f_{X_{n-1}}(x_{n-1})\,dx_{n-1}\cdots f_{X_1}(x_1)\, dx_{1}\\
 &\quad + \int_{-\infty}^{\infty}\sum_{i \neq n}|x_i|\cdots \int_{-\infty}^{\infty}\int_{ (\theta-g_1(w)x_1 -\sum_{l\neq 1,n} g_l(\bar{w})x_l)/g_n(\bar{w}) }^{ (\theta-\sum_{l\neq n} g_l(\bar{w})x_l)/g_n(\bar{w})  }f_{X_n}(z)\, dz\\
 &\hspace{15em} \times f_{X_{n-1}}(x_{n-1})\,dx_{n-1}\cdots f_{X_1}(x_1)\, dx_{1}\\
 &\leq \frac{\bar{c}_{X_n}c_{X}}{g_n(\bar{w})}|g_1(w) - g_1(\bar{w})| +\frac{c_{X_n}(c_{\bar{X}}+(n-2)c_X^2)}{g_n(\bar{w})}|g_1(w) - g_1(\bar{w})|\\
 & = (\bar{c}_{X_n}c_{X} + c_{X_n}(c_{\bar{X}}+(n-2)c_X^2)) \frac{\left(e^{\bar{w}_1}+\sum_{i \neq 1}e^{w_i}\right)}{e^{w_n}}\frac{\left(\sum_{i \neq 1}e^{w_i}\right)|e^{w_1}-e^{\bar{w}_1}|}{\left(\sum_i e^{w_i}\right)\left(e^{\bar{w}_1}+\sum_{i \neq 1}e^{w_i}\right)}\\
& =  (\bar{c}_{X_n}c_{X} + c_{X_n}(c_{\bar{X}}+(n-2)c_X^2))\frac{\sum_{i \neq 1} g_i(w)}{g_n(w)}\frac{e^{w_1}}{\sum_i e^{w_i}}|w_1-\bar{w}_1 |\\
& \leq  (n-1)(\bar{c}_{X_n}c_{X} + c_{X_n}(c_{\bar{X}}+(n-2)c_X^2))\frac{e^{w_1}}{\sum_i e^{w_i}}|w_1-\bar{w}_1 |\\
&\leq   (n-1)(\bar{c}_{X_n}c_{X} + c_{X_n}(c_{\bar{X}}+(n-2)c_X^2))|w_1-\bar{w}_1 |.
\end{align*}
Thus for the case $w_1>\bar{w}_1$, we have
\begin{align*}
&\mathbb{E}\left[\sum_i|X_i|\left|\mathbbm{1}_{\{\sum_{i=1}^{n}g_i(w)X_i \geq \theta\}} -  \mathbbm{1}_{\{\sum_{i = 1}^n g_i(\bar{w})X_i\geq \theta\}}\right|\right]\\
&\leq 2(n-1) (c_{X} (\bar{c}_{X_n}+\bar{c}_{X_1})+(c_{\bar{X}}+(n-2)c_X^2)( c_{X_n}+ c_{X_1}))|w_1-\bar{w}_1 |.
\end{align*}
\end{enumerate}
Combining the two cases, one obtains
\begin{align*}
&\mathbb{E}\left[\sum_i|X_i|\left|\mathbbm{1}_{\{\sum_{i=1}^{n}g_i(w)X_i \geq \theta\}} -  \mathbbm{1}_{\{\sum_{i = 1}^n g_i(\bar{w})X_i\geq \theta\}}\right|\right]\\
&\leq 2(n-1) (c_{X} (\bar{c}_{X_n}+\bar{c}_{X_1})+(c_{\bar{X}}+(n-2)c_X^2)( c_{X_n}+ c_{X_1}))|w_1-\bar{w}_1 |.
\end{align*}

\subsection{Auxiliary results}
\begin{lemma}\label{lem:boundedVariance} Let Assumption \ref{expressionH}, \ref{iid}, \ref{clc} and \ref{assum:dissipativity} hold. For any $t  \in [nT, (n+1)T]$, $n \in \mathbb{N}$ and $k = 1, \dots, K+1$, $K+1 \leq T$, one obtains
\[
\E\left[\left|H(\bar{\theta}^{\lambda}_{nT + k-1},X_{nT + k}) -h(\bar{\theta}^{\lambda}_{nT + k-1})\right|^2\right] \leq e^{-a\lambda nT}\bar{\sigma}_Z\mathbb{E}[V_2(\theta_0)]+\tilde{\sigma}_Z,
\]
where
\begin{align}\label{sigmaZ}
\begin{split}
\bar{\sigma}_Z &=4\mathbb{E}\left[K_{\rho}(X_0)\right]\left(L^2+L_1^2\right)\\
\tilde{\sigma}_Z &=4\mathbb{E}\left[K_{\rho}(X_0)\right]\left(L^2+L_1^2\right)c_1(\lambda_{\max}+a^{-1})+4|h(0)|^2+8L_2^2\mathbb{E}\left[K_{\rho}(X_0)\right]+8\mathbb{E}\left[F_*^2(X_0)\right].
\end{split}
\end{align}
%with $\hat{\sigma}_Z = \E[(1+|X_0|+|\E[X_0]|)^{2\rho}|X_0 - \E[X_0]|^2]$.
\end{lemma}
\begin{proof} %Recall  $\mathcal{H}_t = \mathcal{F}^{\lambda}_{\infty} \vee \mathcal{G}_{\lfloor t \rfloor}$.
One notices that by Remark \ref{growth} and \ref{hlip},
\begin{align*}
&\E\left[\left|H(\bar{\theta}^{\lambda}_{nT + k-1},X_{nT + k}) -h(\bar{\theta}^{\lambda}_{nT + k-1})\right|^2\right] \\
%& =\E\left[\E\left[\left.\left|h(\bar{\zeta}_t^{\lambda,n}) - H(\bar{\zeta}_t^{\lambda,n},X_{nT+k})\right|^2\right|\mathcal{H}_{nT}\right]\right] \\
%&=\E\left[\E\left[\left.\left|\E\left[\left.H(\bar{\zeta}_t^{\lambda,n}, X_{nT+k})\right|\mathcal{H}_{nT}\right] - H(\bar{\zeta}_t^{\lambda,n},X_{nT+k})\right|^2\right|\mathcal{H}_{nT}\right]\right] \\
%&\leq 4\E\left[\E\left[\left.\left| H(\bar{\zeta}_t^{\lambda,n},X_{nT+k})- H(\bar{\zeta}_t^{\lambda,n}, \E\left[\left. X_{nT+k}\right|\mathcal{H}_{nT}\right])\right|^2\right|\mathcal{H}_{nT} \right]\right] \\
%&\leq 4L_2^2\hat{\sigma}_Z\E\left[\left(1+\left|\bar{\zeta}_t^{\lambda,n} \right|\right)^2\right],
&\leq 2 \mathbb{E}\left[\left|h(\bar{\theta}^{\lambda}_{nT + k-1}) \right|^2\right] +2 \mathbb{E} \left[\left|H(\bar{\theta}^{\lambda}_{nT + k-1},X_{nT + k}) \right|^2\right] \\
& \leq 2 \mathbb{E}\left[\left(L \left|\bar{\theta}^{\lambda}_{nT + k-1}\right|+|h(0)|\right)^2\right] +2\mathbb{E}\left[\left((1+|X_{nT+k}|)^{\rho+1}\left(L_1\left|\bar{\theta}^{\lambda}_{nT + k-1}\right|+L_2\right)+F_*(X_{nT+k})\right)^2\right]\\
& \leq 4L^2 \mathbb{E}\left[\left|\bar{\theta}^{\lambda}_{nT + k-1} \right|^2\right] +4|h(0)|^2+4L_1^2\mathbb{E}\left[K_{\rho}(X_0)\right]\mathbb{E}\left[\left|\bar{\theta}^{\lambda}_{nT + k-1} \right|^2\right] +8L_2^2\mathbb{E}\left[K_{\rho}(X_0)\right]+8\mathbb{E}\left[F_*^2(X_0)\right]\\
& \leq 4\mathbb{E}\left[K_{\rho}(X_0)\right]\left(L^2+L_1^2 \right)\left(e^{-a\lambda nT}\mathbb{E}[V_2(\theta_0)] +c_1(\lambda_{\max}+a^{-1})\right) \\
&\quad +4|h(0)|^2+8L_2^2\mathbb{E}\left[K_{\rho}(X_0)\right]+8\mathbb{E}\left[F_*^2(X_0)\right],
\end{align*}
where the last inequality holds due to Lemma \ref{lem:moment_SGLD_2p}. Finally, one obtains
\[
\E\left[\left|h(\bar{\zeta}_t^{\lambda,n}) - H(\bar{\zeta}_t^{\lambda,n},X_{nT+k})\right|^2\right] \leq e^{-a\lambda nT}\bar{\sigma}_Z \E[V_2(\theta_0)]+\tilde{\sigma}_Z ,
\]
where $\bar{\sigma}_Z = 4\mathbb{E}\left[K_{\rho}(X_0)\right]\left(L^2+L_1^2\right)$ and $\tilde{\sigma}_Z =4\mathbb{E}\left[K_{\rho}(X_0)\right]\left(L^2+L_1^2\right)c_1(\lambda_{\max}+a^{-1})+4|h(0)|^2+8L_2^2\mathbb{E}\left[K_{\rho}(X_0)\right]+8\mathbb{E}\left[F_*^2(X_0)\right]$.
\end{proof}
\begin{lemma} \label{onestepest} Let Assumption \ref{expressionH}, \ref{iid} and \ref{assum:dissipativity} hold. For any $t >0$, one obtains
\begin{align*}
\E\left[ \left| \bar{\theta}^{\lambda}_t - \bar{\theta}^{\lambda}_{\floor{t}} \right|^2 \right] \leq \lambda(e^{-a\lambda \floor{t}}\bar{\sigma}_Y \mathbb{E}[V_2(\theta_0)] + \tilde{\sigma}_Y),
\end{align*}
where
\begin{align}\label{sigmaY}
\begin{split}
\bar{\sigma}_Y &=  2\lambda_{\max}L_1^2\mathbb{E}\left[K_{\rho}(X_0)\right]\\
 \tilde{\sigma}_Y & = 2\lambda_{\max}L_1^2\mathbb{E}\left[K_{\rho}(X_0)\right]c_1(\lambda_{\max}+a^{-1})+4\lambda_{\max} L_2^2\mathbb{E}\left[K_{\rho}(X_0)\right]+4\lambda_{\max}\mathbb{E}\left[F_*^2(X_0)\right]+2d\beta^{-1}.
\end{split}
\end{align}
\end{lemma}
\begin{proof} For any $t>0$, one calculates
\begin{align*}
\E\left[\left| \bar{\theta}^{\lambda}_t - \bar{\theta}^{\lambda}_{\floor{t}} \right|^2\right] &= \E\left[\left| -\lambda \int_{\floor{t}}^t H(\bar{\theta}^{\lambda}_{\floor{t}},X_{\ceil{t}}) d s + \sqrt{2\beta^{-1}\lambda} ( \tilde{B}_t^{\lambda}-\tilde{B}_{\floor{t}}^{\lambda}) \right|^2\right]\\
& \leq \lambda^2\mathbb{E}\left[\left((1+|X_{\ceil{t}}|)^{\rho+1}(L_1|\bar{\theta}^{\lambda}_{\floor{t}}|+L_2)+F_*(X_{\ceil{t}})\right)^2\right]+ 2d\lambda\beta^{-1},
\end{align*}
where the inequality holds due to Remark \ref{growth} and by applying Lemma \ref{lem:moment_SGLD_2p}, one obtains
\begin{align*}
\E\left[\left| \bar{\theta}^{\lambda}_t - \bar{\theta}^{\lambda}_{\floor{t}} \right|^2\right] & \leq 2\lambda^2 L_1^2\mathbb{E}\left[K_{\rho}(X_0)\right]\mathbb{E}[|\bar{\theta}^{\lambda}_{\floor{t}}|^2]+4\lambda^2 L_2^2 \mathbb{E}\left[K_{\rho}(X_0)\right]+4\lambda^2\mathbb{E}\left[F_*^2(X_0)\right] +2d\lambda\beta^{-1}\\
&\leq \lambda((1-a\lambda)^{\floor{t}}\bar{\sigma}_Y \mathbb{E}[V_2(\theta_0)] + \tilde{\sigma}_Y),
\end{align*}
where $\bar{\sigma}_Y = 2\lambda_{\max}L_1^2 \mathbb{E}\left[K_{\rho}(X_0)\right]$ and $\tilde{\sigma}_Y = 2\lambda_{\max}L_1^2\mathbb{E}\left[K_{\rho}(X_0)\right]c_1(\lambda_{\max}+a^{-1})+4\lambda_{\max} L_2^2 \mathbb{E}\left[K_{\rho}(X_0)\right]+4\lambda_{\max}\mathbb{E}\left[F_*^2(X_0)\right]+2d\beta^{-1}$.
\end{proof}

\begin{lemma}\label{zt2ndmoment} Let Assumption \ref{expressionH}, \ref{iid}, and \ref{assum:dissipativity} hold. Then, for any $t>0$, one obtains
\[
\E[|Z_t|^2] \leq e^{-at}\E[|\theta_0|^2]+\left(\frac{2d}{a\beta}+\frac{2b}{a}+\frac{\E[K_1^2(X_0)]}{a^2}\right)(1-e^{-at}).
\]
\end{lemma}
\begin{proof} For any $t>0$, by applying It\^o's formula to $e^{at}|Z_t|^2$, one obtains, almost surely
\begin{align*}
de^{at}|Z_t|^2 & = a e^{at}|Z_t|^2dt-2  e^{at}\langle Z_t, h(Z_t)\rangle dt +2 e^{at}\langle Z_t, \sqrt{2 \beta^{-1}}dB_t \rangle +2d \beta^{-1}e^{at}dt.
\end{align*}
Then, integrating both sides and taking expectation yield
\begin{align*}
e^{at}\E[|Z_t|^2] & = \E[|\theta_0|^2]+ a  \int_0^t e^{as}\E[|Z_s|^2]ds-2  \int_0^t e^{as}\E[\langle Z_s, h(Z_s)\rangle] ds +2d \beta^{-1}\int_0^te^{as}ds,
\end{align*}
which implies by using Assumption \ref{assum:dissipativity}
\begin{align*}
e^{at}\E[|Z_t|^2] & = \E[|\theta_0|^2]+ a \int_0^t e^{as}\E[|Z_s|^2]ds-2a  \int_0^t e^{as}\E[|Z_s|^2] ds+2b  \int_0^te^{as}ds\\
&\quad +2 \int_0^te^{as}\E[|Z_s|]\E[K_1(X_0)]ds +2d  \beta^{-1}\int_0^te^{as}ds\\
& \leq \E[|\theta_0|^2]+(2b+\E[K_1^2(X_0)]/a+2d\beta^{-1})(e^{at}-1)/a.
\end{align*}
Finally, one obtains
\[
\E[|Z_t|^2] \leq e^{-at}\E[|\theta_0|^2]+(2b+\E[K_1^2(X_0)]/a+2d\beta^{-1})(1-e^{-at})/a.
\]
\end{proof}

\end{document}